\documentclass[tikz,12pt]{article}
\usepackage{mheckII}

\usepackage[T1]{fontenc}

\usepackage{tikz-feynman}

\usepackage{manfnt}
\usepackage{longtable}

\usepackage{fancybox}

\usepackage{multirow}

\usepackage{blkarray}

\usepackage{tikz} 
\usetikzlibrary{matrix}

\theoremstyle{theorem}
\newtheorem*{thm}{Theorem}

\newtheorem{fact}{Fact}

\newtheorem{corl}{Corollary}
\newtheorem{pro}{Proposition}

\newtheorem{lem}{Lemma}
\newtheorem*{cla}{Claim}

\newtheorem*{que}{Question}
\newtheorem*{queos}{Simpler Question}

\newtheorem*{condx}{Condition $\boldsymbol{(\ast)}$}

\newtheorem*{tam}{Tamed property}

\newtheorem*{stru}{Structural factorization}
\newtheorem*{rigi}{Rigidity property}

\theoremstyle{definition}
\newtheorem{defn}{Definition}
\newtheorem*{nott}{Notation}

\newtheorem{rem}{Remark}

\newtheorem{exe}{Example}
\newtheorem*{baby}{Extremely simple situation}

\def\Dsl{\,\raise.15ex\hbox{/}\mkern-13.5mu D}
\def\dsl{\,\raise.25ex\hbox{/}\mkern-10.5mu \partial}

\title{Swampland Geometry and\\ the Gauge Couplings 
}

\authors{Sergio Cecotti\footnote{e-mail: {\tt cecotti@sissa.it}}\vskip 9pt

\centerline{SISSA, via Bonomea 265, I-34100 Trieste, ITALY}
}

\abstract{The purpose of this paper is two-fold.
First we review in detail the geometric aspects of the swampland program for supersymmetric 
4d effective theories using a new and unifying language we dub
``domestic geometry'', the generalization of special K\"ahler
geometry which does not require the underlying manifold to be K\"ahler or
have a complex structure. All 4d SUGRAs are described
by domestic geometry. As special K\"ahler geometries, domestic geometries carry 
\emph{formal} brane amplitudes: 
when the domestic geometry describes the supersymmetric low-energy limit of a
consistent quantum theory of gravity,   
its formal brane amplitudes have the right properties to be \emph{actual} branes. The main datum of 
the domestic geometry of a 4d SUGRA is
 its gauge coupling, seen as a map
from a manifold which satisfies the geometric Ooguri-Vafa conjectures to the Siegel variety;
to understand the properties of the quantum-consistent gauge couplings
we discuss several novel aspects of such ``Ooguri-Vafa'' manifolds, including their Liouville properties.

Our second goal is to present some novel speculation on the extension of the swampland program to \emph{non}-supersymmetric effective theories of gravity. The idea is that the domestic geometric description of
 the quantum-consistent effective theories extends, possibly with some qualifications,
also to the non-supersymmetric case.   
}

\begin{document}
\maketitle

\tableofcontents

\newpage 

\section{Introduction}

The swampland program \cite{Vaf,OoV}  (for reviews see \cite{Rev1,Rev2}) aims to characterize the sparse subset of effective field theories which arise
as low-energy limits of consistent theories of quantum gravity inside the much larger set of formal theories which ``look'' consistent from 
a low-energy perspective, but cannot be completed to a fully consistent theory of quantum gravity. A model which cannot be consistently completed is said to belong to the \emph{swampland.}

The program has taken the form of a rapidly growing list of conjectural necessary conditions (the ``swampland conjectures'' \cite{Vaf,OoV,Rev1,Rev2}) that all effective theories of quantum gravity should
obey. These conditions are motivated by general physical considerations (in particular the thermodynamics of black holes \cite{bsei}) as well as by lessons drawn from  the large supply of consistent effective theories of gravity which describe the light degrees of freedom in a controlled vacuum of superstring theory (the string lamppost principle (SLP) \cite{slp}). We have full analytic control on the quantum stability of a candidate string  vacuum only when it preserves some supersymmetry,
and we can write a precise effective Lagrangian only when its couplings are protected by a SUSY non-renormalization theorem: 
so all reliable examples at our disposal have \emph{extended} supersymmetry.\footnote{\ This limitation in the available explicit examples may \emph{a priori} lead to some bias in our understanding of quantum gravity.}

The best understood examples are the effective theories of Type IIB compactified to 4d on some Calabi-Yau (CY) 3-fold: these examples played a major role in the development of the swampland program
\cite{OoV}. The resulting low-energy 4d $\cn=2$ supergravities are described
(in the vector sector) by special K\"ahler geometry \cite{cec,stro}.
Special K\"ahler geometry is a very powerful tool to study the quantum consistency
of an $\cn=2$ supergravity; a significant part of the swampland program 
\cite{grimm1,grimm2,grimm3,klemm,grimm4,grimm5,swampIII,sg1,sg2} is dedicated to the detailed analysis of the ``motivic''
special K\"ahler manifolds which describe the moduli geometry of \emph{actual} CY manifolds (or motives).  In our opinion,  the swampland program has reached a rather satisfactory shape in this specific $\cn=2$ set-up and ``motivic''
special K\"ahler geometry must be seen as a successful \emph{model}
of swampland theory. It would be highly desirable to extend this model theory to broader contexts.

Supersymmetric effective theories are better behaved in many ways, and SUSY
 seems to play a crucial role in the consistency of many explicit examples. In some contexts  \cite{susyADS} it is hard to grasp how a model can possibly be consistent  \emph{without} being supersymmetric. 
 
 On the other hand, the real world does not look supersymmetric at low energy, and hence consistent non-supersymmetric effective theories of gravity ought to exist. An important goal of the swampland program is to say something less vague about the \textit{non}-SUSY effective theories of consistent quantum gravity. They are expected to be quite rare and remarkable animals, in a sense 
even more ``magical'' than their SUSY counterparts.
Morally speaking, these theories should enjoy  
``all the good properties'' of SUSY while avoiding its phenomenological drawbacks as the existence of (unobserved) super-partners. We should not expect such ``magical'' theories to share the properties which \emph{generically} hold 
for a non-SUSY field theory, because they are extremely \emph{non-}generic. From this point of view, the highly non-generic value of the real-world effective cosmological constant $\Lambda$ is hardly a surprise.

The author's own prejudice is that there should be a ``more general'' swampland principle, which under appropriate circumstances reduces to (or implies) SUSY, but which continues to make sense in some very restricted non-supersymmetric context where most of the ``good'' facts about SUSY still hold. The reader will find many echoes of the prejudice in this paper.
\medskip

The purpose of this paper is two-fold. First we present a systematic review
of the geometric aspects of the swampland program in the SUSY context
from a novel unifying perspective.
The two key concepts of this approach are:
\begin{itemize}
\item[\textit{(i)}] \emph{Ooguri-Vafa (OV) manifolds}, i.e.\! the Riemannian spaces with the correct properties (according to \cite{Vaf,OoV}) to be the scalars' spaces of a consistent effective theory;
\item[\textit{(ii)}] \emph{domestic geometry}, with its brane amplitudes and generalized entropy functions.
\end{itemize}
Domestic geometry is the direct generalization of special K\"ahler geometry which
does \emph{not} require the underlying manifold to carry a complex structure. 
All 4d supergravities are described by a domestic geometry
 in the same exact way that the vector-multiplet sector of 4d $\cn=2$ SUGRA is described 
by the usual special K\"ahler geometry. In particular, the vectors' couplings
have an universal expression in terms of domestic brane amplitudes.
Domestic geometry looks particularly deep and natural when the 
underlying manifold is OV. Combining the two ingredients \textit{(i)},\textit{(ii)}
one gets \emph{arithmetic} domestic geometry which is the obvious 
domestic generalization of the class of ``motivic'' special K\"ahler geometries which describe the vector-sector of quantum-consistent 4d $\cn=2$ supergravities.
This allows to phrase the swampland conditions for the general SUSY model 
in the same suggestive language used in the $\cn=2$ case \cite{swampIII}:
\begin{quote}
A 4d supergravity is described by a domestic geometry from which 
we compute \emph{formal} brane amplitudes. \textit{If the supergravity is \emph{not} in the swampland,
the formal amplitudes describe \emph{actual} branes.}
\end{quote}

\noindent OV manifolds have several nice properties. In particular, they are Liouvillic for the sub-harmonic functions.
This entails that the rigidity properties of ``motivic'' special K\"ahler geometry (the ``power of holomorphy''  \cite{power})
hold for all OV manifolds even if they do not carry any complex structure.
\medskip

Our second main purpose is to present some novel speculations about the
swampland conditions in the \emph{non-}SUSY case.
The space of light scalars $\mathscr{M}$ is (conjecturally) an Ooguri-Vafa manifold
in \emph{all} quantum-consistent
effective theories of gravity, supersymmetric or not. Arithmetic domestic geometry
is naturally defined on all OV manifolds, and its statements make
perfect sense 
for \emph{non}-SUSY effective Lagrangians.
One is then led to speculate that domestic geometry -- which applies
to all SUSY cases -- may also be relevant to describe vectors' couplings
 in quantum-consistent
\emph{non}-SUSY 4d effective theories. The speculation may be stated
at different levels of precision. 
Domestic geometry may be: \textit{(a)} just qualitatively
 valid, or \textit{(b)} semi-quantitatively correct, or even 
 \textit{(c)} exact. 
Besides its aesthetic geometric appeal, and the evidence from SUSY examples, 
the speculation rests on
some heuristic physical arguments, based on the idea of ``naturalness'' in the IR description,
which suggest that 
at least
version \textit{(b)} should hold.

\subparagraph{Warnings to the reader.} The materials presented in this paper are
of different nature and logical status. There are: \textit{(a)} mathematically rigorous
geometric constructions/results; \textit{(b)} physical statements which aim only
to a physicists' level of rigor; \textit{(c)} widely accepted swampland conjectures (which
are taken as ``facts'' in most of the recent literature); \textit{(d)} newly \emph{proposed}
swampland conjectures which are still open to discussion by the community; \textit{(e)}
 speculative statements which are meant as an invitation to further work
in the hope of re-formulating them as precise swampland conjectures in the near future. 
To make the distinction clear
to the reader,
statements related to \textit{(a)} are labelled \textbf{Theorem}, \textbf{Lemma}, and so.
Most statements in the logical classes \textit{(b)} and \textit{(c)} are dubbed \textbf{Facts}.
New conjectures and speculations are called by their proper name. 

In the paper we often use the phrase: \textsf{``[this statement is] true in all known examples
 of reliable quantum-consistent effective theories of gravity''} or similar ones.  Since there is
 an ongoing debate in the community about which examples should be considered
 ``reliable'', we need to specify the class of examples we have in mind. Our notion of
 \emph{reliability} is the most conservative one, that is, we restrict to examples which are 
  under full analytic control and are pacifically considered ``reliable'' by all experts in the field.

\subparagraph{Organization of the paper.} The rest of this paper is organized as follows. In section 2 we discuss some general properties of 4d effective field theories and some technical aspects of the singularities and asymptotics of the moduli spaces. In section 3  we define the OV manifolds and discuss their basic properties. In section 4 we describe the gauge couplings seen as a map
from the moduli space to the Siegel variety. Section 5 contains the basics of 
domestic geometry; a large part of the section is a detailed review of
(generalized) $tt^*$ geometry which is the model which inspires
all constructions in domestic geometry.  Section 6 describes how 
domestic geometry applies to all 4d SUGRAs and may be used to reformulate the swampland conditions in the SUSY context in a more convenient way. Section 7 describes in more detail domestic geometry,
presenting the math arguments for the existence and uniqueness of the underlying tamed maps, and
rephrasing them as heuristic physical arguments in favour of our speculation that 
arithmetic domestic geometry is relevant for the \emph{non}-SUSY swampland program.
Some additional technicality is confined in the appendices.
 \medskip
 
 \begin{rem} The recent papers \cite{bbb1,bbb2} describe a different uniform geometric approach to supergravity, dubbed \emph{bosonic supergravity.} It would be interesting to understand geometric aspects of the swampland conditions from that perspective.  
 \end{rem}



\section{Generalities on effective theories}\label{s:setup} 

\subsection{Our set-up}
To make the story a bit shorter, in this paper we consider only four-dimensional effective theories. Although methods and ideas apply to much wider contexts, we assume a vanishing cosmological constant, $\Lambda=0$, and focus on effective Lagrangians valid at parametrically small energies. In particular, we make the following two assumptions:
\begin{itemize}
\item[\bf A1] all visible IR gauge degrees of freedom (d.o.f.) are in their Coulomb phase. Locally at generic points in field space we may choose 
 an electro-magnetic duality frame, with respect to which these d.o.f.\! are  described by Abelian gauge vectors $A^a$ ($a=1,\dots, h$) with field strengths $F^a=dA^a$; 

\item[\bf A2] the (exactly) massless fields carry no electric or magnetic charge under the $A^a$'s. 
\end{itemize}  
 The light degrees of freedom then consist of 
a space-time metric $g_{\mu\nu}$, $h$ vector fields $A^a$, and $m$ massless real scalars $\phi^i$, together with
spin-1/2 fermions and possibly spin-3/2 gravitini (only in the SUSY case). In the Einstein frame (and the chosen electro-magnetic frame) the effective Lagrangian takes (locally in field space) the form
\begin{equation}\label{lag}
\mathscr{L}_\text{eff}= -\sqrt{-g}\Big(\frac{1}{2}R+\frac{1}{2}G(\phi)_{ij}\,\partial^\mu \phi^i\partial_\mu \phi^j -
\frac{i}{16\pi}\tau(\phi)_{ab}F^a_+F^b_++\frac{i}{16\pi}\bar\tau(\phi)_{ab}F^a_-F^b_-+\cdots\Big)
\end{equation}
where\footnote{\ We use the shorthand $F^a_\pm F^b_\pm\equiv (F^a_\pm)_{\mu\nu}(F^b_\pm)^{\mu\nu}$.} $F_\pm^a =\tfrac{1}{2}(1-\mp i\,\ast)F^a$ is the (anti-)self-dual part of the field-strength $F^a$. 
\medskip

The swampland program asks for a characterization of the field-dependent couplings
in the Lagrangians $\mathscr{L}_\text{eff}$ which describe low-energy limits  of consistent quantum gravities. 
In practice, one looks for necessary conditions they should satisfy, which 
are usually phrased as sufficient conditions for the model \eqref{lag} to sink in the swampland.
\medskip
  
In this paper we limit ourselves to the geometric part of the swampland program, i.e.\! to the two-derivative couplings
 $G(\phi)_{ij}$ and $\tau(\phi)_{ab}$. 
 \medskip
 
 The characterization of the consistent scalars' kinetic couplings $G(\phi)_{ij}$ takes advantage from their geometrical interpretation \cite{OoV}.
The scalar fields $\phi^i$ are seen as local coordinates on a connected ``manifold'' $\widetilde{\mathscr{M}}$, of dimension $m$,
endowed with the Riemannian metric $G_{ij}\equiv G(\phi)_{ij}$. 
The Riemannian metric is smooth at generic points of $\widetilde{\mathscr{M}}$.

\begin{rem} In a \emph{generic} non-SUSY effective field theory we do not expect 
a non-trivial space $\widetilde{\mathscr{M}}$ of exactly massless scalars since the
flat directions of the potential are usually lifted by quantum effects.
This needs not to apply in the present context, since quantum-consistent theories of gravity are highly non-generic. On the other side, the conclusions of this paper remain valid even if the scalars parametrizing $\widetilde{\mathscr{M}}$
are not exactly massless but only hierarchically lighter than the scale above which
the low-energy description breaks down.   
\end{rem}

 \subsection{The $U$-duality group}
 
A basic datum of the effective field theory is its \emph{gauge group}.
The discrete part of the gauge group is called the $U$-duality group
$\cg$. $\cg$ is a redundancy of the description, so it is tautologically
 an exact symmetry of the full quantum theory not just of its IR sector. 
 More precisely, by $\cg$ we mean the quotient of the full discrete gauge group of the underlying UV complete theory
 which acts faithfully on the light bosonic fields.
  In particular, $\cg$ acts on fluxes of forms of various degrees $k$
 \be\label{pppoo0}
 \frac{1}{2\pi}\oint_{\Sigma_k} F^{(k)}
 \ee
 which are integral by generalized Dirac quantization: 
 the fluxes \eqref{pppoo0} take value in a lattice $\Lambda$ endowed with
 a non-degenerated bilinear form $\Lambda\otimes \Lambda\to \Z$ (the generalized Dirac pairing).
 The completeness conjecture \cite{bsei} states that in a consistent
 theory all fluxes allowed by Dirac quantization are realized by some physical state.
  In 4d the electro-magnetic charges take value in a lattice
 $V_\Z\cong \Z^{2h}$ with a skew-symmetric principal\footnote{\ \emph{Principal} means that the 
 Dirac pairing $V_\Z\otimes V_\Z\to \Z$
 yields an identification $V_\Z\equiv V_\Z^\vee$.} Dirac pairing. The action of $\cg$ on
 the scalar fields and the electro-magnetic fluxes
defines a group embedding\footnote{\ $\mathrm{Iso}(\widetilde{\mathscr{M}})$ is the isometry group of
 $\widetilde{\mathscr{M}}$ with metric $G(\phi)_{ij}$.}
 \be
\cg\hookrightarrow \mathrm{Iso}(\widetilde{\mathscr{M}}\;)\times Sp(2h,\Z)
 \ee
The kernel of the map $\cg\to \mathrm{Iso}(\widetilde{\mathscr{M}}\;)$ is finite;\footnote{\ Indeed the kernel is both compact and discrete.}
since we always work modulo finite groups,\footnote{\ The symbol $\sim$ stands for equivalence up to commensurability \cite{morris} i.e.\! equality up to finite groups.} we shall not distinguish $\cg$ from its
image in $\mathrm{Iso}(\widetilde{\mathscr{M}}\;)$. The map
$\rho\colon
\cg\to Sp(2h,\Z)$ is called the \emph{monodromy representation}, and
its image $\Gamma\subset Sp(2h,\Z)$ the \emph{monodromy group.}
Hence the $U$-duality group $\cg$ is a group extension
\be\label{kkkz12m}
1\to \cg^0\to \cg \to \Gamma\to 1
\ee
that is,
 $\cg \cong \cg^0\rtimes \Gamma$.
A version of the $\pi_1$-conjecture \cite{OoV},
which holds true in all known examples, states that the $U$-duality group is isomorphic to 
a subgroup  
 $\cg\subset GL(k,\Z)$ (for some $k$)
which is generated by $k\times k$ \emph{unipotent} matrices (modulo finite groups). Thus $\cg$ is a subgroup of the $\mathbb{Q}$-algebraic
group $GL(k,\mathbb{Q})$, and in all examples\footnote{\ The statement holds in all SUSY consistent theories
by a basic fact in domestic geometry, see \S.\ref{s:stru1}. The reliable examples are all supersymmetric.} its $\mathbb{Q}$-Zariski closure $\overline{\cg}^{\,\mathbb{Q}}$ is semi-simple. Hence, modulo finite groups,
$\cg\sim \cg^0\times \Gamma$.
 The moduli space is
 \be
 \mathscr{M}\equiv\widetilde{\mathscr{M}}/\cg,
 \ee
 that is, the space of inequivalent (effective) vacua. The covering space $\widetilde{\mathscr{M}}$
 may be assumed to be simply-connected with no loss. Then $\pi_1(\mathscr{M})\sim \cg$.
 
 \subsection{Singularities of moduli spaces}\label{s:cure}
 
 The purpose of this technical sub-section is to argue that
 we may be cavalier with the singularities 
 of the moduli space $\mathscr{M}$,
 and see in which sense
we may work
 \emph{as if} $\mathscr{M}$ was a good ($\equiv$ complete)
 Riemannian manifold. The sub-section may be omitted in a first reading.
 
 \subsubsection{Singularities from the action of $\cg$}\label{sing-action}
 
 We write $d(\cdot,\cdot)\colon \widetilde{\mathscr{M}}\times\widetilde{\mathscr{M}}\to \R_{\geq0}$
for the distance function defined by the metric $G_{ij}$ and replace $ \widetilde{\mathscr{M}}$ 
by its metric completion. 
Fix a non-trivial unipotent element $u\in \cg$; the $\pi_1$-conjecture \cite{OoV,Rev1} says that
the infimum of the  function
\be\label{infimum}
d_u\colon \widetilde{\mathscr{M}}\to \R_{\geq0}\qquad
d_u(x)\overset{\rm def}{=}d(ux,x)
\ee
is zero. Assuming the conjecture, we consider a sequence of regular points $\{x_i\}\subset \widetilde{\mathscr{M}}_\text{reg}$ such that $d_u(x_i)\searrow 0$.
We have two possibilities:
\begin{itemize}
\item[(1)] $x_i$ escapes to infinity, i.e.\! $d(x_1,x_i)\to\infty$. We say that ``$u$ has a fixed point at infinity''. The distance conjecture \cite{OoV} requires an infinite tower of
states to become exponentially light as we approach the fixed point at infinity. Their electro-magnetic charges belong to
$\mathsf{ker}(\rho(u)-1)$;
\item[(2)] $\{x_i\}$ remains inside some finite-radius ball $B(x_1,r)$, and hence contains a subsequence which converges to a finite-distance point $x_\ast \in\widetilde{\mathscr{M}}$ which is fixed by $u$.
\end{itemize} 
In the second case
the action of $\cg$ yields an obstruction to a smooth extension of the Riemannian metric to $x_\ast$:
indeed its isotropy sub-group $\cg_{x_\ast}\subset\cg\subset \mathrm{Iso}(\widetilde{\mathscr{M}})$ is discrete and contains the infinite group of isometries $u^\Z$, whereas the isotropy sub-group of a regular point in Riemannian geometry is always compact. Hence $x_\ast$ is 
 a finite-distance singularity where the curvature blows up. In a quantum-consistent effective theory  with $\cn\geq2$ supersymmetry
 all singularities of the \emph{completed} covering space $\widetilde{\mathscr{M}}$
are of this kind, i.e.\! fixed-points of non-trivial unipotent elements $u\in\cg$, while the moduli space $\mathscr{M}\equiv\widetilde{\mathscr{M}}/\cg$ has, in addition, orbifold quotient singularities.
Thus in the extended SUSY case all finite- and infinite-distance singularities of $\mathscr{M}$ 
are dictated by the action of the $U$-duality group $\cg$ -- they correspond to vacua where
a non-trivial subgroup of $\cg$ remains unbroken -- and
no other ``accidental'' singularity is present. 
The logic beyond the $\pi_1$ and distance conjectures suggests that this is the case for a general 
consistent effective gravity theory:
\emph{all} singular points are fixed by some non-trivial subgroup of the discrete gauge group $\cg$.
%
%
%

We give a closed look to the two kinds of finite-distance singularities.
 
 \subparagraph{Quotient orbifold singularities.} In general the $U$-duality group $\cg$ does not act
 freely on the regular locus $\widetilde{\mathscr{M}}_\text{reg}\subset\widetilde{\mathscr{M}}$. The isotropy sub-group $\cg_{\tilde x}$ of a smooth point $\tilde x\in \widetilde{\mathscr{M}}_\text{reg}$ is finite. The image of such a point in the moduli space $\mathscr{M}$ is then a mild orbifold singularity.
Orbifold points correspond physically to vacua where a finite sub-group $\cg_{\tilde x}\subset\cg$ remains unbroken. In some stringy
 examples such points correspond to vacua
  with a non-Abelian enhancement of the effective continuous gauge symmetry,
i.e.\! loci where our standing assumptions \textbf{A1}, \textbf{A2} break down and $\mathscr{L}_\text{eff}$ ceases to be a complete description of the IR physics. 

When we have an embedding $\iota\colon \cg\to GL(k,\Z)$ we can ``repair'' the orbifold singularities in a cheap way.
Fix an integer $n\geq 3$; let $r_n\colon
GL(k,\Z)\to GL(k,\Z/n\Z)$ be the reduction $\mod n$, and write $\iota_n\equiv r_n\circ \iota$. 
Consider the exact sequence of groups
\be
1\to \cg_n \to \cg \xrightarrow{\,\iota_n\,} GL(k,\Z/n\Z)\to 1.
\ee
It follows from the Minkowski theorem that the matrix group $\cg_n$
is a normal subgroup of $\cg$ of \emph{finite index} which is torsion-free, in facts \emph{neat} \cite{morris}.
This implies that $\cg_n$ acts freely on $\widetilde{\mathscr{M}}_\text{reg}$ so that
 $\mathscr{M}_n\equiv \widetilde{\mathscr{M}}/\cg_n$ is a \emph{finite} Galois cover of $\mathscr{M}$,
with Galois group $\iota_n(\cg)\equiv \cg/\cg_n$, free of orbifold singularities.
The finite-quotient singularities may be cured by replacing $\mathscr{M}$ with 
$\mathscr{M}_n$ and $\cg$ with $\cg_n$: this is the standard strategy in the math literature 
when studying moduli spaces of projective varieties (in particular of Calabi-Yau 3-folds), and we adopt it.
From now on by $\cg$ we always mean a finite-index, neat, normal subgroup of
the actual $U$-duality group. Correspondingly $\mathscr{M}\equiv \widetilde{\mathscr{M}}/\cg$ is free of finite-quotient
singularities.

\subparagraph{Finite-distance curvature singularities.}   $\widetilde{\mathscr{M}}_\text{reg}$ is not geodesically complete in general;
that is, $\widetilde{\mathscr{M}}_\text{reg}$ may contain half-geodesics $\ell(t)$ originating from a smooth point $\phi_0\equiv \ell(0)$
which cannot be continued after some \emph{finite} value of the proper length $t$. Such a {finite-length maximal half-geodesic 
represents a physical transition -- which takes \emph{finite} time and costs \emph{finite} energy per unit volume  --
from our initial configuration $\phi_0$ to a physical situation where the IR description provided by  $\mathscr{L}_\text{eff}$ 
is no longer valid. This finite-time process should be perfectly regular from the viewpoint of the UV complete theory.
What happens is that $\ell(t)$ stops at a vacuum where ``new physics'' becomes relevant in the infra-red:
some additional degrees of freedom of the fundamental UV theory get massless, and 
our IR description breaks down. 
From the viewpoint of the UV fundamental theory, these singularities typically correspond to points where different branches of the
space of vacua meet each other (\emph{transition points}).
Since the process involves a finite energy density,
 one expects that there exists a refined effective Lagrangian $\mathscr{L}_\text{new}$ which includes the relevant
  ``new physics''
and is valid up to  some higher but still finite energy scale. $\mathscr{L}_\text{new}$ allows us to extend the IR description, and hence the physical process described by the half-geodesic $\ell(t)$, beyond
 the domain $\widetilde{\mathscr{M}}_\text{reg}$. This means that family of finite-length half-geodesics in $(\widetilde{\mathscr{M}}_\text{reg}, G_{ij})$ 
 ending at a given singular point $x_\ast$ is associated to
a \emph{finite} number of new states becoming massless; in addition their spins must be $\leq \tfrac{3}{2}$. This is to be contrasted with the case of a fixed point at $\infty$, where
an \emph{infinite} tower of states get light \cite{OoV}.

We expect all finite-distance curvature singularities to correspond to fixed points 
$x_\ast$ under a parabolic subgroup of $\cg_\ast\subset \cg$
(cfr.\! the discussion after eqn.\eqref{infimum}). We have a finite set of vectors $\{q\}\subset V_\Z$
which correspond to the electro-magnetic charges of the finitely-many states which becomes
massless at $x_\ast$. They should be invariant under $\cg_\ast$ (modulo finite groups),
i.e.\! $\rho(u)\,q=q$ for $u\in \cg_\ast$.  
At such a singularity, the metric is continuous
(in appropriate local coordinates) but the curvature invariants blow-up.
If our effective theory has $\cn\geq5$ supersymmetry, there is no matter supermultiplets which may become massless, so
no finite-distance curvature singularities may be present. The same holds for $\cn=3,4$ as a result of a SUSY non-renormalization theorem.\footnote{\ The curvature singularity is proportional to the contribution of the gauge coupling beta-function from the states (with e.m.\!\! charges $\{q\}$) which become massless at $x_\ast$, so it vanishes for $\cn=3,4$ where the only matter supermultiplets are vector-multiplets which yield a zero net contribution to the $\beta$-function.
Geometrically, this non-renormalization theorem corresponds to the fact that, for $\cn\geq3$ SUGRA, $\widetilde{\mathscr{M}}$ is a symmetric space,
whose curvature is parallel, $\nabla_i R_{jklm}=0$, so cannot blow up anywhere.}

\subsubsection{Example: Type IIB on a 3-CY}
The various kinds of singularities are well illustrated by the vector-scalars' moduli space $\mathscr{M}_v$ of Type IIB compactified on a Calabi-Yau 3-fold $X$. The regular locus of its covering space, $\widetilde{\mathscr{M}}_{v,\,\text{reg}}$, may be identified with the
moduli space of marked and polarized Calabi-Yau's in the smooth deformation class of $X$.
$\widetilde{\mathscr{M}}_{v,\,\text{reg}}$ is equipped with its Weil-Petersson (WP) K\"ahler metric $G_{i\bar j}$,
which is the metric appearing in the scalars' kinetic terms \cite{cec,stro}.
Singularities at finite distance in the WP metric correspond to points where there is a conifold transition (more generally, an extremal transition) to a CY with different Hodge numbers \cite{noncom}. According to a celebrated suggestion by Reid \cite{reid}, all CY
moduli spaces are expected to be connected through such transitions.  The full Type IIB string theory remains regular
at those transitions, but the low-energy description based on the effective $\cn=2$ supergravity breaks down. 
On $\widetilde{\mathscr{M}}_{v,\,\text{reg}}$ there is also another, better behaved, canonical K\"ahler metric $K_{i\bar j}$, the Hodge one (a.k.a.\! the $K$-metric \cite{ttA,ttB}),
which has the expression (with $n=\dim_\C\mathscr{M}_v$) \cite{ttA,ttB,guy,cec-insta}
\begin{align}\label{kkaq123}
K_{i\bar j}&= (n+3)G_{i\bar j}+R_{i\bar j}\geq 2 \,G_{i\bar j},\\
 R_{i\bar j}&\geq -(n+1)G_{i\bar j},\label{brrricssi}
\end{align}
where $R_{i\bar j}$ is the Ricci curvature of the WP metric $G_{i\bar j}$. It is important to notice that, while the Ricci curvature\footnote{\ More in detail, each holomorphic sectional curvature of $G_{i\bar j}$ is bounded below by $-2$.}
$R_{i\bar j}$ of the WP metric satisfies the \emph{lower} bound \eqref{brrricssi}, the Ricci curvature of the $K$-metric
$R^K_{i\bar j}$ satisfies an \emph{upper} bound \cite{guy}
\be\label{kasq1z2x}
R^K_{i\bar j}\leq -\frac{1}{(\sqrt{n}+1)^2+1}\,K_{i\bar j}.
\ee
From eqn.\eqref{kkaq123} we see that points at infinite distance in the WP metric are also at infinite distance in the $K$-metric.
The opposite statement is false:
in terms of the $K$-metric 
all finite-distance curvature singularities are pushed at infinite distance. Indeed,
at conifold points the metric $G_{i\bar j}$ remains bounded while the Ricci curvature blows up so that the $K$-metric blows up.
 
The Torelli space is the completion of $\widetilde{\mathscr{M}}_{v,\,\text{reg}}$
with respect to the $K$-metric \cite{torelli1,torelli2,torelli3}
\be
\widetilde{\mathscr{M}}^{\mspace{4mu}K}_v\equiv \big(\widetilde{\mathscr{M}}_{v,\,\text{reg}}\big)^{\begin{smallmatrix}K\text{-metric\phantom{m|}}\\\text{completion}\end{smallmatrix}}.
\ee
$\widetilde{\mathscr{M}}^{\mspace{4mu}K}_v$ is a smooth space diffeomorphic to $\R^{2m}$
\cite{torelli1,torelli2,torelli3}; since $\cg$ is torsion-less, 
 the space
\be
\mathscr{M}^K_v\equiv \widetilde{\mathscr{M}}^{\mspace{4mu}K}_v/\cg
\ee
is a version of a finite cover of the moduli space which is a smooth K\"ahler manifold, complete for the $K$-metric,
with a contractible universal cover. Its fundamental group $\pi_1(\mathscr{M}^K_v)\cong \cg$ satisfies (a refined version of) the $\pi_1$-conjecture \cite{OoV}. Unfortunately the nice manifold $\mathscr{M}^K_v\equiv \widetilde{\mathscr{M}}^{\mspace{4mu}K}_v/\cg$
is not complete for the physical WP metric. 

$\mathscr{M}^K_v$ is the natural space to parametrize the complex structures on a \emph{fixed} smooth-class of CY 3-folds: $\mathscr{M}^K_v$ does not talk to moduli spaces of CY's with different topologies which are pushed infinitely away in the $K$-metric. The physical moduli space of Type IIB compactified on $X$ instead
consists of several branches of vacua, with extremal transitions between them;
since the different branches are not
infinitely separated, the physically relevant metric $G_{i\bar j}$ on each branch cannot be complete.  
\medskip

There is another description of these moduli spaces more in the spirit of Algebraic Geometry. 
There is a compact projective variety $\bar{\mathscr{M}}_v$ and an effective divisor $D$ such that
\be\label{xxxrw}
\mathscr{M}^K_v=\bar{\mathscr{M}}_v\setminus D.
\ee
By Hironaka theorem, we can choose the pair $(\bar{\mathscr{M}}_v, D)$ so that $\bar{\mathscr{M}}_v$ is smooth and
$D=\sum_i D_i$ is a simple normal crossing divisor.  The points in the support of $D$ are at infinite distance in the $K$-metric.
With respect to the WP metric $D$ splits as $D=D_\infty+D_f$, where $D_f$ (resp.\! $D_\infty$) is the singular locus at finite (resp.\! infinite)
distance. 
%

\subsubsection{Smoothing surgery}\label{surgery!}

Returning to the general case, the presence of finite-distance curvature singularities
in $\widetilde{\mathscr{M}}$ seems to be unavoidable when $m\geq2$, unless the low-energy theory is a $\cn\geq3$ supergravity or some truncation thereof.
In all other cases the moduli space cannot be both 
smooth and complete for the kinetic-terms metric $G_{ij}$.
 
Working with non-complete and/or non-smooth Riemannian spaces is technically inconvenient. We try to improve
the situation by smoothing out the singularities with some local surgery,
i.e.\! by modifying the metric $G_{ij}$ in the vicinity of the ``bad'' points.

We stress that the modified metric is meant to be a mere technical trick to simplify the analysis of the geometry
of the moduli space $\mathscr{M}$. 
However it is suggestive to phrase the surgery as it was an actual modification of 
the effective Lagrangian $\mathscr{L}_\text{eff}$. 
The modification would be almost ``harmless'', since the original Lagrangian
itself gave a poor description of the IR physics 
 near the singularity in field space,
so the local modification affects physical processes 
which were already outside the scope of $\mathscr{L}_\text{eff}$.
In a sense the locally modified Lagrangian is still a ``good''
 effective Lagrangian, and should satisfy the swampland consistency conditions
 as far as they do not involve the region near the 
 ``bad''  points in field space.

The allowed deformations of $\mathscr{L}_\text{eff}$ are quite restricted since
they should preserve all gauge symmetries.
In particular: \textbf{(1)} the deformation of the kinetic-terms metric should leave
the discrete gauge group $\cg$ as an \emph{exact} symmetry of the problem;
\textbf{(2)} the holonomy and isometry Lie algebras of $G_{ij}$,
 $\mathfrak{hol}(\mathscr{M})$ and $\mathfrak{iso}(\mathscr{M})$,
  should be preserved.

Under the hypothesis that all singularities are fixed by a subgroup $\cg_\ast$ of the neat group
$\cg$, 
it suffices to modify the metric on the universal cover $\widetilde{\mathscr{M}}$ in such a way
that $\cg$ is a freely acting group of isometries of the deformed metric. 
Since the Ricci curvature is expected to blow up at a fixed point of an unipotent isometry,
a perturbation of the form
\be\label{opp024x}
G_{ij}\to G^\epsilon_{ij}\equiv G_{ij}+\epsilon\, R_{ij},\quad \epsilon>0
\ee
suggests itself.
Except at loci where the curvature blows-up while the metric remains finite -- which is exactly the characterization of
finite-distance singularities --
the correction to the metric is negligible for $\epsilon$ very small, so \eqref{opp024x} is essentially a local modification of the geometry
around the finite-distance singular loci. $\cg$ is still an exact isometry of the perturbed metric, so we do not spoil the discrete gauge symmetry. More generally, the surgery \eqref{opp024x} does not spoil any symmetry the original geometry may have. 
Condition \textbf{(2)} is also satisfied.\footnote{\ Let us check that the holonomy algebra  $\mathfrak{hol}(\mathscr{M})$ is preserved.
We may assume $\widetilde{\mathscr{M}}$ to be irreducible without loss.
If the Riemannian metric $G_{ij}$ is locally symmetric, it is Einstein, and the modification \eqref{opp024x} is just
a change of overall normalization of the metric by a factor $1+O(\epsilon)$. Otherwise $\mathfrak{hol}(\mathscr{M})$
is one of the seven Berger holonomies \cite{besse}. For generic holonomy $\mathfrak{so}(m)$ there is nothing to show. If $G_{ij}$ is K\"ahler, so
is $G^\epsilon_{ij}$. For $\mathscr{M}$ a Calabi-Yau, hyperK\"ahler, quaternionic-K\"ahler, $G_2$- or $Spin(7)$-manifold,
$G_{ij}$ is either Ricci-flat or Einstein so the deformation \eqref{opp024x} is either trivial or a slight rescaling of the metric.} We can see the modification \eqref{opp024x} as the result of a 
 backward Ricci-flow \cite{ricciflow} of the metric  by the small time $t=-\epsilon/2$, so that, heuristically, it looks like
 a RG flow to a slightly larger energy scale, in line with the physical interpretation of the finite-distance singularities
 as loci where new physics comes in.
 In view of eqn.\eqref{kkaq123}, replacing the WP metric $G_{j\bar k}$ by the Hodge one $K_{j\bar k}$ on
 a 3-CY moduli space amounts to the surgery \eqref{opp024x}
 with $\epsilon=1/(m+1)$; indeed our proposed prescription \eqref{opp024x} is modelled
 on the standard math treatment of 3-CY complex moduli spaces.

The modification \eqref{opp024x} makes sense provided  the Ricci curvature is bounded below, $R_{ij}\geq -K G_{ij}$, so that
$G^\epsilon_{ij}$ is positive-definite for small $\epsilon$. The idea is that the modification replaces
a small region around the singular point $x_\ast$ by a cusp of infinite length but finite volume\footnote{\ The basic reason of the finiteness of the volume is that we mod out the infinite group $\cg_\ast$ which maps a small neighborhood of $x_\ast$ into itself.}
of order $O(\epsilon^k)$ for some $k>0$. When this happens, the singularity is pushed at infinite distance, and geodesic completeness is restored. Thus the volume conjecture still holds after the surgery, however the new ``spurious''
cusps are not associated to towers of light states as the genuine infinite ends of $\mathscr{M}$.
  
In  this paper we \emph{assume} that it is always possible to modify the metric locally at the singularities, while preserving $\cg$ and $\mathfrak{hol}(\mathscr{M})$,
 by replacing a neighborhood of the finite-distance singularity with a cusp of volume $O(\epsilon^k)$, so that the resulting Riemannian space is complete and smooth.
All our arguments below are meant to apply to the ``regularized'' moduli manifold so constructed, which we shall denote simply $\mathscr{M}$.  
In the known examples the assumption holds true.

After the modification the geometric swampland conjectures still hold if they were satisfied by the original 
$\mathscr{L}_\text{eff}$, with the only exception that the distance conjecture does not apply to the spurious infinite ends
introduced by blowing-up finite-distance singularities.

%

\subsection{Behavior at infinity}

In some argument below we need some more technical aspect
of the geometry of $\mathscr{M}$ at infinite distance.
In this sub-section we sketch the main issues; it  may be omitted in a quick reading.

\subsubsection{Sign of curvature at infinity}
In their original paper \cite{OoV} Ooguri and Vafa conjectured that
the scalars' space $\mathscr{M}$ is non-compact of finite-volume. They also
conjectured that the moduli-space scalar curvature $R$ is negative at infinity.
In ref.\!\cite{counter} Trenner and Wilson constructed an explicit ``counter-example''\footnote{\ See \textbf{Theorem 3.2} of ref.\cite{counter}.} to the last statement in the context of
Type IIB on a certain 3-CY with $h^{2,1}=3$. In that example there is a real curve $C$ in moduli space, parametrized by $s\in\R$,
such that, as we approach a ``large complex structure limit'' along this curve, the WP scalar curvature $R$ behaves as \cite{counter}
\be
R= \frac{32}{81}\,s^3+O(s^2)\quad \text{as }s\to+\infty\ \text{along }C,
\ee
so in this limit $R$ is positive and unbounded. We wrote \emph{``counter-example''} between quotes
because this example does not contradict the physical picture of \cite{OoV}.
In the language of eqn.\eqref{xxxrw}, the physical intuition for Calabi-Yau moduli spaces goes roughly as follows: as we approach a \emph{generic} point on the divisor $D_\infty\subset \bar{\mathscr{M}}$ (the infinite-distance locus)
the Ricci curvature of the WP metric becomes negative hence bounded by \eqref{brrricssi}, while as we approach
a \emph{generic} point on $D_f$ (finite-distance singularities) $R_{ij}$ becomes positive and \emph{divergent} because of the contributions from
loops of the finitely many additional light particles which can be computed in some ``enlarged'' effective \emph{field} theory. 
What happens at the special points at infinity $D_\infty \cap D_f$?
The obvious guess is that if we approach the intersection point following a curve $C$ along which
\be
m_f\to 0,\quad m_\infty\to 0\quad\text{with}\quad \frac{m_f}{m_\infty}\to 0,
\ee 
where $m_f$ (resp.\! $m_\infty$) is the mass scales of the particles getting light along $D_f$ (resp.\! $D_\infty$),
 then the divergent positive contribution to $R$ from the finitely many massless particles along $D_f$
may win over the bounded negative contribution from the infinite tower of light states along $D_\infty$.
This is what happens in the Trenner-Wilson example; along their curve $C$
$m_f=O(s^{-2})$ and $m_\infty=O(s^{-1})$.
We see that the  non-positivity of $R$ at infinity has the same physical origin as the failure of the WP
metric to be complete \cite{noncom}.
Then after replacing $G_{i\bar j}$ with its regularized version $G^\epsilon_{i\bar j}$, which it 
makes $\mathscr{M}$ into a complete manifold, we expect that also the problem with the sign of
the scalar curvature at infinity is solved, that is, we expect that its scalar curvature $R^\epsilon$
to be negative and bounded (for fixed $\epsilon$) everywhere at infinity. 
Indeed, along $D_\infty$ the $\epsilon$-modification is inessential while along $D_f$ we have
$G^\epsilon_{i\bar j}\approx \epsilon\, K_{i\bar j}$ so that from eqn.\eqref{kasq1z2x} the scalar curvature is
asymptotically negative and large (of order $O(1/\epsilon)$). As a check we have computed (using \textsc{Mathematica})
 the scalar curvature $R^\epsilon$
in the Trenner-Wilson example along the curve $C$:
\be
R^\epsilon= -\frac{6}{\epsilon}\big(1+O(\epsilon^2)\big)+O(1/s),\quad\text{as }s\to+\infty\ \text{along }C.
\ee
In other words: the points on the curve $C$ for sufficiently large (but finite) $s$
are outside the region where the Lagrangian $\mathscr{L}_\text{eff}$
yields a reliable IR description of the physics, and the coupling $G_{i\bar j}$
(i.e.\! the WP metric) needs not to behave in a ``physically reasonable'' way at these
points.
Fix a regular point $p_0\in\mathscr{M}$; for all points at a sufficiently large
distance from $p_0$ which are not too close to special loci where
some ``new physics'' appears, $R$ is negative.

\subsubsection{Cusps}\label{s:cusps}

We first consider the following simple but typical\footnote{\ \emph{Typical} means, in particular, that this is the situation in all known explicit examples.} situation:
$M$ is a complete non-compact Riemannian manifold and
there is a compact subset $K\subset M$ such that
 $M\setminus K$ is the disjoint union of finitely many ``ends at infinity'' of $M$,
the $\alpha$-th end $E_\alpha$ being diffeomorphic to $\R\times Z_\alpha$ for some connected  manifold $Z_\alpha$, while the metric in $E_\alpha$ has the asymptotic form
\begin{equation}\label{rt5167}
ds^2\equiv G_{ij}\,dx^i dx^j\approx dr^2+g(r,u)_{ab}\, du^a du^b,\qquad \text{for }r\gg1,
\end{equation} 
where $r$ is the distance from some base point $\ast\in M$, $u^a$ are local coordinates in $Z_\alpha$, and $g(r,u)_{ab}$ is some $r$-dependent metric on $Z_\alpha$.
Finiteness of the volume of $E_\alpha$ requires $\sqrt{\det g(r,u)}$ to decay more rapidly than
$1/r$ for large $r$. In the region where \eqref{rt5167} holds we have
\be
R_{rr}\approx -\partial_r^2\log\sqrt{\det g}-\frac{1}{4}\big\|\partial_r g\big\|^2
\ee
If $g(r,u)_{ab}$ goes to zero as
slowly as a negative power of $r$, the \textsc{rhs} is $O(1/r^2)$
and $R_{rr}/g_{rr}$ is not bounded away from zero. 
On the other hand if $g(r,u)_{ab}$ goes to zero more rapidly than an exponential,
say as $\approx\! C\, e^{-c\, r^k}$ with $k>1$ the curvature $R_{rr}$ is
$O(r^{2(k-1)})$ and unbounded below for large $r$.
So, if the Ricci tensor is negative and bounded for large $r$, $g(r,u)_{ab}$ should
be a sum of terms with exponential decay
\be\label{jjjasq12x}
g(r,u)_{ab}\approx \sum_i e^{-c_i\,r}\,h^{(i)}\mspace{-2mu}(u)_{ab}.
\ee
Assuming the asymptotic metric to be enough regular, this leads to bounds for large $r$ of the form
\be\label{jasqw1z}
- K_1\, G_{ij}\leq R_{ij}\leq -K_2\, G_{ij}\qquad \text{for large $r$ along $E_\alpha$}
\ee
for some constants $K_1,K_2>0$. We shall call a finite-volume end $E_\alpha$ 
with the behaviour \eqref{jasqw1z} a ``cusp''.
Prototypical examples are the cusps in an arithmetic quotient of 
a non-compact symmetric space\footnote{\ Here $G(\R)$ is a non-compact real Lie group
seen as a concrete group of matrices via a suitable representation of degree $\ell$, $K\subset G(\R)$ is a maximal compact
subgroup, and $G(\Z)\equiv G(\R)\cap GL(\ell,\Z)$ is the arithmetic subgroup consisting of matrices with integral entries.
More generally, we may replace $G(\Z)$ by a commensurable subgroup of $G(\R)$.}
\be\label{kkkkas12x}
G(\Z)\backslash G(\R)/K.
\ee
Eqn.\eqref{kkkkas12x} is the general form of the moduli space $\mathscr{M}$ when $\mathscr{L}_\text{eff}$
has a large supersymmetry (more than 8 supercharges).
In these cases the nice geometry of the ends of the moduli space
is directly related to the physics of quantum gravity as
described by the  distance conjecture \cite{book}. In these extended \textsc{susy}
examples the $U$-duality group $\cg\equiv G(\Z)$ acts faithfully
(modulo finite groups) on
the electro-magnetic charge lattice $V_\Z$,
and each point at infinity $x_\infty\in\overline{\mathscr{M}}$ is fixed by a parabolic
subgroup $\cg_{x_\infty}\subset G(\Z)$.  
States carrying electro-magnetic charges $q\in V_\Z$
 invariant under  $\cg_{x_\infty}$ have masses proportional to the length
 of the image in the Siegel domain of the shortest loop in $\mathscr{M}$
 based at $x$ in the homotopy class of the elements of $\cg_{x_\infty}$;
 since the map $\mu$ is a totally geodesic isometric embedding
 for $\cn\geq3$ the mass is also proportional to the length of the pre-image loop
 in $\mathscr{M}$, which for a good cusp is exponentially small, cfr.\! eqn.\eqref{jjjasq12x}. 
 The same applies (with some subtlety) in the $\cn=2$ case,
 using the relation between the kinetic terms of scalars and vectors
 implied by SUSY which replaces the totally geodesic embedding condition. 
Conjecturally this extends to the general quantum-consistent effective theory: the length of the loop in target space should
be exponentially small to fit with the predictions of the distance conjecture.
If the pre-image loop has a length which vanishes more rapidly than any 
exponential, say $O(e^{-c r^k})$, $k>1$,
 the derivative would be of order $e^{ c r^k}$,
 $\|d\mu\|^2 =O(e^{3 c r^{k}})$ which looks unreasonable.
 
\subparagraph{Condition $\boldsymbol{(\ast)}$.} Although the large $r$ behaviour \eqref{jasqw1z}
is expected for all quantum-consistent effective theories, to be very conservative 
in this paper we shall assume a
much weaker condition on the large $r$ behaviour of the geometry
after the smoothing surgery in \S.\ref{surgery!}.
First we assume that
 the Ricci curvature of $\mathscr{M}$ is still bounded below
$R_{ij}\geq - K G_{ij}$. 
Since $\mathscr{M}$ is complete, for all $R>0$
there exists a Lipschitz continuous function $h_R\colon \mathscr{M}\to \R$ such that for some fixed constant $k > 0$ \cite{ricciflat}:\footnote{\ For instance, we can choose $h_R=\varphi(r/R)$ where $\varphi$ is a smooth function on the real line with $0\leq \varphi\leq 1$, $\varphi=1$ for $x\leq 1$ and $\varphi=2$ for $x\geq2$.} 
\be\label{poiqwert}
0\leq h_R\leq 1,\qquad h_R=\begin{cases} 1 &\text{for }r\leq R\\
0 &\text{for }r\geq 2R,
\end{cases}\qquad \big|dh_R\big|<\frac{k}{R}.
\ee
For typical asymptotic metrics of the form \eqref{rt5167},\eqref{jjjasq12x}
the Laplacian of $h_R$ is of order $O(R^{-1})$ for large $R$.
We shall require only the much weaker

\begin{condx}
The Laplacian of $h_R$ is bounded by a constant $C$ independent of $R$
\be\label{xxxqwrt56}
|\Delta h_R|<C.
\ee
\end{condx}

\section{OV manifolds}

In this section we assume that there exists a suitable local surgery, along the lines
described in the previous section, such that the singularities of a suitable finite cover of
the moduli space get repaired resulting in a \emph{smooth} Riemannian manifold $\mathscr{M}$
which still satisfies the Ooguri-Vafa geometric swampland conjectures \cite{OoV}. This certainly holds in the
known examples of quantum consistent effective theories of gravity.

A smooth manifold which satisfies the Ooguri-Vafa geometric conditions, together with some
mild ``regularity'' conditions, will be called
an \emph{OV manifold}. Understanding the geometry of OV manifolds and its physical implications
 is one of the themes of this paper.
\smallskip

We propose the following definition of OV manifold:
 
\begin{defn}The point is a zero-dimensional OV space. In positive dimension an \textit{OV manifold} is
 a connected, complete, Riemannian manifold $M$ with a smooth simply-connected cover $\widetilde{M}$
 which has a de Rham decomposition of the form 
\be
\widetilde{M}= F\times \widetilde{M}_1\times \cdots\times \widetilde{M}_s
\ee
 and a smooth \textit{finite cover} of the form
 \be\label{finitecover}
  M^\flat = F\times  \widetilde{M}_1/\mathcal{G}_1\times \cdots\times \widetilde{M}_s/\mathcal{G}_s,\qquad
 \mathcal{G}_k\subset\mathrm{Iso}(\widetilde{M}_k)
\ee
such that:
\begin{itemize}
\item[\bf OV0.] the flat factor $F$ is either trivial or the real line $\mathbb{R}$;
\item[\bf OV1.] $M_k\equiv \widetilde{M}_k/\mathcal{G}_k$ is a complete, irreducible, \textit{non-compact} manifold of \textit{finite volume};
\item[\bf OV2.] $\widetilde{M}_k$ is diffeomorphic to $\mathbb{R}^{m_k}$
and $\mathcal{G}_k\cong\pi_1(M_k)$ is a torsion-less discrete group of isometries of $\widetilde{M}_k$
generated by elements $\{u_i\}$ such that
\be\label{kkkcx123}
\underset{x\in \widetilde{M}_k}{\mathrm{inf}}\, d(u_i x,x) =0.
\ee
\item[\bf OV3.] the Ricci curvature $R_{ij}^{(k)}$ of $M_k$ is bounded below by a negative constant
\be\label{eeeeecq}
R_{ij}^{(k)}\geq - K_k\,g_{ij}^{(k)},\qquad K_k>0,
\ee
and condition $\boldsymbol{(\ast)}$ (eqn.\eqref{xxxqwrt56}) is satisfied.
\end{itemize}
\end{defn}

Some comments on the definition are in order:
\begin{itemize}
\item \textbf{OV1} is the volume conjecture \cite{OoV} and \textbf{OV2} is a refined version of the $\pi_1$ conjecture. A stronger version of \textbf{OV2}, which holds in all known examples, would be:
\begin{itemize}
\item[\bf OV2$\boldsymbol{\ast}$] $\cg_k$ is 
 isomorphic to a strongly approximant\footnote{\ We recall the definition of ``strong approximant'' subgroup. We see the arithmetic group $G(\Z)$ as a concrete group of integral matrices, i.e.
it comes with a preferred embedding $G(\Z)\subset GL(n,\Z)$ for some $n$. For $p$ a prime, we write $G(\Z/p\Z)$ for the finite group of Lie type obtained by reducing the matrices mod $p$. We have the canonical surjection $G(\Z)\xrightarrow{\pi_p} G(\Z/p\Z)$. $\gamma\colon \Gamma\hookrightarrow G(\Z)$ is said to be a \emph{strong approximant} iff the group homomorphism $\pi_p\circ \gamma\colon \Gamma\to G(\Z/p\Z)$ is surjective for \emph{almost all} primes $p$.} subgroup $\mathring{\cg}_k$ of an arithmetic group $G(\Z)\subset GL(n,\Z)$. $\mathring{\cg}_k$ is required to be neat, semi-simple, and to have a finite-index subgroup generated by finitely-many unipotent elements $\{u_i\}\subset \mathring{\cg}_k\cong \cg_k$ which satisfy eqn.\eqref{kkkcx123}.
\end{itemize}
\item the point and the real line $\mathbb{R}$ are OV manifolds. This is required by math elegance and is consistent since 
all the `magic' properties of OV manifolds 
are shared by the point and $\mathbb{R}$.  The point and $\mathbb{R}$ do appear as moduli of consistent gravities,
and even as their factor spaces: think of M-theory on $\mathbb{R}^{10-k,1}\times S^k$ for $k=0,1,2$; in particular 
for $k=2$ $\mathscr{M}= \mathbb{R}\times  (SL(2,\mathbb{Z})\backslash SL(2,\mathbb{R})/U(1))$.
\item \textbf{OV3} is a milder version of the ``regularity'' conditions discussed in \S.\,\ref{s:cusps} related to the distance conjecture and the expected behaviour at infinity.
\end{itemize}

\begin{rem} 
 The \emph{emergence proposal} (\S.5 of \cite{Rev2})
gives a model-independent bridge between the distance conjecture and the
geometry of $\mathscr{M}$ at infinite distance. 
One starts with a compactification $\overline{\mathscr{M}}$ of $\mathscr{M}$;
the
 metric is singular along the loci in $\overline{\mathscr{M}}$ where more
 degrees of freedom get massless: the singular contributions to the metric come
from loops of light states. When an \emph{infinite} tower of
d.o.f.\! get massless the divergence of the metric is so severe that the locus is pushed at infinite distance.
The
singular part of the metric due to a infinite tower of light states
is universal, and can be read from any example. Thus one gets back the strong
version of \textbf{OV3} inferred from examples in \S.\,\ref{s:cusps}: there exists a $h_R$ as in \eqref{poiqwert} with $|\Delta h_R|\lessapprox c/R$. We opted for the
much weaker version $|\Delta h_R| < c^\prime$: we 
see \textbf{OV3} as a mere technical condition, not an additional swampland conjecture, and our \textbf{OV3} is just the weakest assumption
required to prove
the
``magical'' properties of the OV manifolds.
\end{rem}

\subsection{First consequences of the definition}

\subsubsection{Holonomy groups and the graded algebra $\cp^\bullet$}\label{s:xxx}

A non-compact, complete manifold of non-negative Ricci curvature cannot have finite volume\footnote{\ This follows from the Calabi-Yau lower bound on the volume \cite{ricciflat,calabi} see \textbf{\emph{Theorem 7}} and \textbf{Appendix} (iii) in ref.\cite{ricciflat}.}}, so the Ricci  curvature of the irreducible factor spaces
$\widetilde{M}_k$ somewhere should have some strictly negative eigenvalue.
This rules out a few Riemannian geometries:
$$
\begin{tabular}{cc||c}\hline\hline
allowed $\widetilde{M}_k$  & $\mathfrak{hol}(\widetilde{M}_k)$ & ruled out \\\hline
strictly generic holonomy & $\mathfrak{so}(m_k)$ & Calabi-Yau\\
strictly K\"ahler  & $\mathfrak{u}(m_k/2)$ & hyperK\"ahler\\
strict negative quaternionic-K\"ahler & $\mathfrak{sp}(1)\!\oplus\mathfrak{sp}(m_k/4)$ & positive quaternionic-K\"ahler\\
non-compact symmetric spaces $G/K$ & $\mathfrak{k}$ & $G_2$ and $Spin(7)$ manifolds\\\hline\hline
\end{tabular}
$$
It is remarkable that, a 
part for ``strictly generic'' (which, roughly, corresponds to the non-SUSY case), the list of allowed geometries is
reminiscent of the list of target spaces for supergravity with more than 2 supercharges. We recall that the holonomy algebra
$\mathfrak{hol}(\mathscr{M})$ of the scalars' manifold $\mathscr{M}$ of a supergravity is determined as follows:
from its SUSY algebra and super-multiplet content we read the R-symmetry Lie algebra $\mathfrak{r}$ and its representation $\sigma$ on the scalar fields;  $\mathfrak{hol}(\mathscr{M})$ then has the form \cite{book} (here $m=\dim_\R\mathscr{M}$)
\be\label{whatR}
\sigma(\mathfrak{r})\subseteq\mathfrak{hol}(\mathscr{M})\subseteq\sigma(\mathfrak{r})+\mathfrak{z}(\sigma(\mathfrak{r}))\subset\mathfrak{so}(m)\qquad \mathfrak{z}(\sigma(\mathfrak{r}))\overset{\rm def}{\equiv}\text{centralizer of $\mathfrak{r}$ in $\mathfrak{so}(m)$}.
\ee  
In SUGRA the R-symmetry is a gauge symmetry with composite
gauge fields $\partial_\mu \phi^i\, {\Omega(\phi)_i^a}_b$, where ${\Omega(\phi)_i^a}_b$ is the projection of the Levi-Civita connection of $G(\phi)_{ij}$ on the sub-algebra $\sigma(\mathfrak{r})$. 

Let $\mathfrak{h}_k\equiv\mathfrak{hol}(\widetilde{M}_k)\subset \mathfrak{so}(m_k)$ be the irreducible holonomy algebra of the $k$-th factor space $\widetilde{M}_k$.
The graded algebra $\cp_k^\bullet$ of parallel forms on $\widetilde{M}_k$ is\footnote{\ If $V$ is a vector space carrying a representation of the reductive Lie algebra $\mathfrak{h}$, the symbol $V^\mathfrak{h}$
denotes the $\mathfrak{h}$-invariant subspace of $V$, that is, the subspace on which $\mathfrak{h}$ acts trivially.}
\be
\cp_k^\bullet\equiv\bigoplus_{j=0}^{m_k} \cp_k^j \cong \bigoplus_{j=0}^{m_k}\left(\wedge^j \R^{m_k}\right)^{\!\mathfrak{h}_k}
\ee
 When $\mathfrak{h}_k=\mathfrak{so}(m_k)$, $\cp_k^\bullet$ is spanned by 1 and a volume form $\varepsilon$:
 we say that $\cp_k^\bullet$ is trivial.
When $\widetilde{M}_k$ is strictly K\"ahler, $\cp_k^\bullet=\R[\omega]/\omega^{1+m_k/2}$ with $\omega$ the K\"ahler form.
When $\widetilde{M}_k$ is strictly quaternionic-K\"ahler $\cp_k^\bullet=\R[\Omega]/\Omega^{1+m_k/4}$  with $\Omega$
 the canonical 4-form.
When $\widetilde{M}_k$ is a symmetric manifold the algebra $\cp_k^\bullet$ is typically larger;  but there are two exceptions: $SO(n,1)/SO(n)$ which has generic holonomy $\mathfrak{so}(n)$, and $SU(m,1)/U(m)$ which has strict K\"ahler holonomy $\mathfrak{u}(m)$. We stress that the surgery of \S.\,\ref{surgery!} preserves
the algebra $\cp^\bullet$.

The generic OV manifold has no non-trivial parallel form. Let us consider an example in the other extremum, where $\cp^\bullet$ is so rich that fully determines the metric
(up to overall normalization).

\begin{exe}\label{e1} Suppose the irreducible OV manifold $\mathscr{M}$ has R-algebra $\mathfrak{r}=\mathfrak{su}(8)$ and R-representation $\sigma$ the \textbf{70}. Then -- up to finite covers -- $\mathscr{M}$ is globally isometric to the
locally symmetric space
\be
E_{7}(\Z)\big\backslash E_{7}(\R)\big/SU(8)
\ee
where $E_{7}(\Z)$ denotes a maximal arithmetic subgroup of $E_{7}(\R)$ (the split real form of $E_7$).\footnote{\ With respect to some structure of 
$E_{7}(\R)$ as an algebraic group defined over $\mathbb{Q}$. The $\mathbb{Q}$-algebraic structure is determined by quantum consistency: in this case $E_{7}$ is identified with the universal Chevalley group of type $E_7$ (a scheme over $\Z$)
and $E_{7}(\Z)$ is the groups of its points valued in $\Z$.} This is the moduli space of Type II
compactified on $T^6$ with the correct $U$-duality group $\cg=E_{7}(\Z)$ (up to commensurability).
\end{exe}

%
%

\subsubsection{Liouville property}\label{s:liouville}

Many nice properties of supersymmetric field theories (with or without gravity), such as the non-renormalization/rigidity theorems,
may be traced back to the fact that their moduli spaces $\mathscr{M}$ are ``Liouvillic''.
This property holds automatically for OV manifolds. Morally speaking, ``most'' of the good aspects of SUSY
appear to be given for free once the swampland conjectures are satisfied.

\begin{defn}
The manifold $M$ is \emph{Liouvillic} for the class $\mathcal{C}$ of functions
$f\colon M\to\mathbb{R}$ if
\be
f\in\mathcal{C}\ \text{and}\ |f|< K<\infty\quad\Rightarrow\quad f=\text{constant} 
\ee
\end{defn}

\begin{pro} Irreducible OV manifolds $M$ are Liouvillic for the sub-harmonic functions,
\be
\text{ i.e.}\quad\Delta f\geq0\ \text{and}\ |f|< K<\infty\quad\Rightarrow\quad f=\text{constant}
\ee
\end{pro}

\begin{proof} For $M\equiv \mathbb{R}$ this reduces to the well-known fact that a bounded
 convex function $f\colon\R\to \R$ is a constant. For $M$ a complete, non-compact manifold of finite volume, we write $u\equiv f+K$. $u$ is a non-negative function bounded above by $2K$, and for all $p>1$
\be
\int_\mathscr{M} u^p\, d\mspace{1mu}\mathsf{vol} \leq (2K)^p\cdot \mathsf{Vol}(\mathscr{M})<\infty,
\ee 
and then $u$ is a constant by \textbf{\emph{Theorem 3}} in \cite{ricciflat}. 
\end{proof}

The statements holds even for reducible OV manifolds as long as they have no flat factor
in eqn.\eqref{finitecover}. In particular, it holds for all 4d SUGRA satisfying the Ooguri-Vafa
swampland conjectures \cite{OoV}.
The Liouville property seems to be crucial for the swampland story:
this is the case for all supersymmetric consistent effective theories. 

\begin{rem} Standard Seiberg-Witten theory \cite{SW1,SW2} is grounded on the fact that
 the Coulomb branch of an UV complete $\mathcal{N}=2$ QFTs satisfies the much weaker property of being 
 Liouvillic for the sub-\emph{pluri}harmonic functions.
 \end{rem}

\section{The gauge couplings $\tau(\phi)_{ab}$ as a map}\label{asamap}

One wishes to understand which gauge coupling $\tau(\phi)_{ab}$ may appear in a consistent effective 
Lagrangian \eqref{lag} of
quantum gravity. These couplings have a natural geometric description which cries for 
an intrinsic characterization of the allowed $\tau(\phi)_{ab}$.

\subsection{Generalities}\label{s:gennn}

For fixed values of the scalar fields $\phi^i$ (and given duality-frame), the gauge coupling $\tau(\phi)_{ab}$ is an element of Siegel's upper half-space
\be
\boldsymbol{H}_h\equiv \Big\{\tau\in \mathsf{Mat}_{h\times h}(\C)\;\Big|\; \tau=\tau^t,\ \mathrm{Im}\,\tau >0\Big\}\cong Sp(2h,\R)/U(h).
\ee
$\boldsymbol{H}_h$ is a non-compact K\"ahler symmetric space 
on which the group $Sp(2h,\R)$ acts transitively by isometries. 
Its maximal arithmetic subgroup $Sp(2h,\Z)\subset Sp(2h,\R)$ is the group of electro-magnetic duality-frame rotations.
 As a Riemannian space $\boldsymbol{H}_h$ is Hadamard\footnote{\ That is, a simply-connected manifold with non-positive sectional curvatures; in facts $\boldsymbol{H}_h$ has even non-positive curvature operators.} \cite{hadamard1,hadamard2} so diffeomorphic to $\R^{h(h+1)}$. As a complex manifold $\boldsymbol{H}_h$ is biholomorphic to the (symmetric) bounded domain in $\C^{h(h+1)/2}$ of symmetric $h\times h$ matrices $Z$ with $1-Z\bar Z>0$.

\medskip

In general the complex gauge couplings $\tau(\phi)_{ab}$ are not well-defined
 functions on the scalars' space $\mathscr{M}\equiv \widetilde{\mathscr{M}}/\cg$. Indeed, when we go around a non-contractible loop $\gamma$ in $\mathscr{M}$, we may come back 
 and find that the duality frame was rotated by a non-trivial element of $Sp(2h,\Z)$. This yields a \emph{monodromy representation}
 \be\rho\colon \pi_1(\mathscr{M})\equiv \cg\to Sp(2h,\Z).
 \ee
 The image $\rho(\pi_1(\mathscr{M}))\subset Sp(2h,\Z)$ is the \emph{monodromy group} $\Gamma$,
cfr.\! eqn.\eqref{kkkz12m}.

 To get actual coupling \emph{functions,} it is convenient to lift the gauge couplings to the smooth
  universal cover  $\widetilde{\mathscr{M}}$;
 then we may identify them with 
the map
\be\label{Tmu}
\widetilde\mu\colon\widetilde{\mathscr{M}}\to Sp(2h,\R)/U(h)\equiv \boldsymbol{H}_h,\qquad \widetilde{\mu}\colon \phi\mapsto \tau(\phi)_{ab},
\ee
 which lifts the intrinsic gauge coupling map $\mu$ \footnote{\ Throughout the paper double-headed arrows stand for canonical projections.}
\be\label{tildemucovers}
\begin{gathered}
\xymatrix{\widetilde{\mathscr{M}}\ar[rr]^(0.28){\tilde\mu}\ar@{->>}[d] && Sp(2h,\R)/U(h)\ar@{->>}[d]\\
\mathscr{M}\ar[rr]_(0.28){\mu}&& Sp(2h,\Z)\backslash Sp(2h,\R)/U(h)}
\end{gathered}
\ee
One also says that the lifted map
$\tilde\mu$ is \emph{twisted} by the monodromy representation $\rho$, i.e.\! $\tilde\mu$ satisfies the functional equations
\be\label{twisted}
\tilde\mu(g\cdot x)= \rho(g)\cdot \tilde\mu(x),\quad \forall\;g\in\cg\subset \mathrm{Iso}(\widetilde{\mathscr{M}}),\ \ x\in\widetilde{\mathscr{M}}.
\ee

The target space of $\mu$ in \eqref{tildemucovers}, $Sp(2h,\Z)\backslash \boldsymbol{H}_h$, is an irreducible, non-compact, locally symmetric space of the special form $G_\Z\backslash G/K$ where $G_\Z\subset G$ is an arithmetic (hence Zariski-dense \cite{morris}) discrete subgroup\footnote{\
With respect to the obvious structure of $G\equiv Sp(2h,\R)$ as an algebraic group over $\R$.
In facts, $Sp(2h,\Z)\subset Sp(2h,\R)$ is a maximal arithmetic subgroup \cite{borebb}.
In the present context it is more natural to consider $G$ as the real locus of an algebraic group defined over $\mathbb{Q}$, see \cite{swampIII} and references therein.}.
Maps between spaces with these special properties are well studied in mathematics. 

From an  algebro-geometric viewpoint the target space $Sp(2h,\Z)\backslash \boldsymbol{H}_h$ is a normal quasi-projective variety \cite{bb,bb2}, called the \textit{Siegel variety} (or Siegel scheme),
a special instance of Shimura variety \cite{sh1,sh2}. $Sp(2h,\Z)\backslash \boldsymbol{H}_h$ is not smooth since  $Sp(2h,\Z)/\{\pm1\}$ does 
not act freely on $\boldsymbol{H}_h$. This can be easily repaired by replacing $Sp(2h,\Z)$ with a neat finite-index subgroup
$\Lambda_h\subset Sp(2h,\Z)$
\cite{morris}. Then $\Lambda_h\backslash \boldsymbol{H}_h$ is a smooth quasi-projective \emph{finite} cover 
of $Sp(2h,\Z)\backslash \boldsymbol{H}_h$.  $\Lambda_h\backslash \boldsymbol{H}_h$ is a non-compact manifold of finite volume
whose fundamental group $\Lambda_h$ is generated by unipotent elements: $\Lambda_h\backslash \boldsymbol{H}_h$
\textit{is a basic example of irreducible
OV manifold.}
$\Lambda_h\backslash \boldsymbol{H}_h$ has a canonical compactification to a normal
 projective variety, the Baily-Borel compactification $\overline{\Lambda_h\backslash \boldsymbol{H}_h}^{\mspace{1.5mu}BB}$ \cite{bb,bb2}. $\overline{\Lambda_h\backslash \boldsymbol{H}_h}^{\mspace{1.5mu}BB}$ is not smooth for $h>1$. However we may blow-up it into a smooth projective variety $\boldsymbol{Y}_{\!h}$ such that the divisor at infinity $\boldsymbol{Y}_{\!h}\setminus (\Lambda_h\backslash \boldsymbol{H}_h)$ is simple normal crossing.
 This can be done rather explicitly in terms of a suitable toroidal compactification of $\Lambda_h\backslash \boldsymbol{H}_h$ \cite{tor1,tor2,tor3}.

Then, replacing $\mathscr{M}$ by a finite cover (if necessary),
 we see the gauge coupling $\mu$ as a map between \emph{smooth} OV manifolds
\be\label{hasqwe4}
 \mu\colon \mathscr{M}\to \Lambda_h\backslash \boldsymbol{H}_h.
\ee 
%
%
%
%
It is remarkable that both the source and target spaces of $\mu$ are OV manifolds.
In all known examples of quantum-consistent theories (such as compactifications of the superstring on
Calabi-Yau's) both the image $\mu(\mathscr{M})\hookrightarrow \Lambda_h\backslash \boldsymbol{H}_h$ and
the fibers $\mu^{-1}(s)\hookrightarrow \mathscr{M}$ (equipped with the Riemannian structure induced by their respective embedding\footnote{\ In these examples $\mu$ is an embedding not just an immersion.}) are also  OV spaces.\footnote{\ The statement holds without exceptions because we defined the point to be an OV manifold.}
This is likely to remain true for all effective theories of quantum-consistent  gravity.
%

 \medskip

One goal of this note is to give a preliminary discussion of the following

\begin{que} What properties should have the gauge coupling map $\mu$  in eqn.\,\eqref{hasqwe4} for the Lagrangian \eqref{lag}
 not to sink in the swampland? 
\end{que} 

If a simple condition on $\mu$ exists at all, it should be invariant under $Sp(2h,\Z)$ rotations of the electro-magnetic frame. 
Since we consider only the very extreme IR limit,  this means that
the property should be invariant for $Sp(2h,\R)$.

\subsection{Three different viewpoints on $\tau(\phi)_{ab}$}\label{3differeentt}

Since the coupling $\tau(\phi)_{ab}$ is multi-valued in $\mathscr{M}$,\footnote{\ 
In the SUSY case (and also in the non-SUSY one under the naturalness condition
we propose in section 7), $\tau(\phi)_{ab}$ \emph{must} be multi-valued unless it is a numerical constant as in the examples discussed in \cite{Cecotti:2018ufg}.} its value in a given vacuum $\phi$ is not an intrinsic observable, and we should replace it with some invariantly-defined
quantity.

\subparagraph{Automorphic viewpoint.} In the case of a single photon $h=1$, the modular curve $SL(2,\Z)\backslash \boldsymbol{H}_1$
is biholomorphic to the punctured sphere
\be
j\colon \overline{SL(2,\Z)\backslash \boldsymbol{H}_1}\ \,\widetilde{\to}\ \,\mathbb{P}^1,
\ee 
and one may use as the intrinsically-defined
gauge coupling the value of the Hauptmodul $j$ which is independent of the branch of the
multi-valued function $\tau$ (here $q\equiv e^{2\pi i\tau}$)
\be
j(\tau)\equiv \frac{\big(\theta_2(q)^8+\theta_3(q)^8+\theta_4(q)^8\big)^3}{13\mspace{1.5mu}824\;\eta(q)^{24}}= 
\frac{1}{q}+744+196\mspace{1.5mu}884\,q+\cdots
\ee
This is the viewpoint taken in \cite{Cecotti:2018ufg}. In principle this strategy may be applied for all $h$ since, by the already mentioned
Baily-Borel theorem \cite{bb}, $\Lambda_h\backslash \boldsymbol{H}_h$ has 
enough automorphic invariant functions to fully characterize its points.

\medskip

\subparagraph{Total space viewpoint.} The space $\Lambda_h\backslash \boldsymbol{H}_h$ (resp.\! $\boldsymbol{H}_h$) is the moduli space of enhanced\footnote{\ The Siegel variety $Sp(2h,\Z)\backslash \boldsymbol{H}_h$ is the moduli space 
of principally polarized Abelian varieties. Going to the smooth covering Shimura variety $\Lambda_h\backslash \boldsymbol{H}_h$
leads to the moduli of polarized Abelian varieties endowed with some extra structure: e.g.\! if $\Lambda_h$ is the kernel of
$Sp(2h,\Z)\to Sp(2h,\Z/ m\Z)$ ($m\geq3$) the extra structure is a choice of generators of the group of $m$-torsion points.
An Abelian variety with such extra structure is called an \emph{enhanced} Abelian variety.} 
principally polarized Abelian varieties over $\C$ of dimension $h$ (resp.\! \emph{marked,} principally polarized, Abelian varieties). 
As for the axion-dilaton $\tau\equiv C_0+ie^{-\Phi}$ in $F$-theory \cite{Fth},
often it is more convenient to think of the gauge couplings as a fibration $\varpi\colon\mathscr{X}\to\mathscr{M}$
whose fibers $\mathscr{X}_\phi$ are (enhanced) principally polarized Abelian varieties of periods $\tau(\phi)_{ab}$, and then describe the properties of the gauge couplings in terms of the intrinsic geometry of the total space
$\mathscr{X}$ of the fibration. 
\medskip

There is a more physical ``total space'' construction. One compactifies the 4d effective theory down to 3d on a circle;
each 4d vector field yields two new scalars in 3d and we get two additional scalars from the metric by the KK mechanism.
The resulting 3d scalars' manifold $\mathscr{M}_3$ is fibered over the 4d scalars' manifold $\mathscr{M}$,
and the gauge coupling map $\mu$ is encoded in the fiber's geometry (see \S.\ref{addgrav} for details).
So the two 4d couplings $G(\phi)_{ij}$ and $\mu(\phi)_{ab}$ get geometrically unified in the intrinsic Riemannian geometry of
the 3d target manifold $\mathscr{M}_3$ which, according to the swampland conjectures, should
also be an Ooguri-Vafa space with its own discrete gauge group $\cg_3\,\triangleright\,\cg\sim \cg^0\times\Gamma$.
We shall use this ``total space'' viewpoint when convenient.
\medskip

\subparagraph{Elementary viewpoint.} We mostly adopt a more naive strategy: instead of considering subtle automorphic invariants of the gauge couplings, we shall see \eqref{hasqwe4} as a mere smooth map between Riemannian manifolds, and rely on the invariants of smooth maps which are defined in Differential Geometry (DG) textbooks. 

\subsection{DG invariants of the gauge couplings $\tau(\phi)_{ab}$}\label{s:dginv}


The basic DG invariants of a smooth map $\phi\colon X\to Y$ between Riemann manifolds 
(with metrics $G_{ij}$ and $h_{ab}$, respectively) are:
\begin{itemize}
\item its \emph{energy} $E[\phi]$ given by the value of the Dirichlet integral
\be\label{sssigmamm}
E[\phi]= \frac{1}{2}\int_Xd^nx\,\sqrt{\det G}\, G^{ij}\, h_{ab}\, \partial_i\phi^a\,\partial_j\partial^b\phi^a,
\ee
that is, the action of the Euclidean $\sigma$-model with target $Y$ and source space-time $X$;
\item its \emph{tension field} on $X$
\be\label{tensionq}
T(x)^a\overset{\rm def}{=} -\mspace{1mu}h(\phi(x))^{ab}\,\frac{\delta E[\phi]}{\delta \phi^b(x)}\equiv D^i\partial_i\phi^a(x)\in C^\infty(X,\phi^*TY).
\ee 
\end{itemize} 

We replace our original \textbf{Question} in \S.\ref{s:gennn} with a less ambitious one:
\begin{queos}\label{que???} What can we say about the DG invariants $E[\mu]$ and $T[\mu]$
of the gauge coupling map $\mu\colon \mathscr{M}\to \Lambda_h\backslash \boldsymbol{H}_h$ in a quantum-consistent 4d effective theory of gravity?
\end{queos}
\noindent A partial answer will be given in \S.\,\ref{s:answe}. 
\medskip

Although these DG invariants have a simple definition, in quantum-consistent effective theories of gravity
 they seem to involve rather deep number-theoretical issues. For instance, the
allowed energy levels $E[\mu]$ of the gauge couplings $\mu$ in a quantum-consistent effective
theory of gravity belong to a certain discrete subset $\Xi\subset \R_{\geq0}$,
the \emph{gauge couplings' energy spectrum,} which carries a number-theoretic fragrance. 
Computing $\Xi$ is very hard except in some extremely simple class of effective models.

\begin{baby} The simplest possible quantum-consistent 4d effective theories are the ones with the properties:
$$\boldsymbol{(\ddagger)}\qquad\text{\begin{minipage}{320pt}\it {\bf(1)} the effective theory has $\cn\geq2$ local supersymmetry, and
{\bf(2)} its $U$-duality group $\cg$ is commensurable to the group $\boldsymbol{G}(\Z)$ of ``$\mspace{2mu}\Z$-valued'' points in a universal
Chevalley group-scheme $\boldsymbol{G}$ without simple factors of type $C_{h}$.\end{minipage}}
$$
The prototypical 4d effective theory satisfying $\boldsymbol{(\ddagger)}$
is obtained by compactifying the 10d Type II superstring on a flat 6-torus $T^6$ (see appendix of \cite{soule}): in this case
$\cn=8$ and $\boldsymbol{G}$ has type $E_7$ (cfr.\! \textbf{Example \ref{e1}}).
Whenever $\boldsymbol{(\ddagger)}$ holds one has (for a standard normalization of the metrics)
\be\label{spectrum}
\Xi= \ci \cdot \frac{m}{2\,\mathsf{Vol}(K)}\,\prod_{\ell=1}^r \zeta(d_\ell) 
\ee
where $\{d_\ell\}$ are the degrees of the independent Casimir invariants of the
real Lie group $\boldsymbol{G}(\R)$, $K\subset G(\R)$
is a maximal compact subgroup,\footnote{\ One can show that the absolute ranks of the two Lie groups $G(\R)$ and $K$ are equal.} $\mathsf{Vol}(K)$ its volume
(computed by the Macdonald formula \cite{mac}); $\zeta(s)$ is the Riemann $\zeta$-function, $m\equiv\dim G-\dim K$, and $\ci\subset \mathbb{N}$ is the set
of indices of finite-index subgroups of the maximal arithmetic group $\boldsymbol{G}(\Z)$.  If the $U$-duality group is precisely $\cg$, the energy of the coupling constants is given by the \textsc{rhs}
of \eqref{spectrum} with $\ci$ replaced by $[\boldsymbol{G}(\Z)\colon \cg]$. Eqn.\eqref{spectrum} follows from standard \textsc{susy} arguments together with the Langlands volume formula
for arithmetic quotients \cite{langlands}.
 In the case of Type II on $T^6$ all degrees $d_\ell$ are even,
 so the energy of the gauge coupling has a closed expression in terms of Bernoulli numbers:
$E(\mu)$ is a know rational number times $\pi^{35}$. For Type II on $T^5$, a part for a power of $\pi$,
the energy has a transcendental factor $\zeta(5)\zeta(9)$.
\end{baby}

\subsection{Tension field and harmonic maps}

To answer the \textbf{Simpler Question},
we start by recalling some basic facts about energies and tensions of
smooth maps. Building on these facts, in the next section we shall introduce a more detailed and elegant structure, that
we call \emph{domestic geometry,} modelled on the \textit{variations of Hodge structures} (VHS) \cite{Gbook,deligne,reva,revb,MT4,periods} and the more general \textit{$tt^*$ geometry} \cite{tt*,Cecotti:1990wz,Cecotti:1992rm,dubrovin}\!\!\cite{ttA}.
\vskip6pt

Let $(X,G)$ and $(Y,h)$ be two Riemannian manifolds. As already mentioned, the energy $E(\phi)$ of a map $\phi\colon X\to Y$ is the action of the Euclidean $\sigma$-model with target space $Y$ and space-time $X$, see eqn.\eqref{sssigmamm}.
The map $\phi$ is \emph{harmonic} iff it is a solutions of the corresponding equations of motion\footnote{\ In eqn.\eqref{har}
$\gamma^k_{ij}$ and $\Gamma^a_{bc}$ are  the Christoffel symbols for, respectively, the metric $G_{ij}$ and $h_{ab}$.}
\begin{equation}\label{har}
0=T[\phi]^a \overset{\rm def}{=} G^{ij}\mspace{1mu}D_i\partial_j\phi^a \equiv G^{ij}\Big(\partial_i\partial_j \phi^a - \gamma^k_{ij}\,\partial_k\phi^a+
\Gamma^a_{bc}\mspace{1mu}\partial_i\phi^b\partial_j\phi^c\Big),
\end{equation}
that is, if its tension  
 $T[\phi]\in \phi^*TY$ vanishes. A finite-energy harmonic map is simply an \emph{instanton} of the $\sigma$-model \eqref{sssigmamm}, i.e.\! a classical Euclidean-signature solution of finite action.
 
 When $X$ has dimension 1, a harmonic map is just a geodesic on $Y$ of constant velocity. More generally,
a map $\phi$ is \emph{totally geodesic} iff the full matrix $D_i\partial_j\phi^a$ vanishes and not just its trace as for a general harmonic map. 

When the source space $X$ is K\"ahler,
eqn.\eqref{har} reduces to 
\be
0=G^{i\bar k} D_i\partial_{\bar k}\phi^a\equiv G^{i\bar k}\Big({\delta^a}_b\mspace{2mu}\partial_i+\Gamma^a_{bc}\mspace{2mu}\partial_i\phi^c\Big)\partial_{\bar k}\phi^b
\ee 
and all dependence on the source-space Christoffel symbols $\gamma^k_{ij}$ drops out. In this situation the map $\phi$ is said to be \emph{pluri-harmonic} if the full type-(1,1) tensor
$D_i\partial_{\bar k}\phi^a$ vanishes and not just its trace. This condition depends only on the complex structure of $X$, and is independent of the specific K\"ahler metric $G_{i\bar k}$. 
This means that a pluri-harmonic map is a classical solution of the $\sigma$-model for \emph{all}
choices of the source metric $G_{i\bar k}$ as long as it is K\"ahler. In particular $\phi$ remains a solution if we repair the singularities in 
the K\"ahler metric $G_{i\bar k}$ by a local surgery which keeps it K\"ahler as in \S.\,\ref{surgery!}. 
If, in addition, the target space $Y$ is also K\"ahler
with complex coordinates $z^a$, the pluri-harmonic condition reduces to $D_i\partial_{\bar k}z^a=0$ 
which is automatically satisfied if the stronger condition
$\partial_{\bar k}z^a=0$ holds, i.e.\! if the map $z\colon X\to Y$ is \emph{holomorphic}.

A useful fact is that the composition $\iota\circ \phi\colon X\to Z$
of a harmonic  map $\phi\colon X\to Y$ and a totally geodesic map
$\iota\colon Y\to Z$  is also harmonic \cite{handbook}.

\subsubsection{Energy and tension of maps into symmetric spaces}
The last observation allows us to describe in a simple way the harmonic maps from a Riemannian manifold $M$
 to a symmetric space $G/H$. The map of main interest is the covering gauge coupling $\widetilde{\mu}$ in eqn.\eqref{Tmu},
 and we write explicit expressions for $G/H= Sp(2h,\R)/U(h)$, the general case then being  obvious.  
One considers Cartan's totally geodesic embedding
 \cite{helga}\!\!\cite{handbook}\footnote{\ The target space $Sp(2h,\R)$ is endowed with a
 $Sp(2h,\R)\times Sp(2h,\R)$-invariant \emph{indefinite} pseudo-Riemannian metric, which induces on
 the image of $\iota$ a positive-definite Riemannian metric so that $\iota$ is an isometry onto its
 image (for a proper normalization of the metrics). See eqn.\eqref{principal} and footnote  \ref{uuuuff}.}
\begin{equation}\label{wSb}
\begin{aligned}
&\iota\colon Sp(2h,\R)/U(h)\to Sp(2h,\R),\\
&\iota\colon\ce U(h) \mapsto \cs\overset{\rm def}{=} \ce\ce^t \in Sp(2h,\R)\quad\text{with}\quad
\cs^t=\cs,\quad \cs > 0
\end{aligned}
\end{equation}
where $\ce\in Sp(2h,\R)$ (called a \emph{vielbein} \cite{book})
is any chosen representative in $Sp(2h,\R)$ of the given point in the coset $Sp(2h,\R)/U(h)$. $\cs$ is independent of the choice of $\ce$. 
Then the map
\be\label{tildephi}
\cs\equiv \iota\circ \widetilde{\mu}\colon \widetilde{\mathscr{M}}\to Sp(2h,\R),\qquad
\cs\colon x\mapsto \cs(x)\equiv \ce(x)\ce(x)^t
\ee
is harmonic iff $\widetilde{\mu}$ is harmonic. The energy $E(\widetilde{\mu})$
of the gauge coupling, written in terms of $\cs$,  becomes 
 the action of the $Sp(2h,\R)$ principal chiral model\footnote{\ \label{uuuuff}The statement is a bit formal, since the
action integral in the \textsc{lhs} of eqn.\eqref{principal} does not correspond to a positive-definite 
metric on the non-compact group $Sp(2h,\R)$; all formulae are meant to be analytic continuations from the corresponding compact group $U\mspace{-1mu}Sp(2h)$. However the metric is positive-definite when restricted to the image of the Cartan map $\iota$, i.e.\! on the space of symmetric, positive-definite, real, symplectic $2h\times 2h$ matrices. As a matter of notation, when $y\in\mathfrak{sp}(2h,\R)$,
 $y^\text{o}$ stands for the odd part under the Cartan involution $\theta$, $y^\text{o}\equiv (y-y^\theta)/2$, i.e.\! for  the symmetric
 part of the $2h\times 2h$ matrix $y$. The \textsc{rhs} of the identity \eqref{principal} is manifestly non-negative.}
\be\label{principal}
\frac{1}{8}\int_{\mathscr{M}}\!d^nx\, \sqrt{G}\, G^{ij}\,\mathrm{tr}\Big(\!(\cs^{-1}\partial_i \cs)(\cs^{-1}\partial_j \cs)\!\Big)\equiv \frac{1}{2}\int_{\mathscr{M}}\!d^nx\, \sqrt{G}\, G^{ij}\,\mathrm{tr}\Big(\!(\ce^{-1}\partial_i\ce)^\text{o}(\ce^{-1}\partial_j\ce)^\text{o}\Big)
\ee
  The tension field of the gauge coupling $\widetilde{\mu}$ is more conveniently written as
  \be
  T\equiv d\ast\!\big(\cs^{-1} d\cs),
  \ee
 and $\widetilde{\mu}$ is harmonic iff $\cs$ is a \emph{symmetric} classical soliton of the chiral model
defined on the space-time
$\mathscr{M}$, that is, iff
\begin{equation}\label{tzero}
d\ast\!\big(\cs^{-1} d\cs)=0.
\end{equation}
\emph{Symmetric}   means that the solution satisfies the two additional conditions
  $\cs^t=\cs$ and $\cs>0$.

\subsubsection{Cartan gauge couplings}
We call the inverse matrix $\cs^{AB}$ to  $\cs_{AB}\equiv (\ce\ce^t)_{AB}$
 the \textit{Cartan form} of the gauge couplings. 
 $\iota$ is a global isometry between the Siegel space $\boldsymbol{H}_h$ and the manifold of
 symmetric, positive, symplectic $2h\times 2h$ matrices, so $\cs^{AB}$ and $\tau_{ab}$
 contain exactly the same information.
 With respect to the usual gauge coupling $\tau_{ab}$, its Cartan version $\cs^{AB}$ has the advantage 
 of transforming linearly
 under rotations of the duality frame. All observables, being independent of the frame,
  have nicer expressions when written in terms of $\cs^{AB}$.  
 Writing $\tau_{ab}=X_{ab}+i Y_{ab}$, one has 
 \be
 \cs^{AB} =\left(\begin{array}{c|c}Y^{-1} & -Y^{-1}X\\\hline
-X Y^{-1} & Y+ XY^{-1}X\end{array}\right)^{\!\!\!AB}
 \ee

\section{Domestic geometry}

In 4d $\cn=2$ supergravity the couplings of the vector-multiplets are described
by special K\"ahler geometry, which is equivalent \cite{cec,stro} to the geometry of variations of
Hodge structure (VHS) with non-zero Hodge numbers $h^{3,0}=h^{0,3}=1$
and $h^{2,1}=h^{1,2}=m$, where $m$ is the complex dimension of the special K\"ahler manifold $M$.
VHS itself is a special case of Higgs bundle geometry \cite{simpson}, which we shall refer to as \textit{generalized 
$tt^*$ geometry.} 
In all these geometries $M$ is a K\"ahler manifold.
In the $\cn=2$ SUSY case the geometric swampland problem may be rephrased
as the question of which special K\"ahler geometries do arise as low-energy limits of consistent theories of quantum gravity. A necessary condition \cite{swampIII} is that the corresponding $tt^*$ geometry enjoys certain  
arithmetic properties summarized in the VHS \textit{structure theorem} \cite{reva,revb,MT4,periods}. 
\medskip

In this section we introduce a geometry -- dubbed \textit{domestic} -- modelled on
$tt^*$, which does not require the manifold $M$ to have a complex structure.
Whenever $M$ is K\"ahler, domestic geometry automatically reduces to (generalized)
$tt^*$ geometry.
\medskip

We start with a review of $tt^*$ geometry from a viewpoint which makes natural its
domestic generalization. The review is rather detailed, because we need results and formulae which 
cannot be found in the physical $tt^*$  literature.
Before going to that, we recall some definitions. 

\begin{nott}
We write $G(\R)$ for a non-compact, connected,
semi-simple, real
Lie group with Lie algebra $\mathfrak{g}$, $K\subset G(\R)$ for a maximal compact subgroup,
and $G(\Z)\subset G(\R)$ for a maximal arithmetic 
subgroup. $G(\C)$ and $K(\C)$ stand for the complexification of the Lie groups
$G(\R)$ and $K$, respectively.
\end{nott}

\subsection{(Arithmetic) tamed maps}
$M$ is an oriented Riemannian $m$-fold with a graded algebra $\cp^\bullet$ of parallel forms.
\begin{defn} $X$ a Riemannian manifold. A smooth map $\mu\colon M\to X$ is 
\emph{tamed} iff
\be\label{hhasq1}
  D\ast (d\mu \wedge\Omega)=0\quad\text{for all }\Omega\in\mathcal{P}^\bullet.
\ee
It suffices to require \eqref{hhasq1} for parallel forms $\Omega$ of degree $\leq m/2$
since
\be
\ast D\ast (d\mu \wedge \Omega) \equiv (-1)^{m-1}\, D\ast (d\mu\wedge \ast\, \Omega).
\ee
\end{defn}
Specializing to $\Omega=1$ we see that \emph{tamed} $\Rightarrow$ \emph{harmonic}, i.e.\! the tension field of a tamed map $\mu$ vanishes, $T[\mu]=0$.
Written in components eqn.\eqref{hhasq1} requires  $D^j \partial_i \mu\in \mathrm{End}(TM)\otimes \mu^* TX$
to satisfy
\be\label{whatincompon}
\Omega_{j[i_1\cdots i_{k-1}}\,D^j\partial_{i_k]}\mu=0\quad \text{for all }\Omega\in \cp^k,\quad k=0,1,\cdots, m.
\ee

We are interested only in tamed maps whose target space $X$ is a locally symmetric
space of non-compact type, $X\equiv \Lambda\backslash G(\R)/K$, the prototypical example being 
a smooth finite cover of the Siegel variety 
\be
\Lambda\backslash G(\R)/K \overset{\rm e.g.}{=} \Lambda_h\backslash Sp(2h,\R)/U(h).
\ee

\medskip

In sections 6 and 7 below we address \emph{inter alia} the following 
\begin{que}
What does it mean for 
 the gauge coupling $\mu$ in eqn.\eqref{hasqwe4}
to be \emph{tamed?}
\end{que}
We shall see that being tamed is a natural condition for the gauge couplings which has important consequences.  
Note that in the physical context, $\cp^\bullet$ is an algebra of invariants for the continuous gauge symmetry,
so $\mu$ is tamed iff the algebra of invariant tensors under the ``gauge coupling endomorphisms''
 $D^i\partial_j\tau_{ab}$ contains all gauge invariants.  
 
\smallskip

\subparagraph{Arithmetic tamed maps.} It is convenient to lift the map $\mu$ to a map between simply-connected covers 
\be
\widetilde{\mu}\colon\widetilde{M}\to G(\R)/K
\ee 
which is
 \emph{twisted} by the monodromy representation $\rho\colon \pi_1(M)\to \Lambda$ of the fundamental group
 \be\label{twisted}
 \widetilde{\mu}\circ \xi = \rho(\xi)\cdot \widetilde{\mu},\qquad \forall\;\xi\in \text{(deck group $\widetilde{M}\to M$),}
 \ee 
 that is, $\widetilde{\mu}$ is the lift which makes the following diagram to commute
 \be
 \begin{gathered}
 \xymatrix{\widetilde{M}\ar@{->>}[d]\ar[rr]^{\widetilde{\mu}}&& G(\R)/K\ar@{->>}[d]\\
 M\ar[rr]^\mu && \Lambda\backslash G(\R)/K}
 \end{gathered}
 \ee
 The image $\Gamma\equiv\mu_\ast(\pi_1(M))\subset\Lambda$ is the \emph{monodromy group}.
 The tamed map $\mu$ is \emph{arithmetic} iff
 \begin{itemize}
 \item[(i)] $\Lambda$, hence $\Gamma$,
  is a neat sub-group of a maximal arithmetic subgroup $G(\Z)$: 
  \be
  \Gamma\subset\Lambda\subset G(\Z)\subset G(\R);
  \ee
  \item[(ii)]  its energy
 is finite, $E[\mu]<\infty$.
 \end{itemize} 
 
 An arithmetic tamed map is thus an instanton of the 
 $\Lambda\backslash G(\R)/K$ $\sigma$-model with some additional properties.

\subsubsection{Special cases}\label{s:special}
For $M$ a Riemannian manifold of dimension $m\geq3$
with generic holonomy algebra $\mathfrak{so}(m)$, \emph{tamed} is equivalent to
\emph{harmonic}. There are some special cases:
\begin{itemize}
\item[(a)] for $M= \mathbb{R}$:  \emph{tamed} $\equiv$ \emph{harmonic} $\equiv$ \emph{geodesic;}
\item[(b)] for $M$ strictly K\"ahler: \emph{tamed} $\equiv$ \emph{pluri-harmonic,}\,\footnote{\ A holomorphic map between K\"ahler manifolds is a special instance of pluri-harmonic map.} i.e.\! $D\overline{\partial}\mu=0$
\item[(c)] for $M$ quaternionic-K\"ahler (with $m\geq8$): \emph{tamed} $\equiv$ \emph{totally geodesic}\,\footnote{\ Let us sketch a proof.
We write $a$, $b$ for the ``flat'' indices of an orthonormal frame in the tangent space $T$ at base point in $\cm$.
Let $\mathfrak{ann}(\cp^\bullet)\equiv\{A_{ab}\in \mathrm{End}(T)\colon \Omega_{b[a_1\cdot a_{k-1}}A_{a_k]b}=0\ \text{for all }\Omega\in\cp^\bullet\}$ be the annihilator of $\cp^\bullet$. From eqn.\eqref{whatincompon} we see that $D_a\partial_b\mu\in \mathfrak{ann}\cap \odot^2T$.
For a strict quaternionic-K\"ahkler manifold of dimension $\geq8$ one has $\cp^\bullet=\R[\Omega]$ with $\Omega$ the canonical 4-form.
Then  $\mathfrak{ann}(\cp^\bullet)\cap\odot^2 T=0$ see \textbf{Proposition 1.2} of \cite{corlette2}.} i.e.\! $D_i\partial_j\mu=0$,
\item[(d)] for $M$ locally isometric to a symmetric space $G/K$ -- not of the form $SO(n,1)/SO(n)$ or $SU(n,1)/U(n)$ -- 
we typically have 
\emph{tamed} $\equiv$ \emph{totally geodesic,} see appendix \ref{ttamed} for examples.
\end{itemize}
 
 \medskip

Next we consider the special case (b) in some detail.
This leads to $tt^*$ geometry which is the model geometry which inspires all our constructions.

\subsection{Generalized $tt^*$ geometry}\label{s:tt*}

In this subsection $M$ is a K\"ahler space with local holomorphic coordinates $t^i$.
 
 \begin{defn}\label{tt*defn} A \emph{generalized $tt^*$ geometry} (or Higgs bundle) is a tamed map
from $M$ into a locally symmetric space $\Lambda\backslash G(\R)/K$ of non-compact type.
 We say that the generalized $tt^*$ geometry is
 \emph{arithmetic} iff the underlying tamed map $\mu$ is arithmetic.
 \end{defn}
 
 Let us see how this definition leads to the usual $tt^*$ formalism \cite{tt*}. We recall that the maximal compact subgroup $K\subset G(\R)$ is the fixed locus of a Cartan involution $\theta$  of the semi-simple Lie group $G(\R)$. The Lie algebra
 $\mathfrak{g}$ of $G(\R)$ splits into $\theta$-even and $\theta$-odd parts,
 $\mathfrak{g}=\mathfrak{k}\oplus\mathfrak{p}$, where $\mathfrak{k}$ is the Lie algebra of $K$
 while $\mathfrak{p}$ is a $K$-module.  
 We fix a faithful real representation $\sigma\colon G(\R)\to SL(V,\R)$
 such that $V\cong V^\vee$,
 so that $G(\R)$ is seen as a concrete group of $n\times n$ real unimodular matrices $g$ ($n\equiv \dim_\R V$).
We choose conventions so that, in terms of matrices, $\theta$ is the inverse of the transpose, and we shall write  $g^t$ for $(g^{-1})^\theta$.
 The Maurier-Cartan form $g^{-1}d g$ is a 1-form on the manifold $G(\R)$ with coefficients 
 in $\mathfrak{g}\subset \mathfrak{sl}(n)$ which may be decomposed into $\theta$-even and $\theta$-odd parts.
 \medskip
 
 Let $\widetilde{\mu}\colon \widetilde{M}\to G(\R)/K$ be \emph{any} smooth map
  twisted by the appropriate monodromy representation $\rho$, eqn.\eqref{twisted}.
 We choose a lift $f\colon \widetilde{M}\to G(\R)$ of $\widetilde{\mu}$ and use it to pull-back
 the Maurier-Cartan form to a 1-form on $\widetilde{M}$ with coefficients in $\mathfrak{g}$
 which we
  decompose in $\theta$-even and $\theta$-odd parts as well as in $(p,q)$ form type
\be
f^*\mspace{-2mu}(g^{-1} dg)=A+\bar A+C+\bar C,\quad \text{where}\ \left\{\begin{aligned}
A&\overset{\rm def}{=} f^*\mspace{-2mu}(g^{-1} dg)^\text{even}_{(1,0)}, &&
\bar A&\overset{\rm def}{=} f^*\mspace{-2mu}(g^{-1} dg)^\text{even}_{(0,1)}\\
C&\overset{\rm def}{=}f^*\mspace{-2mu}(g^{-1} dg)^\text{odd}_{(1,0)}, &&
\bar C&\overset{\rm def}{=}f^*\mspace{-2mu}(g^{-1} dg)^\text{odd}_{(0,1)}
 \end{aligned}\right.
\ee
We write $D\equiv\partial+A$, $\bar D\equiv \bar\partial +\bar A$ and $C\equiv C_i\, dt^i$, $\bar C\equiv  \bar C_{\bar k}\, d{\bar t}^{\mspace{1mu}k}$ with
$C_i, \bar C_{\bar k}\in \mathfrak{p}\subset\mathfrak{sl}(n)$.  By construction $D+\bar D$ is a $K$-connection on $\widetilde{M}$,
while
\be
\nabla\equiv D+\bar D+C+\bar C\equiv f^*(d+g^{-1}dg)
\ee 
is a \emph{flat} $G(\R)$-connection on $\widetilde{M}$.
Under a change of lift $f\to f^\prime= f \cdot u$ (with $u\colon \widetilde{M}\to K$)
both connections change by the $K$-valued gauge transformation $u$; hence the $K$-gauge invariants are
independent of the chosen lift $f$ of $\widetilde{\mu}$.
If $\widetilde{\mu}$ is twisted by a representation $\rho$ as in eqn.\eqref{twisted},
the forms $A$, $\bar A$, $C$ and $\bar C$ are invariant under the action of the
deck group, so they may be seen as forms on the K\"ahler base $M\equiv \widetilde{M}/\pi_1(M)$
which are canonically defined (modulo $K$-gauge transformations)
 by the original map $\mu\colon M\to \Lambda\backslash G(\R)/K$.
\medskip

Decomposing the identity $\nabla^2=0$ into even/odd and form type we get
the equalities
\be\label{jjjasq12}
\begin{split}
&D^2+C^2=(DC)=D\bar D+\bar D D+C\bar C+\bar C C=\\
&= (D\bar C)+(\bar D C)= \bar D^2+\bar C^2=(\bar D\bar C)=0
\end{split}
\ee
which hold for \emph{all} smooth maps $\widetilde{\mu}$. 
Now suppose that
the map $\widetilde{\mu}$ is harmonic with respect to some K\"ahler metric $G_{k\bar l}$ on $\widetilde{M}$, that is, 
it satisfies the equation 
\be
\bar D^k C_k\equiv G^{\bar l k}\bar D_{\bar l}\, C_k=0,
\ee  
which implies the equality
\be
\bar D^i \bar D^k\,\mathrm{tr}(C_i C_k)= \mathrm{tr}\Big[(\bar D^i C_k)(\bar D^k C_i)\Big]+ \mathrm{tr}\Big(C_k\bar D^i \bar D^k C_i\Big)
\ee
while the identities \eqref{jjjasq12} yield
\begin{align}
 &\mathrm{tr}\Big[(\bar D^i C_k)(\bar D^k C_i)\Big]=  \mathrm{tr}\Big[(D_k \bar C^i)(\bar D^k C_i)\Big]\equiv \|\bar D C\|^2\\
&\begin{aligned}
 &\mathrm{tr}\Big(C_k\bar D^i \bar D^k C_i\Big)=  \mathrm{tr}\Big(C_k[\bar D^i, \bar D^k] C_i\Big)=-\mathrm{tr}\Big(C_k[\bar C^i, \bar C^k] C_i\Big)=\\ &= \mathrm{tr}\Big([C_i, C_k]\,[\bar C^k, \bar C^i]\Big) \equiv \big\|\mspace{1mu}[C_i, C_k]\mspace{1mu}\big\|^2
 \equiv \big\|\mspace{2mu} C^2\mspace{1mu}\big\|^2
 \end{aligned}
\end{align}
This shows the 
\begin{lem}[Sampson's Bochner-formula \cite{samp,simp2}]  Let $M$ be K\"ahler and $\mu\colon M\to \Lambda\backslash G(\R)/K$ be a \emph{harmonic} map. 
Then
\be\label{hasqwert}
\bar D^i \bar D^k\,\mathrm{tr}\Big(C_i C_k\Big)= \big\|\bar D C\big\|^2+\big\|\mspace{2mu} C^2\mspace{1mu}\big\|^2.
\ee
\end{lem}
The \textsc{rhs} of \eqref{hasqwert} is the sum of two non-negative terms:
hence the integral over $M$ of the \textsc{lhs} vanishes if and only if the two terms on the right are both identically zero;
this implies \textbf{(1)} $\bar D C=0$ which (by definition) says that $\mu$ is pluri-harmonic,
and \textbf{(2)} the algebra generated by the coefficient matrices $C_i$ is Abelian. In facts from eqn.\eqref{hasqwert}
 we see that \textbf{(2)} is an automatic consequence of \textbf{(1)}.\footnote{\ For an alternative proof that  \textbf{(1)} $\Rightarrow$  \textbf{(2)} see \cite{dubrovin} or the appendix of \cite{swampIII}.} 
 The \textsc{lhs} of \eqref{hasqwert} is a total derivative, so its integral over $M$ is a surface term: in particular, when $M$
 is compact a harmonic map is automatically pluri-harmonic \cite{samp}. 
 More generally:
 \smallskip
 
  \textit{A harmonic map $\mu\colon M\to \Lambda\backslash G(\R)/K$ is pluri-harmonic \emph{($\equiv$ tamed)} if and only if}
 \be\label{bochboch}
 \int_{\partial M} \ast\,\mspace{1mu}\mathrm{tr}\big(C_k \bar D^k C\big)\equiv
  \int_{\partial M} \ast\,\mspace{1mu}\mathrm{tr}\big(C_k D \bar C^k\big) =0,
 \ee
 a condition which depends only on the asymptotic behaviour of $\mu$ at infinity in $M$.
 
 \subparagraph{The $tt^*$ equations.} In a generalized $tt^*$ geometry, $\mu$ is pluri-harmonic, and hence $\bar D C=0$. In view of \eqref{jjjasq12},\,\eqref{hasqwert} this implies the 
  $tt^*$ PDEs \cite{tt*}
\be
\begin{split}
&D^2=C^2=(DC)=D\bar D+\bar D D+C\bar C+\bar C C=\\
&= (D\bar C)=(\bar D C)= \bar D^2=\bar C^2=(\bar D\bar C)=0.
\end{split}
\ee
These equations may be summarized in the following statement:
\begin{pro} \textit{For $M$ strictly K\"ahler, $\mu\colon M\to \Lambda\backslash G(\R)/K$ is tamed if and only if}
\be\label{pppp00}
 \big(\nabla^{(\zeta)}\big)^2\equiv\big(D+\bar D+\zeta^{-1}\,C+\zeta\,\bar C\big)^{\!2}=0\ \ \text{for all }\zeta\in\mathbb{P}^1.
\ee
\end{pro}

We write $\mathcal{R}$ for a (commutative)  enveloping algebra $\cu\mathfrak{a}$, where $\mathfrak{a}\subset \mathfrak{gl}(V,\C)$
is a maximal commutative $\C$-subalgebra containing the matrices $C_i$. $\mathcal{R}$ is known as a \emph{chiral ring} \cite{chiralrings}.

\begin{defn}
A $tt^*$ geometry is \emph{strict} if it has a \emph{spectral flow} i.e.\! we can choose $\mathcal{R}$ so that
 $V\cong\mathcal{R}$
as $\mathcal{R}$-modules \cite{chiralrings}. Since $V\cong V^\vee$, this implies $\mathcal{R}\cong \mathcal{R}^\vee$
as $\mathcal{R}$-modules $\Rightarrow$ in a strict $tt^*$ geometry the chiral ring is a (commutative) Frobenius algebra\footnote{\ See \textbf{theorem 1.3} in \cite{polac}.}. The Frobenius pairing is known as 
the topological metric $\eta\colon \mathcal{R}^{\otimes 2}\to \C$ \cite{tt*,dubrovin}.
\end{defn}

The vacuum geometry of a 2d (2,2) QFT is described by a strict arithmetic $tt^*$ geometry
\cite{tt*,Cecotti:1992rm}.
The vacuum bundle\footnote{\ We recall that the vacuum bundle $\mathscr{V}\to\cm$ is the
holomorphic sub-bundle
of the trivial Hilbert space bundle $\mathscr{V}\times\boldsymbol{H}\to \cm$ whose
 fiber $\mathscr{V}_t\subset \boldsymbol{H}$ at  
$t\in\cm$ is given by the subspace of zero-energy states for the Hamiltonian $H_t$ with $F$-term couplings $t$;
$\mathscr{V}$ is equipped with the sub-bundle Hermitian metric induced by the Hilbert-space Hermitian product in $\boldsymbol{H}$. The Hermitian bundle $\mathscr{V}$ and its metric
are insensitive to deformations of $D$-terms couplings \cite{tt*}. } $\mathscr{V}\to \cm$ over the $F$-term coupling space $\cm$
is holomorphic with structure group $K\subset U(\dim_\C\mathscr{V})$, and  $D+\bar D$ is an unitary connection on the Hermitian vector bundle $\mathscr{V}$ which coincides with the Berry connection in the quantum-mechanical sense \cite{tt*}.
A choice of trivialization identifies the fibers of $\mathscr{V}$ with the 
complexification $V_\C$ of the representation space $V$ of the real Lie group $G(\R)$.
Therefore the fibers of $\mathscr{V}$ carry a real structure\footnote{\ A real structure on a $\C$-space is 
an anti-linear map which squares to the identity.} to be identified with the physical PCT operation \cite{tt*}.

The chiral ring $\mathcal{R}$ is more invariantly seen as the fiber $\mathscr{R}_t$ of a sub-bundle $\mathscr{R}\hookrightarrow \mathrm{End}(\mathscr{V})$
consisting of commutative endomorphisms.
For a generalized $tt^*$ the $C_i$'s yield the sub-bundle chain
\be\label{jq1231z}
T\mspace{-2mu}M\hookrightarrow \mathscr{R}\hookrightarrow \mathrm{End}(\mathscr{V}),
\ee
while for a \emph{strict} $tt^*$ geometry
\be\label{jaqq1231}
T\mspace{-2mu}M\hookrightarrow \mathscr{V}\cong  \mathscr{R}\qquad \begin{smallmatrix}\text{\bf \underline{strict} $tt^*$ geometry}\end{smallmatrix}
\ee

\subparagraph{$tt^*$ metric.} To simplify the notation,
we write $\boldsymbol{g}$ for $f^*g$.
Since $\boldsymbol{g}\in G(\R)$, the connection $A+\bar A\equiv (\boldsymbol{g}^{-1}d \boldsymbol{g})^\text{even}$ is
the $K$-connection written in an \emph{unitary}
trivialization of the Hermitian bundle $\mathscr{V}$; 
more precisely, the trivialization is \emph{orthogonal} because of the 
reality structure on $\mathscr{V}$ \cite{tt*}.
Since $\bar D^2=0$ the $K$-connection is also holomorphic,
and it is convenient to work in a \emph{holomorphic} trivialization where $\bar A\equiv0$.
There is a map $U\colon M\to K(\C)$ such that
$U \bar D U^{-1}=\bar\partial$. The transformation between the orthogonal and the holomorphic trivializations is given by 
the complex $K(\C)$-gauge transformation $\boldsymbol{g}\to \boldsymbol{g} U^{-1}$.
The connection $D+\bar D$ is both holomorphic and metric, hence is the unique Chern connection:
  one has
\be\label{holomgg}
A\equiv h\partial h^{-1} = (\boldsymbol{g} U^{-1})^{-1} d (\boldsymbol{g} U^{-1})\Big|_{\theta\ \text{even}}\qquad \bar A=0,
\ee
where $h$ is the fiber metric on $\mathscr{V}$ in the chosen holomorphic trivialization (such that
the topological metric $\eta\equiv \boldsymbol{1}$).
Comparing the two complex gauges
\be
h= U\bar U^{-1}\equiv UU^\dagger
\ee
The Hermitian metric $h$ is called the \textit{$tt^*$ metric} on $\mathscr{V}$ \cite{tt*}.
It satisfies the \textit{reality condition} $h\bar h=1$ \cite{tt*}.\footnote{\ In our conventions the 
topological metric $\eta=1$ as a consequence of our choice $g^\theta=(g^{-1})^t$.}

For a \emph{strict} $tt^*$ geometry, in view of eqn.\eqref{jaqq1231}, the $tt^*$ metric $h$ is identified with
a Hermitian metric on the fibers of $\mathscr{R}$ which induces the sub-bundle metric on $TM$, i.e.\! a Hermitian metric on $M$. 
It is natural to multiply the $tt^*$ metric on $\mathscr{R}$ by a normalization factor so that the section $\boldsymbol{1}\in\mathscr{R}$ 
has norm 1.
In 2d (2,2) QFT the normalized $tt^*$ metric on the coupling space $\cm$ plays the role of the Zamolodchikov metric \cite{tt*}.

\subparagraph{Hodge metric.} In generalized $tt^*$ there is a second, better behaved metric on $M$,
namely the sub-bundle metric on $TM\hookrightarrow \mathrm{End}(\mathscr{V})$ induced by the
$tt^*$ metric on $\mathrm{End}(\mathscr{V})\cong \mathscr{V}\otimes \mathscr{V}^\vee$.
This metric exists independently of the spectral flow and is always K\"ahler
\cite{ttA}. Its K\"ahler form is
\be
i\,K_{i\bar j}\, dt^i\wedge d\bar t^{\bar j}\equiv i\,\mathrm{tr}(C\wedge \bar C). 
\ee
Let $G_{i\bar j}$ be \emph{any} K\"ahler metric on $M$. The (1,1) tensor on $M$
\be
T_{i\bar j}= K_{i\bar j}-G_{i\bar j}\, G^{k\bar l}\,K_{k\bar l}.
\ee
is automatically conserved
\be\label{pqw12aaa}
\nabla^i T_{i\bar j}= G^{i\bar h}\,\nabla_{\bar h} K_{i \bar j}- G^{k\bar l}\,\nabla_{\bar j}K_{k\bar l}=
G^{i\bar h}\Big(\nabla_{\bar h}K_{i\bar j}-\nabla_{\bar j}K_{i\bar h}\Big)\equiv 0.
\ee
Eqn.\eqref{pqw12aaa} has a simple explanation. The map $\mu$, being pluri-harmonic, is harmonic -- hence a classical solution to the $\Lambda\backslash G(\R)/K$ $\sigma$-model --
for \emph{all} choices of the ``spacetime'' K\"ahler metric $G_{i\bar j}$. $T_{i\bar j}$ is just the energy-momentum tensor evaluated on this on-shell field configuration and hence is conserved.

\subsubsection{HIVb  brane amplitudes}

The $tt^*$ geometry of a 2d (2,2) QFT computes  important physical quantities.
The basic ones are
the Hori-Iqbal-Vafa half-BPS brane amplitudes (HIVb) $\Psi(\zeta)_a$ \cite{branes}
which are sections of the bundle
 $\mathscr{V}^\vee\cong \mathscr{V}$ over $\widetilde{M}$ \footnote{\ Eqn.\eqref{hiv} is written in the conventions common in the $tt^*$ literature \cite{tt*},
in particular the  bracket $\langle \cdots |\cdots \rangle$ is linear in its first argument rather than anti-linear
as in usual conventions. The index $a$ is a quantum number labelling the different fundamental branes.}
\be\label{hiv}
\Psi(\zeta)_a[\phi]= \big\langle\, \phi\,\big|\,a\ \text{brane}\big\rangle,\qquad \phi\in\mathscr{V}.
\ee 
The $\tfrac{1}{2}$-BPS brane amplitudes depend on a twistor parameter $\zeta\in \mathbb{P}^1$ which 
specifies the two linear combinations of the supercharges which leave them invariant \cite{branes}.
The brane amplitudes are solutions to the linear PDEs 
\be\label{tt*lax}
\big(D+\bar D+\zeta^{-1}\,C+\zeta\,\bar C\big)\Psi(\zeta)_a=0
\ee
and depend on a choice of $K(\C)$-gauge (i.e.\! of trivialization of $\mathscr{V}\to \widetilde{M}$);
under a change of gauge
\be
\Psi(\zeta)_a\to U\,  \Psi(\zeta)_a,\qquad U\colon \widetilde{M}\to K(\C).
\ee
The $tt^*$ PDEs \eqref{pppp00} are the integrability conditions of the brane equation \eqref{tt*lax}.
\medskip

\subparagraph{Fundamental solutions.} A fundamental solution to eqn.\eqref{tt*lax} is a map 
\be\label{firstX}
\Phi(\zeta)\colon\widetilde{M}\to \sigma(G(\C))\subset SL(n,\C)
\ee
such that the columns $\Phi(\zeta; t,\bar t)_a$
($a=1,\dots, n$) of the matrix $\Phi(\zeta)\equiv\Phi(\zeta;t,\bar t)$ yield a basis of linearly independent solutions of \eqref{tt*lax}.
In a given $K(\C)$-gauge, the fundamental solution is unique up to multiplication on the right by a matrix
$L(\zeta)\in \sigma(G(\C))$ which depends only on the twistor parameter $\zeta$
\be
\Phi(\zeta;t,\bar t)\to \Phi(\zeta; t,\bar t)\, L(\zeta).
\ee
In an orthogonal trivialization of $\mathscr{V}$, the matrix $L(\zeta)$ may be chosen so that $\Phi(\zeta)$ satisfies the
symmetry and reality conditions
  \be\label{reality}
 \Phi(-\zeta)=\Phi(\zeta)^\theta\equiv (\Phi(\zeta)^t)^{-1},\qquad \Phi(\zeta)= \overline{\Phi(1/\bar\zeta)},
  \ee
 and $\Phi(e^{i\theta})\in G(\R)$.
 In the SUSY literature it is more common to use the
 holomorphic gauge (cfr.\! eqn.\eqref{holomgg})
 \be
 \Phi(\zeta)_\text{holo}= U\Phi(\zeta),\qquad \Phi(\zeta)_\text{holo}=h\, \overline{\Phi(1/\bar \zeta)}_\text{holo}
 \ee
 which yield the usual formula for the $tt^*$ metric $h$ as a bilinear in the solution of the linear problem
 \eqref{tt*lax} 
 \cite{tt*,Cecotti:1992rm,dubrovin}
 \be
h= \Phi(\zeta)_\text{holo}\mspace{1mu}\overline{\Phi(-1/\bar \zeta)}_\text{holo}^{\,t}.
\ee
 
 \subparagraph{Integral structure.} The HIVb  fundamental brane amplitudes correspond to a special basis
 of solutions to \eqref{tt*lax}, so
 \be
 \Psi(\zeta)_a= \big(\Phi(\zeta)L(\zeta)\big)_a,
 \ee
 for some $L(\zeta)$. To get the appropriate $L(\zeta)$  note that 
the space of physical $\tfrac{1}{2}$-BPS branes
 has an integral structure: for (2,2) $\sigma$-models it arises because the physical branes 
 have support on a sub-manifold of the target space, and hence represent integral
 elements of the relevant homology group. More generally, the integral structure arises because of Dirac quantization of the brane charge.
 Then the representation space $V$ of $G(\R)$ has the form
 \be
 V= V_\Z\otimes_\Z \R,
 \ee
 with $V_\Z\subset V$ a lattice preserved by the arithmetic subgroup $G(\Z)\subset G(\R)$.
 The brane map $\Psi(\zeta)\colon \widetilde{M}\to \sigma(G(\R))$
 is twisted by a monodromy representation (cfr.\! eqn.\eqref{twisted})
 \be\label{monodrr}
 \xi^\ast\mspace{1mu} \Psi(\zeta)= \Psi(\zeta)\cdot \rho_\zeta(\xi)^{-1},\quad \forall\;\xi\in \text{(deck group $\widetilde{M}\to M$)}
 \ee
 which should respect the arithmetic structure, so $\rho_\zeta(\xi)\in G(\Z)$ and hence $\zeta$-independent. 
 Setting $\zeta=1$ and comparing with eqn.\eqref{twisted}, we see that the basic brane amplitudes
 are given by an integral basis of solutions to \eqref{tt*lax} on which the monodromy
 action is given by multiplication on the right by $\rho(\xi)^{-1}$, where $\rho$ is the monodromy representation of the $tt^*$ geometry.
 Then the $tt^*$ metric $h$ is given by a Hermitian form in the brane amplitudes (written in a holomorphic $K(\C)$-gauge)
 \be\label{lastX}
 h= \Psi(\zeta)_\text{holo}\, I\, \Psi(-1/\bar\zeta)^\dagger_\text{holo},\qquad
 \text{where }\ I\equiv L(\zeta)^{-1}(L(-1/\bar \zeta)^{-1})^\dagger,
 \ee 
with $I$ the intersection form between dual bases of BPS.\footnote{\ In the most interesting case,
i.e.\! graded $tt^*$ geometries, we will present the explicit form of $I$ (see eqn.\eqref{whatI}) checking that $I$ is an element of $G(\Z)$ (as physically expected),
hence $\zeta$-independent.
}
 \medskip
 
\subparagraph{Brane amplitudes in general $tt^*$ geometry.}
 The equations of the HIVb  brane amplitudes, eqns.\eqref{tt*lax},\eqref{reality}, as well as their
integral structure, continue to make sense for all generalized arithmetic 
$tt^*$ geometry, whether it has a spectral flow or not. 

Given a fundamental brane amplitude $\Psi(\zeta)$
we may construct other ones by changing the representation $\sigma$ of $G(\R)$.
Usually the physical branes are given by the fundamental representation $\sigma_\text{fund}$;
all other representations $\sigma$ are sub-representations of some $(\sigma_\text{fund})^{\otimes s}$
defined by an invariant tensor $t\in \sigma\otimes (\sigma_\text{fund}^\vee)^{\otimes s}$;
so the branes amplitudes associated to an arbitrary representation $\sigma$
 may be interpreted as physical \emph{multi-}brane amplitudes
\be
\boldsymbol{\Psi}(\zeta)_{\boldsymbol I}^{\ \boldsymbol{A}}\equiv \boldsymbol{\Psi}(\zeta)_{\boldsymbol I}^{\ A_1\cdots A_s}= t_I^{i_1\cdots i_s}\, \Psi(\zeta)_{i_1}^{\ A_1}\cdots \Psi(\zeta)_{i_s}^{\ A_s}\qquad\quad
\begin{smallmatrix}t_{\boldsymbol I}^{i_1\cdots i_s}\ G(\mathbb{R})\text{-invariant}\\
\text{\phantom{i}tensor}\end{smallmatrix}
\ee

\subsubsection{Graded $tt^*$ geometries \& VHS}\label{s:graded} When the 2d (2,2) QFT is superconformal,
the $tt^*$ geometry has further structure
induced by the superconformal $U(1)_R$ charge. This additional structure characterizes the variations of Hodge structure
inside the larger class of   
$tt^*$ geometries.

\begin{defn} An arithmetic generalized $tt^*$ geometry is \emph{graded} 
 iff there is a grading element $Q\in i\,\mathfrak{g}\equiv i\,\mathfrak{Lie}(G(\mathbb{R}))$ (the \emph{``$U(1)_R$ charge''})
such that
\be\label{jasqw12}
[Q, A]=0,\quad [Q,C]=-C, \quad [Q,\bar C]=\bar C. 
\ee
A graded $tt^*$ geometry is a \emph{variation of Hodge structure (VHS)} if, in addition,
\be\label{pqw12z}
e^{2\pi i Q}\,\big|_{V}=(-1)^{\hat c},\qquad e^{i\pi Q}\in G(\Z),\qquad \mathsf{Ad}(e^{\pi i Q})(g)=g^\theta\ \ \text{for }g\in G(\R)
\ee
where
\be\label{pqw12zz}
 \hat c\overset{\rm def}{=} 2\,\max\!\big\{\text{eigenvalues of}\ Q\ \text{in }V\big\}\in \mathbb{N}
\ee
When \eqref{pqw12z},\eqref{pqw12zz} hold, the pair $(V, Q)$ is called a \emph{Hodge representation of the Lie group $G(\R)$ of weight $\hat c$}. Hodge representations are classified in \cite{MT4}.
\end{defn}

We identify elements of $\mathfrak{g}$ and respectively $G(\R)$ with the matrices which represent them in the real representation space $V$;
one has $g^\theta\equiv (g^t)^{-1}$ for $g\in G(\R)$.
Then eqn.\eqref{pqw12z} implies that $\Omega\equiv e^{i\pi Q}$ is a matrix with integral entries which
satisfies 
\be
g^t\, \Omega\, g =\Omega\quad \text{for all }g\in G(\R),
\ee
while $\Omega$ is symmetric (resp.\! anti-symmetric) for $\hat c$ even (resp.\! odd).
Hence $V$ is an orthogonal (resp.\! symplectic) real representation of $G(\R)$. The non-degenerate, integral, bilinear form
\begin{equation}
\Omega(v,w)\equiv v^t \Omega w,\qquad\text{with}\qquad \Omega(v,w)= (-1)^{\hat c}\,\Omega(w,v),
\end{equation} 
is called the \emph{polarization} of the VHS.
  For a fixed $g\in G(\R)$, we define the \emph{Weil operator} 
  \be
  C_g= g^{-1} e^{i\pi Q} g\in G(\R).
  \ee Then $\Omega(C_g v,\bar w)$ is a positive-definite Hermitian form on $V_\C\equiv V\otimes_\R \C$; indeed
\be\label{poqw12a}
\Omega(C_g v, \bar w)= (-1)^{\hat c}\, w^\dagger \Omega C_g v=  (-1)^{\hat c}\, w^\dagger \Omega g^{-1}\Omega g v = 
 (-1)^{\hat c}\, w^\dagger g^t \Omega^2 g v = w^\dagger g^\dagger g\, v.
\ee

We write $\mathfrak{h}$ for the Lie subalgebra of $\mathfrak{g}$ which commute with $Q$, and
$H\subset G(\R)$ for the corresponding Lie subgroup. By eqns.\eqref{pqw12z},\eqref{poqw12a} the subgroup
$H$ is compact, and hence contained in 
a maximal compact subgroup $K$. $G(\R)/H$ is then a reductive homogeneous space
with a canonical projection into the symmetric space $G(\R)/K$ \cite{griffithsH}. 
Given any $H$-module $W$ we construct canonically a homogeneous bundle $\mathcal{O}(W)\to G(\R)/H$,
with typical fiber $W$, endowed with a unique canonical connection and metric (up to overall normalization)
\cite{griffithsH}
\be
\co(W)= G(\R)\times W\Big/\big\{(g,w)\sim (g h, h^{-1}\cdot w)\, \text{for }h\in H\big\}
\ee

\begin{lem} $G(\R)/H$ is a homogenous \emph{complex} manifold,
and the homogeneous vector bundle $\co(W)\to G(\R)/H$ is holomorphic for all $H$-module $W$.
\end{lem}
\begin{proof} Consider the grading of the complexified Lie algebra\footnote{\ In a VHS the grading is integral, $r\in\Z$.
In this case the first equation \eqref{Pmasqw1} is an \emph{(adjoint) Hodge decomposition of $\mathfrak{g}$} \cite{reva,revb,periods,MT4}.}
\begin{align}\label{Pmasqw1}
&\mathfrak{g}\otimes \C= (\mathfrak{h}\otimes \C)\oplus\!\left(\bigoplus_{r\neq0} \mathfrak{g}^{-r,r}\right),&&\text{where }
\mathfrak{g}^{r,-r}\overset{\rm def}{=}\Big\{X\in \mathfrak{g}\otimes \C \colon [Q, X]=r X\Big\},\\
&\big[\mathfrak{h},\mathfrak{g}^{-r,r}\big]\subseteq \mathfrak{g}^{-r,r},
&&\big[\mathfrak{g}^{-r,r},\mathfrak{g}^{-s,s}\big]\subseteq \mathfrak{g}^{-r-s,r+s}.\label{pzzz1mm}
\end{align}
By eqn.\eqref{pzzz1mm} each summand $\mathfrak{g}^{-r,r}$ is a $H$-module, so it defines a homogeneous vector bundle $\co(\mathfrak{g}^{-r,r})$.
The complexified tangent bundle of the manifold $G(\R)/H$ is
 \be
 TG(\R)/H\otimes \C= \left(\bigoplus_{r>0} \co(\mathfrak{g}^{-r,r})\right)\oplus
  \left(\bigoplus_{r<0} \co(\mathfrak{g}^{-r,r})\right).
 \ee
 We define an almost complex structure on $G(\R)/H$ by declaring the first summand to be the complex distribution of (1,0) vectors.
 The almost complex structure is integrable since this complex distribution is involutive
 by the second equation \eqref{pzzz1mm}.
 For the holomorphic structure of $\co(W)$ see e.g.\! \cite{griffithsH}\!\!\cite{periods}.
\end{proof}

\subparagraph{Infinitesimal period relations.} The sub-bundle $\Theta\equiv \co(\mathfrak{g}^{-1,1})$ of the holomorphic tangent bundle is called the
\textit{Griffiths holomorphic horizontal bundle} \cite{Gbook,periods}. Let $M$ be a complex manifold
and $TM$ its holomorphic tangent bundle. We say that a map 
\be
p\colon M\to G(\R)/H
\ee  
satisfies the Griffiths infinitesimal period relations (IPR) if \cite{Gbook,periods}
\be\label{pqw12bbb}
p_*(TM) \subseteq \Theta.
\ee
In particular, such a map $p$ is holomorphic.
\medskip

Let $V_\C=\oplus_i V_{q_i}$ be the decomposition of
the $G(\R)$-module $V_\C=V\otimes \C$ into eigenspaces of $Q$ of eigenvalue $q_i$. Since $[Q,H]=0$,
each $V_{q_i}$ is a $H$-module and yields a homogeneous bundle $\co(V_{q_i})\to G/H$.

\begin{lem} Let $\tilde\mu\colon \widetilde{M}\to G(\R)/K$ be a lift of the tamed map $\mu$ of a
graded $tt^*$ geometry. Then we have the factorization\vskip-12pt
\be\label{pqe12}
\xymatrix{\widetilde{M}\ar@/^2pc/[rrr]^{\tilde \mu}\ar[rr]^{\tilde p} && G(\R)/H \ar@{->>}[r]& G(\R)/K}
\ee
$\tilde p$, which is called the \emph{(Griffiths) period map} of the graded $tt^*$ geometry,
satisfies the IPR \eqref{pqw12bbb}.
The bundles $\mathscr{V}_{q_i}\equiv \tilde p^* \co(V_{q_i})\to \widetilde{M}$ are called \emph{Hodge bundles}. 
\end{lem}
\begin{proof} In view of eqn.\eqref{Pmasqw1}, eqns.\eqref{jasqw12} are equivalent to the IPR \eqref{pqw12bbb}.
\end{proof}

Usually one identifies the graded $tt^*$ geometry with the period map 
\be\label{Wp?}
p \colon M\to \Gamma\backslash G(\R)/H
\ee
which lifts to $\widetilde{p}$ on the universal cover $\widetilde{M}$. By definition, a period map
 satisfies the IPR.


\subparagraph{Applications to 2d QFT.}
The vacuum geometry of  (2,2) 2d SCFT over the chiral conformal manifold $M$
 is a \emph{graded strict} $tt^*$ geometry. If, in addition, the 
 $U(1)_R$ charges of the chiral primaries \cite{chiralrings} are integral,
the SCFT vacuum geometry is a VHS of CY type,\footnote{\ For a VHS being of CY type is equivalent
to being \emph{strict} as a $tt^*$ geometry, i.e.\!
 that there is a spectral flow isomorphism.} i.e.\! with Hodge number $h^{\hat c,0}=h^{0,\hat c}=1$.
 The vacuum bundle 
 $\mathscr{V}\to M$  then has an orthogonal decomposition (preserved by parallel transport with the Berry connection)
 \be\label{qw1234}
 \mathscr{V}=\bigoplus_{q=-\hat c/2}^{\hat c/2}\mathscr{V}_q,\quad\text{where}\quad Q\big|_{\mathscr{V}_q}= q\in \mathbb{N}-\frac{\hat c}{2}.
 \ee
 In this set-up $\hat c$ is one-third the Virasoro central charge \cite{tt*}. The
 spectral-flow isomorphism is graded by the $U(1)_R$-charge, 
 \be
 \mathscr{R}_q\cong\mathscr{V}_{q-\hat c/2},\quad\text{with}\quad \mathscr{R}=\bigoplus_{q=0}^{\hat c}\mathscr{R}_q,
\ee
and implies
 the ``local Torelli'' property
$TM\cong \mathscr{V}_{1-\hat c/2}$.

\subsubsection{Explicit graded HIVb  brane amplitudes} Comparing with eqn.\eqref{tt*lax} we get an explicit formula for the twistorial  multi-brane amplitudes of a graded $tt^*$ geometry: in the orthogonal trivialization they are $\Z$-linear
combinations of the columns of the matrix
\be\label{explicitbrane}
\boldsymbol{\Psi}(\zeta)_{\boldsymbol{I}}^{\ \boldsymbol{A}}= \big(\zeta^{-Q}\,\boldsymbol{g}^{-1}\big)_{\boldsymbol{I}}^{\ \boldsymbol{A}}\qquad \begin{smallmatrix}\text{monodromy group}\\ \text{acts on the right\phantom{m}}\end{smallmatrix}
\ee
with $\boldsymbol{g}_{\boldsymbol{A}}^{\ \, \boldsymbol{I}}$ the matrix elements of $\boldsymbol{g}\in G(\mathbb{R})$ in the Hodge representation $V$. 
Comparing with eqns.\eqref{firstX}-\eqref{lastX}, we see that $L(\zeta)=\zeta^{-Q}$ and
\be\label{whatI}
I\equiv L(\zeta)^{-1} \big(L(-1/\bar \zeta)^{-1}\big)^\dagger = e^{-\pi i Q}\overset{\rm VHS}{\equiv} (-1)^{\hat c}\,\Omega\in G(\Z),
\ee
where the last equality holds in the VHS case. We see that $I$ is the natural intersection form.
%
%

\subsubsection{Physical quantities from brane amplitudes}\label{phquant}
We can use the brane amplitudes to compute several physical quantities, that is,
 $K$-gauge invariant expressions
which are independent of the choice of $\zeta\in\mathbb{P}^1$.
These physical quantities are well-defined for all generalized $tt^*$
 geometries, graded or non-graded, strict or not. 
Examples are 
\begin{align}
&\bullet\ \text{K\"ahler form of Hodge metric:} && i\, K_{i\bar j}\, dt^i\wedge d\bar t^j=i\, \mathrm{tr}\big(C\wedge\bar C)\\
&\bullet\ \text{Cartan gauge coupling:} &&\mathcal{S}^{AB}=\big(\Psi(\zeta)^t\Psi(\zeta)\big)^{AB}\\
&\bullet\ \text{Hodge bilinears:} && \boldsymbol{\mathcal{S}}^{\boldsymbol{AB}}=\big(\boldsymbol{\Psi}(\zeta)^t\boldsymbol{\Psi}(\zeta)\big)^{\boldsymbol{AB}}
\end{align}
We think of the Hodge bilinears as ``higher versions'' of the gauge coupling.

\subsubsection{$tt^*$ entropy functions \& Mumford-Tate groups}\label{s:MT} 
Suppose our tamed (covering)
map 
\be
\tilde\mu\colon \widetilde{\mathscr{M}}\to Sp(2h,\Z)\backslash Sp(2h,\R)/U(h)
\ee
 is actually the gauge coupling of a 4d effective gravity theory. The classical\footnote{\ By ``classical'' we mean the classical entropy function as computed by the truncation of the effective Lagrangian to two-derivatives; the classical entropy becomes exact asymptotically for large charged $|q|\to\infty$.} entropy function (in Sen's sense \cite{sen1,sen2,sen3}) 
of an extremal black-hole with electro-magnetic charge $q\in V_\Z\cong\Z^{2h}$,
(well-defined on the universal cover $\widetilde{\mathscr{M}}$) is
\be
E_{q}=\pi\, q_A\,\cs^{AB} q_B \equiv\pi\,\|q\|_h^2,
\ee
that is, (up to a factor $\pi$) the Hodge norm-squared  $\|q\|_h^2$
of the charge vector. 

For a general $G$-representation $\sigma$, contained
in $\otimes^s V$, we define the ``generalized entropy functions''
of the generalized $tt^*$ geometry to be
\be
S_{\boldsymbol{q}}\colon \widetilde{\mathscr{M}}\to \mathbb{R}_{+},\qquad x\mapsto \boldsymbol{q_A\, \mathcal{S}}(x)^{\boldsymbol{AB}}\,\boldsymbol{q_B},\qquad
\boldsymbol{q}\in \otimes^s V_\Z, 
\ee
$S_{\boldsymbol{q}}$ reduces to $E_q/\pi$ when $\sigma$ is the fundamental representation
(i.e.\! $s=1$).

\begin{pro}\label{ppro} In (generalized) $tt^*$ geometry the generalized entropy functions
$S_{\boldsymbol{q}}(t)$ are
\emph{sub-pluriharmonic} (in particular \emph{sub-harmonic}) for all $\boldsymbol{q}\in \otimes^s V_\Z$, i.e.\! the matrix $\partial_{t^i}\partial_{\bar t_j} S_{\boldsymbol{q}}$ is semi-definite positive.
\end{pro}

\begin{proof} Set $L=\boldsymbol{\Psi}(-1)\boldsymbol{q}$, so that
$DL=CL$, $\bar DL=\bar CL$, $L^*= L$, and
$S_{\boldsymbol{q}}=L^t L\equiv L^\dagger L$. Then
\be
\bar\partial S_{\boldsymbol{q}} = 2\, L^t \bar D L= 2\, L^t \bar C L  
\ee
and
\be
\begin{split}
\partial\bar\partial S_{\boldsymbol{q}}&= 2 (DL)^t \bar DL+
2 L^t D(\bar C L)= 2 (DL)^t \bar DL- 2\, L^t \bar C C L =\\
&=2\, L^\dagger(C_i C^\dagger_j +C^\dagger_j C_i)L\, dt^i\wedge d\bar t^j
\end{split}\qedhere
\ee
\end{proof}

\medskip

The actual value of the classical entropy for an extremal black hole of charge $q$
is given by the value $\pi\, S_q$ at a critical point $\partial S_q=0$ (if it exists!!).
As an example, we consider the special case of a VHS of CY type with $\hat c=3$, which describes the
gauge coupling $\widetilde{\mu}$ of an effective 4d theory with $\cn=2$
\cite{cec,stro}. A critical point $t_0\in\widetilde{\mathscr{M}}$ of $S_q(t)$ 
corresponds to a $L\equiv\Psi(-1;t_0)q$ which is an eigenvector of $Q^2$
with eigenvalue $9/4$. This observation is called the \textit{attractor mechanism} 
\cite{at1,at2,at3,at4,at5,at6}. \textbf{Proposition \ref{ppro}} implies that the critical point is actually a minimum for the entropy function. 
The attractor mechanism illustrates the point that being sub-harmonic is a very natural property for a physical entropy function, being strictly related to the convexity of
thermodynamical potentials.

\subparagraph{The Mumford-Tate group.}
If the multi-charge $\boldsymbol{q}$ is $\Gamma$-invariant, the generalized entropy function is well-defined on $\mathscr{M}$
not just on $\widetilde{\mathscr{M}}$. Suppose that we have an arithmetic graded $tt^*$ geometry (say, a VHS) whose base
$\mathscr{M}$ is quasi-projective, i.e.\! $\mathscr{M}=\overline{\mathscr{M}}\setminus D_\infty$
for $\overline{\mathscr{M}}$ a smooth (compact) projective variety and
$D_\infty$ a simple normal crossing divisor; then $\mathscr{M}$ is Liouvillic for the sub-pluriharmonic functions. 
Under the above assumptions, one checks (by a careful asymptotic analysis \cite{schm})
that the generalized entropy of a $\Gamma$-invariant charge $\boldsymbol{q}$ is bounded  along the divisors at infinity in the projective closure $\overline{\mathscr{M}}$ of $\mathscr{M}$. Then, being sub-pluriharmonic, $S_{\boldsymbol{q}}(t)$ must be a constant in $\mathscr{M}$. We conclude that  for
a $\Gamma$-invariant multi-charge $\boldsymbol{q}\in\otimes^s V_\Z$
the multi-brane amplitude
$\boldsymbol{\Psi}(\zeta)\cdot\boldsymbol{q}$ is $H$-gauge equivalent
to a constant; this can be expressed as

\begin{pro} Let $\mathsf{Hg}$ be the ring of all $\Gamma$-invariant
multi-charges\ in $\oplus_s (\otimes^s V_\Z)$ and let $M(\R)\subset G(\R)$ be the subgroup
which fixes all elements $\boldsymbol{q}\in \mathsf{Hg}$
and $H_M\equiv H\cap M$ a maximal compact subgroup. Then
the map $p$ in eqn.\eqref{Wp?}
factorizes as
\be\label{sssthm}
\xymatrix{\mathscr{M}\ar@/_2.2pc/[rrrr]^\mu \ar@/^2pc/[rrr]^p\ar[rr]^(0.4){m} && \Gamma\backslash M(\R)/H_M\ar@{->>}[r]& \Lambda\backslash G(\R)/H\ar@{->>}[r]& \Lambda\backslash G(\R)/K }
\ee
\end{pro}
Eqn.\eqref{sssthm} is almost the structure theorem of VHS \cite{reva,revb,MT4,periods}, except that the theorem yields more details on the group $M(\R)$; we shall discuss these results in  the more general domestic context below. In the VHS literature the elements of the $\mathbb{Q}$-algebra $\mathsf{Hg}\otimes\mathbb{Q}$ are called \emph{Hodge tensors} and the $\mathbb{Q}$-algebraic group $M(\mathbb{Q})$ is called the \emph{Mumford-Tate group} \cite{reva,revb,MT4,periods}.

\subsection{Domestic geometry}

(Arithmetic) domestic geometry is defined by  the very same \textbf{Definition \ref{tt*defn}} of $tt^*$ geometry,
except that now we forget that the source Riemannian space $M$ was assumed to be K\"ahler.
An arithmetic domestic geometry on the Riemannian manifold $M$
 is specified by a tamed map
$M\to \Lambda\backslash G(\R)/K$ of finite energy. The geometric structures implied
by domestic geometry  depend crucially on the algebra $\mathcal{P}^\bullet$ of parallel forms 
on $M$ (which we assume to be irreducible with no loss). 
In particular, a domestic geometry on $M$ is a generalized $tt^*$ geometry 
if and only if $\cp^\bullet$ contains a subalgebra $\R[\omega]/\omega^{m+1}$ with $\omega$ a parallel 2-form.
\medskip

Domestic geometry is more general than $tt^*$ geometry as the following example shows.

\subsubsection{Example: 2d SCFT}
The (universal covering of the) conformal manifold $\widetilde{\cm}$ of a 2d (2,2) SCFT splits in a product of spaces associated to the two non-conjugate chiral rings
\be
\widetilde{\cm}=\widetilde{\cm}_\text{chiral}\times \widetilde{\cm}_\text{twisted chiral},
\ee
and the
 Berry geometry on each irreducible factor space is a domestic geometry
of the $\mathfrak{u}(1)$ graded $tt^*$ kind (\S.\,\ref{s:graded}). 

The Berry geometry on the moduli of a 2d (4,4)
SCFT is still a domestic geometry, but \emph{not} a $tt^*$ geometry.\footnote{\ Of course, a (4,4) SCFT is in particular
a (2,2) SCFT; however the space $M$ of marginal deformations which preserve (4,4) SUSY is a non-complex 
submanifold of the complex manifold $\cm$ of marginal deformations which preserve only (2,2) SUSY. } 
Indeed the domestic geometry is $\mathfrak{sp}(1)\oplus\mathfrak{sp}(1)$ graded
rather than $\mathfrak{u}(1)$-graded.
The moduli space $M$ is quaternionic-K\"ahler, and hence the underlying tamed
map $\mu$ is totally geodesic (cfr.\! \S.\,\ref{s:special}): in (real) local coordinates $x^i$
\be\label{lllazq1c}
\nabla_i C_j=0,\quad \text{where }\  C_j\, dx^j= f^*(g^{-1}dg)^\text{odd}.
\ee
Eqn.\eqref{lllazq1c} implies that $\widetilde{M}$ is a non-compact symmetric space with holonomy algebra of the form
$\mathfrak{sp}(1)\oplus\mathfrak{sp}(1)\oplus \mathfrak{j}\subset \mathfrak{so}(4k)$, and hence
\be
\widetilde{M}=SO(4,k)/[SO(4)\times SO(k)].
\ee 
For a (much longer) proof not using domestic geometry, see \cite{deBoer:2008ss}.

\subparagraph{Grading.} As the 2d example illustrates, in physical applications the domestic geometry is graded by the effective R-symmetry Lie algebra $\sigma(\mathfrak{r})$,
cfr.\! eqn.\eqref{whatR}, that is, we have a decomposition
\be\label{pq123z1}
\mathfrak{g}\otimes \C= \bigoplus_{\alpha\in \text{Irr}} \mathfrak{g}_\alpha,\qquad \mathfrak{g}_\text{adj}=\sigma(\mathfrak{r}),\qquad
\mathfrak{g}_\text{triv}=\mathfrak{j}
\ee
where the sum is over the irreducible representations of $\sigma(\mathfrak{r})$.
Eqn.\eqref{pq123z1} generalizes the adjoint Hodge decomposition \eqref{Pmasqw1} of VHS theory to a possibly non-Abelian
$\sigma(\mathfrak{r})$.

\subsubsection{Domestic brane amplitudes} We may repeat much of the $tt^*$ story in this more general setting.
However now, in general, we cannot distinguish differential forms by type,
so we decompose $f^*(g^{-1}dg)$ in just two pieces
\be
A= (f^* g^{-1}d g)^\text{even},\qquad \Phi=(f^* g^{-1}dg)^\text{odd},
\ee 
and we do not have a twistorial $\mathbb{P}^1$-family of flat connections but
only two of them
\be
\nabla^{(\pm)} = d+A \pm \Phi.
\ee  
Consequently, we have only two ``HIVb  brane amplitudes'' $\Psi_\pm$,
which satisfy the equations 
\be
\nabla^{(\pm)}\Psi_\pm=0.
\ee
As a consequence of Dirac quantization of charge, in our applications the group $G(\R)$ preserves some bilinear pairing $V_\Z\otimes V_\Z\to \Z$
given by an integral matrix $\Omega\in G(\Z)$ (symmetric or antisymmetric) with $\Omega\Omega^t=1$.
In this case
\be\label{oooorrty}
\Psi_-=\Psi_+^\theta\,\Omega,\qquad \Psi_\pm^*=\Psi_\pm.
\ee
Again we may consider higher domestic brane amplitudes $\boldsymbol{\Psi}_\pm$ associated to higher representations of the Lie group $G(\R)$ which may be written as multi-linear products of basic brane amplitudes.

The physical quantities of $tt^*$ are  still well-defined (we assume \eqref{oooorrty}):
\begin{align}
&\bullet\ \text{the Riemannian metric:} && ds^2= K_{ij}\,dx^i dx^j\equiv\mathrm{tr}(\Phi_i\Phi_j)dx^i dx^j\\
&\bullet\ \text{Hodge bilinears:} && \boldsymbol{\mathcal{S}}^{\boldsymbol{AB}}=\big(\boldsymbol{\Psi}_\pm^t\boldsymbol{\Psi}_\pm\big)^{\boldsymbol{AB}}\\
&\bullet\ \text{generalized entropy functions:} && S_{\boldsymbol{q}}=\boldsymbol{q_A}\,\boldsymbol{\mathcal{S}^{AB}}\boldsymbol{q_B}
\end{align}
and the symmetric tensor 
\be
T_{ij}= K_{ij}-\frac{1}{2}G_{ij}\, G^{kl} K_{kl}
\ee
is still conserved ($G_{ij}$ is the metric on $M$).

\subsection{Entropy functions in domestic geometry}
When the tamed map $\widetilde{\mu}\colon \widetilde{M}\to Sp(2h,\R)/U(h)$ which defines the domestic geometry
is the gauge coupling of some effective theory of gravity,
Sen's classical entropy function for an extremal black hole of charge $q$ (if it exists!) 
is given by $\pi S_q$, where $S_q$ is the domestic entropy function for the fundamental representation. 
The same argument as in \S.\ref{phquant}
shows that the generalized entropy functions are \emph{sub-tamed},
in particular,\begin{itemize}
\item \textit{sub-harmonic} for $M$ of generic holonomy: $\Delta S_{\boldsymbol{q}}\geq0$
\item \textit{sub-pluriharmonic} for $M$ K\"ahler:
$\partial_i\bar\partial_{\bar j} S_{\boldsymbol{q}}\geq0$
\item \textit{convex} for $M$ quaternionic-K\"ahler and most symmetric spaces:
$\nabla_i\partial_j S_{\boldsymbol{q}}\geq0$. 
\end{itemize} 

\subparagraph{Structure of $\mu$.}
If $M$ is Liouvillic for the sub-tamed functions
and the generalized entropy function $S_{\boldsymbol{q}}$ of a $\Gamma$-invariant
multi-charge $\boldsymbol{q}$ is bounded (a property natural for $\mu$ of finite energy), one concludes that the structure theorem \eqref{sssthm}
holds in the domestic case as well:\vskip-14pt
\be\label{sssthm2}
\xymatrix{M\ar@/^2pc/[rrr]^\mu \ar[rr]_(0.4){m} && \Gamma\backslash M(\R)/K_M\ar@{->>}[r]& \Lambda\backslash G(\R)/K},\qquad K_M=M(\R)\cap K.
\ee  
We shall show in \S.\,\ref{s:stru2} below that this conclusion is valid in the situations of physical interest using a different and more direct approach.
Here we give a rough sketch of an argument along the lines of \S.\,\ref{s:MT}. We have already seen in \S.\,\ref{s:liouville} that all OV manifolds $M$ are Liouvillic for the
sub-harmonic functions hence \emph{a fortiori} for the sub-tamed ones. It remains to show that $S_{\boldsymbol{q}}$ is bounded for $\boldsymbol{q}$ $\Gamma$-invariant.
We have to check the behaviour of $S_{\boldsymbol{q}}$ at infinity in $\widetilde{M}$.
Under our assumptions in \S.\,\ref{s:cusps}, a point at infinity $x_\infty$ is fixed by some
infinite subgroup $P\subset\Gamma\subset Sp(2h,\Z)$ and hence
  is mapped by a continuous extension of  
 $\widetilde{\mu}$ in a point at infinity $y_\infty\in\overline{Sp(2h,\R)/U(h)}$ in the compactification  of $Sp(2h,\R)/U(h)$
 \cite{bb2} which is also fixed by $P$, i.e.\! such that $Py_\infty=y_\infty$. Since the entropy function 
 \be
 S_{\boldsymbol{q}}=\widetilde{\mu}^*\|\boldsymbol{q}\|^2_h
 \ee is the pull-back
 of the Hodge norm-square of $\boldsymbol{q}$
computed on $Sp(2h,\R)/U(h)$, it suffices to check that, whenever $\boldsymbol{q}$ is fixed by $P$, $\|\boldsymbol{q}\|_h^2$ is bounded in a neighborhood of the point at infinity $y_\infty$ fixed by $P$. The last statement is
a purely group theoretical fact. 
\medskip

The property of being sub-tamed is rather natural for a \emph{physical} entropy function, as
illustrated by the attractor mechanism in the $\cn=2$ case.

\section{Domestic geometry and supergravity}

\subsection{Supergravity in 4d}

Consider a 4d supergravity with any number $\cn$ of light gravitini, matter content,
and couplings. Its scalars' universal covering manifold\footnote{\ As always, we work modulo finite groups and finite covers.}
\be\label{pqw123UU}
\widetilde{\mathscr{M}}=\widetilde{M}_{(1)}\times \widetilde{M}_{(2)}\times\cdots\times \widetilde{M}_{(s)}
\ee
is a product of non-compact spaces in one-to-one correspondence with
the types of supermultiplets present in the model. For instance, in 4d $\cn=2$ SUGRA the scalars'
manifold is the space of hypermultiplet scalars
times the space of vector-multiplet scalars. The map $\widetilde{\mu}$,
which describes the coupling of the scalars to the vectors, splits
into a set of maps $\{\widetilde{\mu}_{(i)}\}$ which describe the coupling of vectors to
the scalars in supermultiplets of the $i$-th type.

The geometry of each factor space $\widetilde{M}_{(i)}$
depends on the corresponding supermultiplet.
\emph{Prima facie}
these geometries look quite different one from the other:
in some cases $\widetilde{M}_{(i)}$ is K\"ahler (possibly with additional structures), in other situations
$\widetilde{M}_{(i)}$ does not even admit a natural complex structure. A general feature is that each
 space $\widetilde{M}_{(i)}$ carries a non-trivial algebra $\cp_{(i)}^\bullet$
of parallel forms determined by the representation of R-symmetry on
the scalars as described in \S.\,\ref{s:xxx}. 

Domestic geometry unifies all these seemingly different geometries in a single one.
The traditional supergravity theory, as well as the geometric swampland
conjectures, may be summarized in the following statement:

\begin{fact}\label{fffac} All 4d SUGRA models are defined by a domestic geometry:
\begin{itemize}
\item[\rm 1.] the $\Gamma$-twisted (covering) gauge coupling map $\tilde\mu\colon \widetilde{\mathscr{M}}\to Sp(2h,\R)/U(h)$
is always \emph{tamed.} 
\item[\rm 2.] More precisely, for each factor manifold $\widetilde{M}_{(i)}$ in eqn.\eqref{pqw123UU}
there is a $\Gamma_{\!(i)}$-twisted, $\cp_{(i)}^\bullet$-tamed map
\be\label{kasqwertzz}
\widetilde{\mu}_{(i)}\colon\widetilde{M}_{(i)}\to G_{(i)}/K_{(i)},\qquad \Gamma_{\!(i)}\subset G_{(i)}
\ee
with target a symmetric space $G_{(i)}/K_{(i)}$ and $\tilde\mu$ factorizes as in the commutative diagram\vskip-26pt
$$
\xymatrix{\widetilde{\mathscr{M}}\ar@/^2.5pc/[rrrr]^{\tilde\mu} \ar@{=}[r]& \widetilde{M}_{(1)}\times\cdots\times \widetilde{M}_{(s)}\ar[rr]^(0.42){\;\tilde \mu_{(1)}\times
\cdots\times \tilde \mu_{(s)}\;} && G_{(1)}/K_{(1)}\times \cdots\times G_{(s)}/K_{(s)}\ \ar@{^{(}->}[r]^(0.62){\iota} & Sp(2h,\R)/U(h)}
$$
where $\iota$ is the totally geodesic embedding induced by a subgroup embedding 
\be
G_{(1)}\times \cdots\times G_{(s)}\overset{\iota}{\hookrightarrow} Sp(2h,\R)
\ee 
and
$\Gamma\sim \prod_i\Gamma_{\!(i)}$ modulo finite groups;
\item[\rm 3.] most couplings in the Lagrangian $\mathscr{L}$ of an (ungauged) 4d SUGRA are given by
universal expressions in terms of brane amplitudes of the domestic geometries $\widetilde{\mu}_{(i)}$. For $\cn\geq3$ and the vector sector
of $\cn=2$ ALL couplings are so expressed;
\item[\rm 4.] the domestic geometry rigidity theorems/structure theorems reproduce and extend the usual SUSY non-renormalization theorems;
\item[\rm 5.] if the domestic geometry defined by a map $\widetilde{\mu}_{(i)}$ in \eqref{kasqwertzz}
is \emph{non-arithmetic} the SUGRA falls in the swampland.
\end{itemize}
The statement holds, with the appropriate modifications, in other spacetime dimensions. 
\end{fact}

As we have seen in \S.\,\ref{s:special}, there are several special cases of domestic geometry depending
on the particular algebra $\cp^\bullet$; when the domestic geometry is a VHS,
we may talk of its Hodge numbers and its Mumford-Tate (MT) group. In table \ref{tttq1zaqX} we list
the domestic geometries which arise in 4d SUGRA with the special properties of each one.

\begin{table}
\begin{tiny}
$$
\begin{tabular}{l|p{0.6cm}|c|p{4cm}|p{3.5cm}|c}\hline\hline
$\mathcal{N}$ \& supermultiplet & name

space &$\C?$ & $\bullet$ kind of domestic geometry

$\bullet$ algebra $\mathcal{P}^\bullet$ if special & $\bullet$ for VHS $h^{p,q}\neq0$

$\bullet$ MT group if special & notes\\\hline\hline 
$\mathcal{N}=1$ chiral & & $\checkmark$ & graded $\hat c=1$
 \textbf{non}-strict $tt^*$ & $\{h^{1,0}\}$ &\\\hline
$\mathcal{N}=2$ vector & SKG & $\checkmark$ & graded $\hat c=3$
 strict $tt^*$ & $\{h^{3,0}=1,h^{2,1}\}$\\\hline
$\mathcal{N}=2$ hyper & QK & NO & \textbf{domestic} tamed by

 $\mathcal{P}^\bullet=\mathbb{C}[1,\Omega]$ ($\Omega$ canonical 4-form) &\\\hline
$\mathcal{N}=3$ vector &&$\checkmark$ & graded $\hat c=3$
 \text{non}-strict $tt^*$
tamed 

by $\mathcal{P}^\bullet=[\wedge^*(\mathbb{C}^{3}\otimes \mathbb{C}^k)]^{\mathfrak{u}(3)\oplus\mathfrak{s\mspace{-2mu}u}(k)}$ & $\{h^{3,0}=3,h^{2,1}=k\}$ 

$\Gamma^\mathbb{Q}(\mathbb{R})=SU(3,k)$
 & LU, SH\\\hline
$\mathcal{N}=4$ gravity & $\boldsymbol{H}_1$ &$\checkmark$ & graded $\hat c=1$
strict $tt^*$ & $\{h^{1,0}=1\}$ & LU, SH\\\hline
$\mathcal{N}=4$ vector && NO & \textbf{domestic} tamed by

 $\mathcal{P}^\bullet=[\wedge^* (\mathbb{R}^{6}\otimes \mathbb{R}^{k})^{\mathfrak{s\mspace{-2mu}o}(6)\oplus\mathfrak{s\mspace{-2mu}o}(k)}$ & & LU\\\hline
$\mathcal{N}=5$ gravity && $\checkmark$ &
graded $\hat c=1$
strict $tt^*$ tamed

by $\mathcal{P}^\bullet=[\wedge^* (\wedge^4\mathbb{C}^5)^\vee]^{\mathfrak{u}(5)}$
 & $\{h^{1,0}=10\}$
 
$\Gamma^{\mathbb{Q}}(\mathbb{R})=SU(5,1)$ 
& LU, SH \\\hline
$\mathcal{N}=6$ gravity & $\widetilde{\text{SKG}}$ & $\checkmark$ & graded $\hat c=3$
strict $tt^*$
tamed 
by $\mathcal{P}^\bullet=[\wedge^*(\wedge^4\mathbb{C}^6)]^{\mathfrak{u}(6)}$ & $\{h^{3,0}=15, h^{2,1}=1\}$

$\Gamma^{\mathbb{Q}}(\mathbb{R})= SO(6,\mathbb{H})$ & LU, SH\\\hline
$\mathcal{N}=8$ gravity & & NO & \textbf{domestic} tamed by

$\mathcal{P}^\bullet=[\wedge^*(\wedge^4_+\mathbb{C}^8)]^{\mathfrak{s\mspace{-2mu}u}(8)}$ & & LU\\\hline\hline
\end{tabular}
$$
 \end{tiny}\vskip-0.3cm
 \caption{\label{tttq1zaqX}\small
 4d SUGRAs as domestic geometries. In second column SKG stands for special K\"ahler geometry, QK for quaternionic-K\"ahler,
$\widetilde{\text{SKG}}$ for \emph{twisted} special K\"ahler (with $e^{i\pi Q}\to -e^{i\pi Q}$), and $\boldsymbol{H}_1$ for the upper half-plane.
In third column $\checkmark$ means that the manifold has a natural complex structure. Forth column specifies
the class of tamed geometries and its taming algebra $\cp^\bullet$. In fifth column we specify the data for domestic geometries which
are in fact VHS. In the last column LU means that the geometry is locally unique (so $\widetilde{M}_{(i)}$ is a symmetric space)
and SH that, under the assumption that the domestic geometry is \emph{arithmetic}
$M_{(i)}\equiv \Gamma_{(i)}\backslash \tilde M_{(i)}$ is a \textit{Shimura variety}.
}
 \end{table} 
 
 \begin{rem} For the benefit of the reader we recall on which grounds  
 each statement in \textbf{Fact \ref{fffac}} rests. Items 1. and 3.  are a rephrasing of well-known
 facts about 4d SUGRA (see e.g.\! the textbook \cite{book}); they are perfectly rigorous in the sense of classical Lagrangian field theory. Item 2 follows from 1. and 3. together with 
 the original Ooguri-Vafa swampland conjectures (which we take as ``facts'' for the purpose of the present paper). In other words: it follows (rigorously) from the above standard facts in supergravity together with the assumption that the scalars' manifold is of the OV class. More in detail: for $\cn\geq3$ SUGRA item 2. is shown in \S.4.9 of \cite{book} and for
 $\cn=2$ SUGRA in refs.\cite{swampIII,cec-insta}. In $\cn=1$ SUGRA the gauge coupling $\tau_{ij}$
 is a \emph{holomorphic} map between the scalars' K\"ahler manifold and
 the Siegel variety $Sp(2h,\Z)\backslash Sp(2h,\R)/U(h)$. Then, by definition, $\tau_{ij}$ is the
 \emph{formal} period map
 of 
a weight-1 VHS; since item 2 is just the structure theorem for period  maps in VHS,
we need only to check that the formal period map $\tau_{ij}$
satisfies the conditions of the structure theorem.
The VHS structure theorem is a consequence of the theorem of the fixed part in VHS: hence 
it is enough to verify that this last theorem holds for $\tau_{ij}$; going through
the details of Schmid's proof \cite{schm} of the fixed part theorem, we see that its 
statement holds if the scalars' K\"ahler space is Liouvillic for the sub-pluriharmonic functions.
This property holds for K\"ahlerian OV manifolds, and item 2 follows for $\cn=1$. Item 4. is a consequence of item 3:
non-renormalization of a Lagrangian coupling means that it cannot 
be deformed continuously (in a consistent way); traditionally the SUSY non-renormalization
properties are attributed to ``the power of holomorphy'' \cite{power}. However holomorphy is just the
cheaper way of implementing rigidity of the HIVb  which determines the coupling. It is obvious
that rigidity of the branes suffices to rule out corrections to the associated couplings, while rigidity applies to
a larger class of situations than holomorphy. Conversely, if the coupling is undeformable,
the associated brane is rigid. Item 5. summarizes the previous discussion of arithmetic
domestic geometry and its relation to Dirac quantization of charge.  
 \end{rem}
 
 \subsection{Swampland vs.\! brane amplitudes}
 The above \textbf{Fact} allows for a more physical interpretation of the
 Ooguri-Vafa geometric swampland conditions. For $\cn\geq2$
 the couplings are expressed in terms of brane amplitudes, and
 the swampland conjectures just say that these brane amplitudes have the properties we expect on physical grounds for 
 \emph{actual} extended objects.
 
 The brane viewpoint makes the swampland story a lot less mysterious. Tautologically, a theory is consistent iff it leads to 
 physically sound predictions for all observables. In, say, $\cn=2$ supergravity
 the brane amplitudes are important physical quantities: if they don't behave in the 
 correct way, the theory is doomed to be inconsistent. This is what (typically) happens when 
 the swampland conjectures of \cite{OoV} are not obeyed.

\subsection{Supergravity as moduli-space gravitational instantons}\label{s:answe}

From table \ref{tttq1zaqX} we see that in all 4d SUGRAs the gauge coupling maps 
\be
\mu_{(i)}\colon \Gamma_{\!(i)}\backslash\widetilde{M}_{(i)}\to \Gamma_{\!(i)}\backslash G_{(i)}/K_{(i)}
\ee
are harmonic, hence solutions to the $\sigma$-model with target space 
$\Gamma_{\!(i)}\backslash G_{(i)}/K_{(i)}
$ and action 
\be
S[\mu_{(i)}]= \frac{1}{2}\int_{\Gamma_{\!(i)}\backslash\widetilde{M}_{(i)}} d^{\mspace{1.5mu}n_{(i)}}x\;\sqrt{\det g_{(i)}}\, g_{(i)}^{kl}\,h(\mu)_{ab}\,\partial_k \mu^a_{(i)}\,\partial_l \mu^b_{(i)}
\ee
where $g_{(i)\,kl}$ is the kinetic-term metric on the $i$-th
factor space $\widetilde{M}_{(i)}$.

If our SUGRA model is not in the swampland, the solutions $\mu_{(i)}$
 have finite energy (action) $S[\mu_{(i)}]<\infty$.
Thus a partial answer to 
the \textbf{Simpler Question} on page \pageref{que???}  is that, in the SUSY case, the tension field of the gauge couplings should vanish while their energy must be finite.
\medskip

Although the above statements are fully correct, they looks a bit unsatisfactory.
We are treating the couplings $g_{(i)}$ and $\mu_{(i)}$
asymmetrically -- the first one as a background metric on $\widetilde{M}_{(i)}$ and the second one as a classical dynamical field -- while
the two couplings are on the same footing in the swampland story, in facts 
geometrically unified in the 3d scalars' manifold $\mathscr{M}_3$
(cfr.\! the ``total space'' viewpoint in \S.\ref{3differeentt}).
Then
it is natural to treat also the
moduli metrics $g_{(i)}$ as classical dynamical fields.

A naive proposal will be to replace the $i$-th $\sigma$-model action by its
minimal coupling to gravity (allowing for a cosmological constant),
that is, to consider the following theory living on the moduli space $M_{(i)}$
\be\label{poqw1z}
S[g_{(i)},\mu_{(i)}]= \int_{\Gamma_i\backslash \widetilde{M}_i} d^{\mspace{1.5mu}n_i}\varphi\;\sqrt{\det g_{(i)}}\!\left(-\frac{1}{2\kappa_{(i)}^2}\,R_{(i)}+\Lambda_{(i)}+\frac{1}{2} g_{(i)}^{kl}\,h(\mu)_{ab}\,\partial_k \mu_{(i)}^a\,\partial_l \mu_{(i)}^b\right)
\ee
However the $\mathscr{M}_3$ ``total space'' viewpoint of \S.\ref{3differeentt}
 suggests that additional KK fields must live on $\mathscr{M}$,
so the proposal \eqref{poqw1z} looks a bit naive, and we should not expect it to
work in full generality.
If we are lucky, $S[g_{(i)},\mu_{(i)}]$ may at best be a consistent truncation of the moduli-space gravity theory
(if it exists!).

\begin{cla}
In 4d $\cn\geq2$ SUGRA all couplings $g_{(i)}$, $\mu_{(i)}$ are solutions to the classical equations of motion
following from the action $S[g_{(i)}, \mu_{(i)}]$ for appropriate constants $\kappa_{(i)}$, $\Lambda_{(i)}$. If the SUGRA arises as the low-energy description of
a consistent quantum theory of gravity, the solutions have finite action, i.e.\! are gravitational instantons
on $M_{(i)}\equiv \Gamma_{(i)}\backslash\widetilde{M}_{(i)}$.
\end{cla}

Indeed, the equation of motion for the moduli-scalars are satisfied since the $\mu_{(i)}$ are harmonic.
One has only to check that the Einstein equations
\be\label{EE}
R_{(i) kl}-\tfrac{1}{2}g_{(i)kl}\,R_{(i)}+\kappa_{(i)}^2\,\Lambda_{(i)}\, g_{(i)kl}=\kappa_{(i)}^2\,T_{(i)kl}
\ee
hold for some constants $\kappa^2_{(i)}$, $\Lambda_{(i)}$.
Equivalently, it is enough to show that the three symmetric tensors $R_{(i)kl}-\tfrac{1}{2}g_{(i)kl}R_{(i)}$, $g_{(i)kl}$, and $T_{(i)kl}$
are not linearly independent. For $\mathcal{N}\geq 3$ all factor spaces $M_{(i)}$ are locally symmetric,
hence Einstein $R_{(i)kl}=-\lambda_{(i)}\, g_{(i)kl}$, while the gauge coupling $\widetilde{\mu}_{(i)}\colon \widetilde{M}_{(i)}\to G_{(i)}/K_{(i)}$ is an isometry\footnote{\ More precisely, we may choose $G_{(i)}$ as small as possible and then $\widetilde{\mu}_{(i)}$ is an isometry; if we choose $G_{(i)}$ non-minimal, $\widetilde{\mu}_{(i)}$ is just a totally geodesic
embedding.} (up to overall normalization), so that the three tensors are equal up to an overall constant
and \textbf{Claim} holds. For $\cn=2$ one has two factor spaces $\widetilde{M}_\text{hyper}$
and $\widetilde{M}_\text{vector}$. $\widetilde{\mu}_\text{hyper}$ is the contant map, so $T_{\text{hyper}\,kl}\equiv0$, while 
$\widetilde{M}_\text{hyper}$, being negative quaternionic-K\"ahler, is Einstein, so eqn.\eqref{EE} holds. The tricky case is
 $\widetilde{M}_\text{vector}$, which is a special K\"ahler manifold. That the Einstein equations \eqref{EE}
 hold in this case was shown in \cite{cec-insta}.
 Finally the last statement follows from the fact that evaluated on the appropriate classical solution the action density is proportional to the volume form \cite{cec-insta}, so that finite action is equivalent to finite volume of $M_{(i)}$,
 which is one of the swampland conjectures.

\smallskip

The \textbf{Claim} may be regarded as a general geometric rigidity theorem, \emph{alias} a general SUSY non-renormalization theorem.
E.g.\! the $\mathcal{N}=2$ case yields the two non-renormalization theorems of $\mathcal{N}=2$ SUGRA.

\subparagraph{$\cn=1$ SUGRA.}
In the  $\cn=1$ case we have 
weaker non-renormalization theorems, so we cannot expect that the 
story is as simple as for $\cn\geq2$. We have a non-renormalization theorem
for $F$-term couplings, so we expect 
the gauge coupling $\mu$ to  be still a tamed map; this is of course correct, since
the gauge coupling is holomorphic and \emph{a fortiori} pluri-harmonic. 
The moduli metric, however, is not expected to be a solution of the Einstein equations following from
a simple action of the form \eqref{poqw1z} since we don't have the corresponding non-renormalization
theorem.
One may speculate about more complicated ``dynamical'' equations for the $\cn=1$ moduli metric $g(\phi)_{ij}$ with additional degrees of freedom propagating on the scalars' moduli space $\mathscr{M}$.
Some proposal will be discussed elsewhere.

\section{Domestic geometry and the swampland}

In the supersymmetric context 
the geometric swampland conditions may be conveniently
rephrased as the requirement that the underlying domestic geometry 
 is arithmetic, i.e.\! as the statement that quantum-consistent
 effective models have formal brane amplitudes with the
right properties to correspond to actual physical branes.
Since domestic geometry and its brane amplitudes $\Psi_\pm$
 make sense on \emph{all} Riemannian manifolds,  one is naturally
  led to ask whether the statement is true for all 
quantum-consistent effective theories of gravity, supersymmetric or not.

We have no quantitative control on the quantum-consistency of \emph{non}-supersymmetric effective theories, so
the question is really a matter for speculation.
There are, however, several reasons to believe that domestic geometry
is somehow on the right track even in the \emph{non}-SUSY case:
\begin{itemize}
\item evidence from the SUSY examples;
\item the elegant and deep connection between domestic geometry and the 
Ooguri-Vafa geometric swampland conjectures which apply in general,
not just in the SUSY context;
\item domestic geometry implies physically desirable properties of the entropy functions;
\item physical ``naturalness'' considerations, see the discussion in \S\S.\,\ref{exceppttt} and \ref{addgrav}.
\end{itemize}

\subsection{Domestic geometry vs.\! OV manifolds}\label{jjjjasz}

OV manifolds are the natural arena for tamed maps. Indeed

\begin{tam}\label{yywe} $\mathscr{M}$ a OV manifold and $G(\R)/K$ a symmetric space of non-compact type. 
All maps $\mu\colon \mathscr{M}\to \Lambda\backslash G(\R)/K$ of \emph{finite energy} (action)
may be continuously deformed into a unique tamed map $\mathring{\mu}$ which is the map of minimum energy in its homotopy class.
In particular,
all harmonic maps
  $\mu\colon\mathscr{M}\to \Lambda\backslash G(\R)/K$ which have finite energy
  (action)
  are automatically tamed.
  \end{tam}
  
  We defer the argument to the end of \S.\,\ref{caseOB:sb}.
  
  \medskip
  
  We assume as our working hypothesis that, in a quantum-consistent theory of gravity, the low-energy
  gauge coupling $\mu\colon\mathscr{M}\to Sp(2h,\Z)\backslash Sp(2h,\R)/U(h)$ has finite energy (action).
  This holds in all known (supersymmetric) examples constructed from the superstring, where this condition is strictly related to various swampland considerations \cite{swampIII}, and looks quite reasonable in general.
Under this hypothesis,   
  the actual gauge coupling
   $\mu$ is a continuous deformation of a well-defined tamed map $\mathring{\mu}$, so
  domestic geometry is at least ``qualitatively correct''.
  The ``correction'' $\mu-\mathring{\mu}$ vanishes in the SUSY case, and 
  we shall speculate that it should be small (or even zero) in general.

We stress that, for $\mathscr{M}$ an OV manifold, the tamed map $\mathring{\mu}$
is uniquely determined by the action on the scalars of
the continuous and discrete (bosonic) gauge symmetries. Under our working assumption, the action of the
gauge symmetries on $\mathscr{M}$ is then restricted by the condition that
a finite-energy tamed map $\mathring{\mu}$ does exist. This yields a severe condition on $\pi_1(\mathscr{M})$
which may be regarded as a stringent refinement  of the swampland $\pi_1$-conjecture of \cite{OoV}. 

\subsection{Tension flow and naturalness of the gauge couplings}\label{s:tension} 

When the source is an OV manifold and the target is an arithmetic quotient of
a symmetric space without compact factors, the conditions of  being
\emph{arithmetic tamed} and being \emph{harmonic} are equivalent for maps of finite energy by the \textbf{Tamed property}. Then, to construct an arithmetic domestic geometry on an OV space $\mathscr{M}$, it suffices to impose that the relevant map $\mu\colon\mathscr{M}\to\Lambda\backslash G(\R)/K$ is harmonic of finite energy.
This last condition has a simple interpretation which we now review.

\subsubsection{Tension flow}

Let $M$, $N$ be Riemannian manifolds. In order to construct a harmonic map $\phi\colon M\to N$ in a given homotopy class (equivalently, a covering harmonic map $\tilde\phi\colon \widetilde{M}\to\widetilde{N}$ twisted by a given monodromy representation $\rho$ of $\pi_1(M)$)
one may think of starting with an arbitrary smooth map $\phi_0$ in that class which has finite energy $E(\phi_0)<\infty$ (if it exists\,!),
and then continuously deform it to decrease its energy until we reach a minimum value.
A convenient way of implementing this variational strategy, pionereed by Eells and Sampson \cite{eells}, is the \emph{tension flow.} One considers a family of maps $\phi_t\colon M\to N$, parametrized by 
$t\in \R$, which satisfies the differential equation\footnote{\ $\mathsf{grad}$ stands for the gradient in the Banach manifold of
smooth maps from $M$ to $N$. Cfr.\! eqn.\eqref{tensionq}.}
\be\label{torflowa}
\frac{d\phi_t}{dt}=D\ast d\phi_t\equiv T(\phi_t)\equiv - \mathsf{grad}\, E(\phi_t)
\ee
with initial value the finite energy map $\phi_0$. If the solution to the PDE \eqref{torflowa} exists and its limit as $t\to+\infty$ is smooth,
$\phi_\infty$ is a finite-energy harmonic map in the homotopy class of $\phi_0$. The existence problem for harmonic maps
is then reduced to showing that the solution to the gradient flow  \eqref{torflowa}
 exists and its limit is regular. 
\medskip

For a family of $\rho$-twisted maps $\widetilde{\mu}_t\colon\widetilde{\mathscr{M}}\to Sp(2h,\R)/U(h)$ the flow equation takes the elegant form (cfr.\! eqn.\eqref{tzero})
\be\label{flowa}
\cs_t^{-1}\frac{d \cs_t}{dt}= D^i(\cs_t^{-1}\,\partial_i \cs_t)
\ee
where $\cs_t=(\cs_t)^t>0$ is the composition of $\widetilde{\mu}_t$
with the Cartan diffeomorphism \eqref{wSb}\footnote{\ $P(2h,\R)$ stands for the 
space of positive-definite $2h\times 2h$ real symmetric matrices.} 
\be
\iota\colon Sp(2h,\R)/U(h)\equiv \boldsymbol{H}_h\overset{\sim}{\to} Sp(2h,\R)\cap P(2h,\R).
\ee
The derivative $D^i$ in \eqref{flowa}  is covariant only with respect to the Levi-Civita connection on the source space $\widetilde{\mathscr{M}}$.
From the form eqn.\eqref{flowa}, it is obvious that along the flow one has $\cs_t\in Sp(2h,\R)$
and $(\cs_t)^t=\cs_t>0$ for all $t$.
%

\medskip

The tension flow has many analogies with the
well known Ricci flow on a manifold $M$ (for a review see \cite{ricciflow})
which we may roughly see as the RG flow of the 2d $\sigma$-model with target $M$.
The analogy is not accidental; indeed the tension flow is a special instance of Ricci flow as
we are going to show.\footnote{\  A different application of the Ricci flow to the swampland program has been discussed in \cite{RRRicc}.} 


\subsubsection{Relation to Ricci flow} Before addressing the question of existence of harmonic maps, let us explain the relation of the gradient flow \eqref{torflowa} 
 with the Ricci flow in a context
where the map $\mu\colon\mathscr{M}\to \Lambda\backslash \boldsymbol{H}_h$ is 
the gauge coupling of a 4d field theory. For simplicity we first consider  an effective model
with only scalars and Abelian vectors (no gravity or fermions)
\begin{equation}\label{lag22}
\mathscr{L}_\text{eff}= -\frac{1}{2}F_\pi^2\,G(\phi)_{ij}\,\partial^\mu \phi^i\partial_\mu \phi^j -
\frac{i}{16\pi}\tau(\phi)_{ab}F^a_+F^b_++\frac{i}{16\pi}\bar\tau(\phi)_{ab}F^a_-F^b_-
\end{equation} 
which we interpret as a field theory with an \emph{explicit} UV cut-off $\Lambda_\text{eff}$
at the energy scale where the IR description breaks down.  The scalar fields $\phi^i$
are seen as adimensional local coordinates on $\mathscr{M}$, and $F_\pi$ is the overall
mass scale of their kinetic terms. Except in \S.\,\ref{exceppttt} we set $F_\pi=1$. 

We compactify this 4d model to 3d on a circle of radius $R$.
Each 4d vector yields two real scalars in 3d: one from the internal component $A_4^a$ and one from the dual to
$A^a_\mu$. The  $2h$ scalars arising from the 4d vectors are periodic since they correspond to the  electric and magnetic $U(1)^h$ holonomies along the circle; we parametrize
the $2h$ holonomies as $\exp(2\pi iy^A)$ ($h=1,\dots, 2h$).

The resulting 3d effective theory is the $\sigma$-model with scalars' manifold the total space of the fibration  
$\mathscr{X}\to\mathscr{M}$ mentioned in \S.\,\ref{3differeentt}:
the fiber $\mathscr{X}_\phi$ is a $h$-dimensional Abelian variety (over $\C$) with periods $\tau(\phi)_{ab}$ and fixed
principal polarization (which we identify with its K\"ahler class).
The total space $\mathscr{X}$ is equipped with the metric
\be\label{3dmeta}
ds^2_\text{3d}= R\, G(\phi)_{ij}\,d\phi^i d\phi^j+ \frac{1}{R}\, \cs(\phi)_{AB}\, dy^A\,dy^B +\big(\text{exponentially small as $R\to\infty$}\big)
\ee 
where $\cs_{AB}\equiv (\ce\ce^t)_{AB}$ is the inverse of the Cartan coupling $\cs^{AB}$. The exponentially small corrections are due to 4d massive particles, carrying electro-magnetic charges, whose world-line wrap the circle; such corrections are well studied in the context of 3d compactifications of 4d $\cn=2$ QFT \cite{GMN}. We shall take $R$ large and ignore the exponential corrections.

Compactifying further down to two dimensions, we get a 2d $\sigma$-model with target space metric proportional to 
\eqref{3dmeta} whose RG flow is given (in the one-loop approximation) by the Ricci flow. The flow preserves the structure of 
the metric \eqref{3dmeta} so it decomposes into a pair of equations of the form
\be\label{2flowsa}
R\,\frac{d}{dt}G(\phi)_{ij}=-2\,R(\phi)_{ij},\qquad \cs^{AC}\mspace{1mu}\frac{d}{dt}\cs(\phi)_{CB}=-2\,{R(\phi)^A}_{B}
\ee 
The Ricci curvature in the fiber directions is easily computed to be
\be
{R^A}_B= -\frac{1}{2}  D^i\mspace{-1mu}{\big(\cs^{-1} \partial_i\mspace{2mu}\cs\big)^{\!A}}_B.
\ee
Comparing with eqn.\eqref{flowa}, we see that the Ricci flow restricted to the fibers has the same form as
the original tension flow. However, in general, the two flows in \eqref{2flowsa}
are coupled together because the covariant derivative $D^i$ depends on the evolving metric $G_{ij}$ on the base. When the base $\mathscr{M}$ is Ricci-flat or Einstein, the two flows are identical.

Note that ${R^A}_A=0$, so that the Ricci flow preserves the volume of the
Abelian fibers. In facts, the flow preserves the fiber's K\"ahler form ($\equiv$ polarization) 
while changing its complex structure. To see this,
we introduce the orthonormal co-frame 
\be
e^m \equiv {\ce_A}^m\,dy^A\qquad \ce\in Sp(2h,\R)
\ee
where $\ce$ is the vielbein in \eqref{wSb}.
The isotropy group $U(h)$ acts on the ``flat'' index $m$ in the representation $\boldsymbol{h}\oplus\boldsymbol{\bar h}$,
defining on the fiber a torsion-less flat $U(h)$-structure hence an integrable complex structure and a closed K\"ahler form which is  given in ``flat'' indices by the constant symplectic matrix $\Omega_{mn}$.
Then the fiber K\"ahler form is
\be\label{kaf}
\omega_\text{fiber}\overset{\rm def}{=} \Omega_{mn}\, e^m\wedge e^n =(\ce \Omega \ce^t)_{AB}\, dy^A\wedge dy^B\equiv 
\Omega_{AB}\,dy^A\wedge dy^B
\ee  
which is preserved by the Ricci flow of the fiber metric
\be\label{flowMCz}
\cs^{-1}\mspace{1mu}\partial_t\mspace{1mu}\cs=  D^i\mspace{-0.5mu}(\cs^{-1}\mspace{1mu}\partial_i\mspace{1mu} \cs).
\ee
The Ricci flow preserves the form \eqref{3dmeta} of the 3d target metric with $\cs_{AB}$ a positive symmetric matrix in
$Sp(2h,\R)$. 
Thus the Ricci flow for the metric \eqref{3dmeta}
evolves the complex moduli of the Abelian fibers of $\mathscr{X}\to\mathscr{M}$ but not their K\"ahler moduli.
At a weakly-coupled fixed-point of the 2d RG flow the tension $D^i(\cs^{-1}\partial_i \cs)$ vanishes and hence the fixed-point
map $\widetilde{\mu}$ is harmonic.

\subsubsection{Naturalness of the effective Lagrangian $\mathscr{L}_\text{eff}$} \label{exceppttt}

In terms of the original 4d effective theory \eqref{lag22} the vanishing of the tension $T(\mu)$ means that the two\footnote{\ There is a third one-loop graph (the tadpole graph) proportional to the Christoffel symbols $\gamma^i_{jk}$.
The tadpole graph may be set to zero by using normal coordinates in the perturbative expansion.}
 (quadratically divergent)
one-loop Feymann graphs in figure \ref{figX} cancel.  Indeed the leading
correction to gauge couplings is proportional to the tension of the gauge coupling
\be\label{correctionx}
\delta\tau_{ab}\propto \frac{\Lambda^2_\text{eff}}{F_\pi^2}\; T(\mu)_{ab}\equiv \frac{\Lambda^2_\text{eff}}{F_\pi^2}\; D^i\partial_i \tau_{ab}. 
\ee
In order for the 4d Lagrangian \eqref{lag22} to be meaningful as a weakly-coupled  effective description
of the low-energy dynamics, the correction \eqref{correctionx} should be rather small. If $F_\pi \lesssim \Lambda_\text{eff}$, the 
tension should be approximately zero $T(\mu)_{ab}\approx0$.  
Thus the vanishing of the tension may be seen as 
a ``naturalness'' requirement for the weakly-coupled 4d effective theory $\mathscr{L}_\text{eff}$. This is equivalent to the statement that the gauge coupling $\mu$ is (approximately) harmonic, i.e.\! that its energy is near the minimum value\footnote{\ Assuming our working hypothesis that the energy of $\mu$ is finite.}  consistent with the monodromy representation $\rho$.

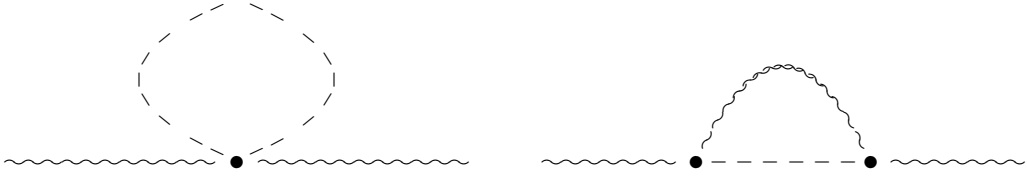
\begin{figure}
$$
\begin{gathered}\xymatrix{&&&
\\
\\
\ar@{~}[rrr]&&&\bullet\ar@{--}@<-0.3ex>@/^3.2pc/[uu]\ar@{--}@<0.3ex>@/_3.2pc/[uu]\ar@{~}[rrr]&&&}
\end{gathered}
\quad
\begin{gathered}
\xymatrix{\\
\\
\ar@{~}[rr]&&\bullet\ar@{--}[rr]\ar@{~}@/^3pc/[rr]&&\bullet \ar@{~}[rr]&&}
\end{gathered}
$$
\caption{\label{figX} One-loop corrections to the photon propagator arising from the scalar-vector couplings in the vectors' kinetic terms. Wavy lines are photons, dashed ones scalars.}
\end{figure}

The interpretation of $T(\mu)_{ab}=0$ as a naturalness condition on $\mathscr{L}_\text{eff}$ is reflected in the fact that
the gauge coupling tension vanishes for \emph{all} 4d supersymmetric theories -- whether they are rigid SUSY QFTs or
supergravities.
This is one of the many ways in which SUSY improves naturalness. 
At weak coupling, naturalness of the gauge couplings does not require
 $\tau(\phi)_{ab}$ to satisfy all the detailed constraints of supersymmetry: the much weaker condition that the gauge coupling map $\mu$
is harmonic suffices. When $\mathscr{M}$ is an OV manifold
and $\mu$ has finite energy, the weaker condition
implies that the gauge coupling $\mu$ is tamed and arithmetic;
whenever the holonomy of the OV manifold is not generic,
we either get an arithmetic $tt^*$ geometry or an arithmetic domestic geometry 
of the kind we have in extended supergravity.\footnote{\ That is, a totally geodesic embedding into the locally symmetric space $\Lambda\backslash G(\R)/K$.} In other words, when $\mathscr{M}$ is an OV manifold with a sufficiently large algebra
$\cp^\bullet$, the weaker naturalness condition on $\mu$ is equivalent to the detailed constraints from supersymmetry.
\medskip

Let $\mu\colon\mathscr{M}\to \Lambda\backslash Sp(2h,\R)/U(h)$ be the gauge coupling of 
an effective theory $\mathscr{L}_\text{eff}$, with $\cm$ an OV space. Under our working
hypothesis that $\mu$ has finite energy, 
the homotopy class $[\mu]$ contains a unique tensionless map $\mathring{\mu}$,
which is automatically tamed. The naturalness argument above suggests that the 
torsion $T(\mu)$, and hence the correction
 $\mu-\mathring{\mu}$,  is small, that is, that domestic geometry is (at least)
 approximatively correct.

\medskip

\subsubsection{Naturalness: Adding gravity}\label{addgrav}

We add gravity to the model \eqref{lag22}, and return to the original Lagrangian \eqref{lag} coupled to gravity, and again compactify the model on $S^1$. In a gravity theory this is just a topological sector of the 4d theory, and the 3d effective low-energy physics ought to be quantum-consistent if the 4d theory is.
In presence of gravity the 3d effective theory has two additional light scalars: $\rho$ corresponding to $g_{44}$ and the dual $z$ of the KK vector $g_{\mu4}$.
The 3d scalars' manifold $\mathscr{M}_3$ is now a fibration over $\mathscr{M}\times \R_\rho$ with fiber a copy of the locally homogeneous  
space $H(\Z)\backslash H(\R)$ where $H(\R)$ is the real Heisenberg Lie group
\be
H(\R)=\{ (z,y^A)\in \R^{2h+1}\} \qquad
(z,y^A)\cdot (w, u^A)= (z+w +\Omega_{AB}\, y^A u^B, y^A+u^A).
\ee
Ignoring quantum corrections (which are supressed as $\rho\to\infty$), the Einstein frame 3d scalars' metric takes the  form
\be\label{THFmett}
\begin{split}
ds^2&=ds^2_\text{base}+ds^2_\text{fiber}=\\
&= \left(G_{ij}\,dx^i dx^j+\frac{1}{2 \rho^2} d\rho^2\right)+\frac{1}{2\rho^2}\bigg(2\rho\, \cs_{AB}\,dy^A dy^B + \big(dz+\Omega_{AB}\, y^A dy^B\big)^2\bigg),\end{split}
\ee
with $G_{ij}$, $\cs_{AB}$ and $\Omega_{AB}$ as before.
The Killing vector $\partial_z$ defines on each fiber a transversely holomorphic foliation (THF)\footnote{\ Probably the best reference for  
transversely holomorphic foliations (THF) is the physics paper \cite{THF} which discuss them in the context of supersymmetry on curved 3d manifolds. They discuss the 3d case; the extension to an arbitrary odd number of dimensions is straightforward.}
endowed with a THF-compatible transversely K\"ahler metric: the normal bundle to the leaves of the transversely holomorphic foliation carries a complex structure as well as the transverse Hermitian metric $\rho^{-1}\cs_{AB}\,dy^A dy^B$ which is K\"ahler when restricted to a fiber with K\"ahler form $\rho^{-1}\,\Omega_{AB}\,dy^A\wedge dy^B$ (cfr.\! eqn.\eqref{kaf}). The complex structure of the normal bundle to the THF foliation
on each fiber is again specified by the uni-modular matrix $\cs_{AB}$. 

In view of the Heisenberg symmetry and the scaling invariance $(\rho, z, y^A)\to (\lambda^2 \rho, \lambda^2 z, \lambda\mspace{1mu}y^A)$, the components of the Ricci tensor along the fiber must have the form
\be
R_{AB}\, dy^A dy^B+ R_{zz}\, dz^2= \frac{1}{\rho} A_{AB}\, dy^A dy^B + \frac{1}{\rho^2} B\big(dz^2 +\Omega_{AB}\, y^A dy^B\big)^2
\ee
where $A_{AB}$ and $B$ are functions only of the coordinates $x^i$ of $\mathscr{M}$. In particular the Ricci flow preserves the THF of each fiber, while
evolves its transverse complex structure and K\"ahler form
\be\label{pqw12}
\frac{d}{dt}\!\left(\frac{1}{\rho}\, \cs_{AB}\right)=-\frac{2}{\rho}\, A_{AB}
\ee
where
\be\label{pzaaac}
A=-\frac{1}{2}\,\cs\big(D^i(\cs^{-1}\partial_i \cs)-c\,\boldsymbol{1}\big)
\ee
for some constant $c$ which depends only on $h$ and $m$. Writing $\cs= f \cs^\prime$, with $\det \cs^\prime=1$, we get for the uni-modular matrix $\cs^\prime$ the same flow equation as before, eqn.\eqref{flowMCz}, so the evolution of the transverse complex structure of the fiber-wise 
THF is again given by the 4d tension flow.
The new aspect is that now the transverse K\"ahler form gets rescaled by the factor $e^{ct}$. This is due to the fact that gravity modifies the Ricci flow equation so that
the appropriate 3d ``fixed point condition'' requires $\mathscr{M}_3$ to be Einstein rather than Ricci-flat;
the easiest way to see this is to compare 3d $\cn\geq3$
rigid SUSY field theories, which have Ricci-flat target spaces, with the $\cn\geq3$ supergravities which have Einstein target spaces with $R_{ij}=-\lambda\mspace{1mu} g_{ij}$,
where $\lambda>0$ is a universal constant which
depends only on the field content. In ``natural'' effective theories the additional term proportional to
$\cs_{AB}/\rho$ in the \textsc{rhs} of eqn.\eqref{pqw12} should be cancelled
by diagrams of the form in figure \ref{figX} with the scalar internal lines replaced by graviton propagators whose contribution is $\propto \cs_{AB}/\rho$.   
Thus the absence of torsion of $\mu$ can again be interpreted as a naturalness
 condition on the 4d effective gravity theory.
All weakly-coupled infrared fixed points then have tensionless gauge couplings $\mu$.

In presence of fermions, the quadratic divergence of the photon propagators, in addition to the diagrams in figure \ref{figX}, gets contributions from 
one-loop of fermions\vskip4pt
\be\label{1loopx}
\xymatrix{\ar@{~}[r]& \mspace{-8mu}\circ\mspace{-8mu} \ar@/^1.5pc/[rr]&& \mspace{-8mu}\circ\mspace{-8mu} \ar@/^1.5pc/[ll]\ar@{~}[r]&}
\ee\vskip12pt
\noindent where $\circ$ stands for Pauli coupling. In the \textsc{susy} case these new contributions would not change our conclusions because of the
 relation of Pauli couplings to the gauge couplings implied by supersymmetry:
 the net effect is just a further contribution to the constant $c$ in \eqref{pzaaac}, that is,
 the Pauli couplings will not modify the flow of the fiber's transverse complex structure. 
 It seems plausible that this conclusion remains valid in all consistent effective theories, possibly with more complicate flows for the fiber's transverse K\"ahler moduli, but 
without affecting the flow of the fiber's transverse complex moduli.
If this is the case, the ``natural'' value of the tension of the gauge coupling is still zero 
(or very small). We expect the fiber's flow to have this property for the following reason:
according to the definition of OV manifold, 
$\pi_1(\mathscr{M})$ is ``big'' and -- unless $\mu$ is the constant map\footnote{\ This exceptional case does happen in quantum-consistent theories: consider the 4d $\cn=2$ SUGRA which describes Type IIB on a \emph{rigid} Calabi-Yau 3-fold.} -- its image $\Gamma$ in the Siegel modular group should also be ``big''.\footnote{\ Assuming $\mu$ not to be the constant map
(which is a special case of harmonic map), 
the technical statement is that $\Gamma$ is Zariski dense in some
 \emph{non-compact} Lie subgroup $G(\R)\subset Sp(2h,\R)$: see \S.\ref{s:stru1}. This means that all algebraic invariants for the group $\Gamma$ are invariants for the full group $G(\R)$,
 which generically is the full group $Sp(2h,\R)$. }
Since $\Gamma$ is a gauge symmetry, the Pauli couplings should be exactly invariant under it. 
In the extreme IR this implies invariance under the continuous group $Sp(2h,\R)$.
Defining (as it is usual) the fermi fields to be invariant
under $\Gamma$, the electro-magnetic field strengths\footnote{\ The dual field strengths $G_a$ are defined as $G_a=\ast\mspace{1mu} \partial\mathscr{L}/\partial F^a$.} $\cf\equiv (G_a, F^a)$ can enter the Pauli interactions only through the $\Gamma$-invariant combination $\ce^{-1}\cf$: these couplings
are automatically $Sp(2h,\R)$-invariant and hence cannot spoil the tension flow of the uni-modular matrix $\cs_{ab}^\prime$ which is a flow in the group $Sp(2h,\R)$. 
In particular, when $\mu$ is harmonic, hence tamed, the Pauli coupling must be a domestic
brane amplitude, as expected.

More generally, the tension field in the \textsc{rhs} of eqn.\eqref{flowa} is the only
$Sp(2h,\R)$-covariant tensor we can write to second order in the scalar field derivatives which flows 
 the  transverse complex structure of the Heisenberg fiber of $\mathscr{M}_3$; all possible covariant corrections (to this order)
will  affect the  fiber's transverse K\"ahler moduli but not its transverse  complex moduli.  

Therefore the condition that $\mu$ is harmonic seems to be a rather plausible ``naturalness'' 
requirement, at least under certain circumstances, especially when the   scalars' manifold $\mathscr{M}$ is an OV space. If the couplings flow in the extreme IR to a ``weak-coupling'' regime (meaning a regime in which the geometric interpretation of the couplings is sound), at the fixed point the gauge coupling $\mu$ should be tensionless, i.e.\! $\mu$ harmonic.  
  As already stressed, we know no reliable example of consistent effective theory where the tension of $\mu$ is not zero.

  \subsubsection{A word of caution}

In the heuristics above we neglected the issues related to finite-distance singularities in scalars' space.  At such points finitely many \emph{charged} states become
  light, and we have additional effective one-loop contributions from diagrams of the form
  \eqref{1loopx} where the internal lines correspond to the new light charged particles and
  the vertices $\circ$ are minimal gauge couplings. In terms of the dependence on the effective
  cut-off $\Lambda_\text{eff}$, these contributions are suppressed by a factor $\Lambda_\text{eff}^{-2}\log \Lambda_\text{eff}$, but there is no reason to expect that -- as functions of the background value of the scalars -- they are proportional to the tension of the gauge couplings $\mu$.
 This suggests that near the finite-distance singular points in $\mathscr{M}$ the effective tension $(\Lambda_\text{eff}/F_\pi)^2\;T(\mu)_{ab}$, or more precisely its invariant norm $(\Lambda_\text{eff}/F_\pi)^2\,\|T(\mu)_{ab}\|$, is small but non-zero. Morally speaking, we expect the gauge couplings $\tau(\phi)_{ab}$ to satisfy  some kind of ``wave-equation with sources'' on $\mathscr{M}$ roughly of the form
 \be
 T(\mu)_{ab}\equiv D^i\partial_i\tau(\phi)_{ab}= \sum_s f_s(\phi;\phi_s)
 \ee  
 where $\phi_s\in\mathscr{M}$ are the finite-distance singular points at which finitely-many
 new degrees of freedom become massless, and the $f_s(\phi;\phi_s)$'s are functions (or distributions) with support in some small region centered at $\phi_s$
 which  
capture the local physics at these special points in moduli space.

\subsection{
Existence, uniqueness, and structure of tamed maps}\label{s:stru1}

In an $\cn=2$ supergravity consistent with the swampland conjectures,
the gauge coupling map
\be\label{hasq1a}
\mu\colon \mathscr{M}\to \Lambda\backslash Sp(2h,\R)/U(h),\qquad \Lambda\subseteq Sp(2h,\Z),
\ee
has a very restricted form \cite{swampIII} as a consequence of the structure theorem
for the underlying VHS period map \cite{reva,revb,periods}. It is natural to ask whether
this fundamental result (and its physical consequences \cite{swampIII})
holds in the more general set-up of domestic geometry, that is, for finite-energy
tamed maps $\mu$ as in eqn.\eqref{hasq1a} whose source $\mathscr{M}$ is an OV manifold.
\medskip

The aim of this subsection is to argue for a positive answer.
For didactical reasons we split the discussion in two parts: first we consider the 
well studied case where the source of the map $\mu$ is a \emph{compact}
Riemannian manifold (whose holonomy algebra may or may not be special). This situation 
has little interest for physics, however it sheds light on the basic structures of domestic
geometry. Then we proceed to the physically relevant set-up of non-compact OV manifolds,
and argue that the good properties of the compact case extends to domestic geometry defined 
over the ``magic''  OV spaces.
\medskip 
  
We apologize for being technical if not pedantic. 
The reader may prefer to jump directly to \S.\,\ref{s:stru2} in a first reading.
However some intermediate result may be of independent interest.

\subsubsection{Compact source space}\label{caseOCC:scc}

The argument goes through several steps. 
First we show that twisted tamed maps \eqref{hasq1a}
with the required properties exist under the assumptions:
\begin{itemize}
\item[\textit{(i)}] the source manifold is compact;
\item[\textit{(ii)}] the monodromy representation has the physically expected properties.
\end{itemize}
Since a tamed map is in particular harmonic, this step requires first to
show that twisted harmonic maps exist, and then that they are tamed.
The second step concerns the 
uniqueness properties of twisted tamed maps with a given
monodromy representation. In the third step we use this information to
get the structural factorization of the domestic geometry.
This last property coincides with the VHS structure theorem whenever the
source manifold is K\"ahler.

\subparagraph{Existence of harmonic maps.}
The most classical existence theorem for harmonic maps is due to Eells and Sampson \cite{eells}; their strategy was to show
 existence and regularity of the tension flow \eqref{torflowa}: let $M$ and $N$ be \emph{compact} Riemannian manifolds and $N$ have non-positive sectional curvatures. Then every smooth map $f\colon M\to N$ is homotopic to a harmonic map. The harmonic map is essentially unique in its homotopy class\footnote{\ For the precise uniqueness statement see e.g.\! \cite{AL2}.} and it is the map of minimal energy in its class.
 
For the gauge coupling map 
\eqref{hasq1a}
the condition on the sectional curvatures of the target space
is satisfied, but  both the source and the target manifolds are non-compact. 
However they are expected to behave ``almost as they were compact'' because they have finite volume and enjoy other good properties. Thanks to these special properties, we may invoke other, more powerful, existence theorems.
\medskip

Our target space is locally symmetric; in this situation one has: 

\begin{thm}[Corlette \cite{corlette1}]
Let $G$ be a real Lie group, $K\subset G$ a maximal compact subgroup, and
$\Lambda\subset G$ any discrete subgroup. Assume the Riemann manifold $M$ is \emph{compact}. A harmonic map 
\be\label{hharmm}
\phi\colon M\to \Lambda\backslash G/K
\ee 
 exists if and only if the monodromy group
$\rho(\pi_1(M))\equiv\Gamma\subseteq \Lambda$ has reductive Zariski closure\footnote{\ We see the Lie group as an algebraic group over $\R$ through its adjoint representation.} $\overline{\Gamma}^{\mspace{2mu}\R}$ in $G$.
\end{thm}

\medskip

In our problem, eqn.\eqref{hasq1a}, the target space has the required form with
\be
G= Sp(2h,\R),\quad K=U(h),
\ee 
 while we may take $\Lambda\equiv \Gamma$ and neat with no loss.
 In the applications we have in mind
 the real Lie group\footnote{\ More correctly: the real Lie group underlying the $\R$-algebraic group $\overline{\Gamma}^{\mspace{2mu}\R}$ \cite{milneG}.} $G^\prime \equiv\overline{\Gamma}^{\mspace{2mu}\R}\subset G$ -- if non-trivial -- is semi-simple and non-compact: these properties follow from the swampland conditions on the monodromy.\footnote{\ By construction, in the relevant applications the group $\Gamma$ is a finite-index, neat, normal subgroup of the physical monodromy group.}
In view of the above \textbf{Theorem}, these ``swampy'' properties of $\Gamma$ guarantee the \emph{existence} of a twisted harmonic map $\phi$. 

The maps of domestic geometry are not just harmonic, they are tamed.
We turn to the question of \emph{tameness} of the harmonic map $\phi$.

\subparagraph{The map $\phi$ is automatically tamed.}
When the compact manifold $M$ has special holonomy 
$\mathfrak{hol}(M)\neq\mathfrak{so}(m)$ -- and hence a non-trivial algebra
 of parallel forms $\cp^\bullet$ --
 the twisted harmonic map $\phi$ is automatically tamed, that is,
\be\label{ttamea}
D\ast(d\phi\wedge\Omega)=0\quad\text{for all}\quad \Omega\in\cp^\bullet
\ee
(see \cite{corlette2}, \textbf{Theorem 3.1}). In particular
\begin{itemize}
\item[(a)] ${M}$ is K\"ahler $\Rightarrow$ $\phi$ is pluriharmonic
$\Rightarrow$ $\phi$ defines a $tt^*$ geometry (cfr.\! \S.\eqref{s:tt*});
\item[(b)] ${M}$ quaternionic K\"ahler of dimension $\geq8$, or symmetric not of the form $SO(n,1)/SO(n)$ or $SU(m,1)/U(m)$ $\Rightarrow$ $\phi$ is totally geodesic,
so its image is either a point or a locally symmetric space of the form $\Gamma^\prime\backslash L/K^\prime$ for $L \subset G$ a Lie subgroup \cite{helga}.
\end{itemize}

The proof of eqn.\eqref{ttamea} is \emph{via} a Bochner argument \cite{corlette2}
which generalizes the one based on eqn.\eqref{hasqwert}
which was valid under the assumption that $M$ is K\"ahler. 
 If $\phi$ is harmonic and the target space has non-positive curvature operators
(as in our case) one has an identity of the form \cite{corlette2}
\be\label{weitza}
\begin{split}
&d \Big(\!\ast \big(d\phi^i \wedge \Omega\big)\wedge D\ast \big(d\phi^j\wedge \ast\mspace{3mu} \Omega\big)\,g_{ij}\Big)=\\
&\quad=
(-1)^{m-1}\Big( |D\ast (d\phi \wedge \Omega)|^2 +\text{non-negative}\Big) d\mspace{1.2mu}\mathsf{vol}.
\end{split}
\ee
For $M$ compact the integral of the \textsc{lhs} vanishes, so $D\ast (d\phi \wedge \Omega)=0$ and 
the map is tamed.

{\subparagraph{Uniqueness properties I.}
One expects that the very same ``swampy'' conditions on $\Gamma$
which guarantee the existence of the tamed map $\phi$
also play the crucial role in the question of \emph{essential uniqueness} of the tamed map $\phi$ in its homotopy class (i.e.\! in the family of twisted maps with the given monodromy representation $\rho$). Here the main theme is to qualify
the adjective ``essential'', that is, to specify under which equivalence relation(s) all homotopic 
tamed maps get identified: the coarser the relation, the weaker the uniqueness property.
The weakest result
is:
\begin{quote}\it The ``swampy'' properties of $\Gamma$ imply that the 
space of tamed maps in the class of $\phi$ is homotopic to a single point.\footnote{\ \label{morse}
We sketch the argument.
Let $N$ be any Riemannian manifold (not necessarily compact) with $\pi_1(N)\equiv \pi$
and let
$\mathfrak{M}(N,\Gamma\backslash G/K)_f$ be the space of $C^2$-maps $N\to \Gamma\backslash G/K$ homotopic to the map $f$ ($\mathfrak{M}(N,\Lambda\backslash G/K)_f$ is endowed with the
$C^2$-topology). By Gottlieb lemma \cite{Got1,Got2} $\mathfrak{M}(N,\Gamma\backslash G/K)_f$ is a $K(C_{\pi,f},1)$
space with $C_{\pi,f}$
 the centralizer of $f_*(\pi)$ in $\Gamma$.  Under our ``swampy'' assumptions $C_{\pi,f}$ is trivial,
and $\mathfrak{M}(N,\Gamma\backslash G/K)_f$ has the homotopy type of a point.
We can compute the homotopy type of $\mathfrak{M}(N,\Gamma\backslash G/K)_f$ by Morse cobordisim
applied to the gradient flow of the energy $E(\phi)$ (i.e.\! to the tension flow). Under the present assumptions the energy $E(\phi)$ is a ``perfect Morse function'' in the sense that its Hessian at a critical point is a non-negative operator: see e.g.\! \textbf{Corollary 9.2.2} of \cite{JJost}. Then we conclude that the space of harmonic functions homotopic to $f$ is contractible.}
\footnote{\ If the condition of non-positive sectional curvatures is replaced by the stronger one of \emph{strictly negative} sectional curvatures, the space of tamed maps is a single point and not just
homotopic to a point \cite{AL2}. In the physical set-up this applies only when we have a single light photon i.e.\! $h=1$.}
\end{quote}}
\medskip

We look for a much stronger result: uniqueness up to the physically natural 
equivalence relation (see \textbf{Definition \ref{uuniess}} below). 
 We start by playing the devil's advocate, and see
 what  happens when the harmonic map is \emph{not} unique.
The following argument is modelled on the classical papers \cite{AL1,AL2};
in facts it reproduces the main points in the proof of \textbf{Theorem 9.7.2} 
of ref.\!\!\cite{JJost} whose hypotheses are identical to ours:
 \textit{(i)} compact source space,
 and \textit{(ii)} target space a complete Riemannian manifold with non-positive sectional curvatures. 
 (We go through the proof because its single steps are more useful for our purposes than the theorem itself).

Let $\phi_0$, $\phi_1$ be two distinct homotopic harmonic maps; one can construct a homotopy
\be\label{terqww}
\phi(x,t)\colon M\times [0,1]\to \Gamma\backslash G/K,\qquad
\left[\begin{aligned}
\phi(x,0)&=\phi_0(x)\\
\phi(x,1)&=\phi_1(x)
\end{aligned}\right. 
\ee
which is \emph{geodesic,} i.e.\! for fixed $x\in M$ the map 
\be
\phi(x,-)\colon [0,1]\to \Gamma\backslash G/K
\ee
is a geodesic. We write $E(t)$ for the energy of the map $\phi(-,t)\colon M\to \Gamma\backslash G/K$.
 A simple computation yields\footnote{\ \label{iiiiiuy} See \textbf{Corollary 9.2.1} in \cite{JJost}.
 For later reference we stress that  equation \eqref{forrrm} holds under the assumption that $\phi(t)$
 is a geodesic family of finite-energy maps (not necessarily harmonic) \emph{independently}
 of the assumption that the source space $M$ is compact (cfr.\! ref.\!\cite{JJost}).}
\be\label{forrrm}
\frac{\partial^2 E(t)}{\partial t^2}=\int_M d\,\mathsf{vol}\Big(\|\nabla_{\partial/\partial t} d\phi(t)\|^2-
g^{\alpha\beta} G_{ij}\, R_{ijkl}\,\partial_t\phi^i\partial_\alpha\phi^j \partial_t \phi^k\partial_\beta\phi^l\Big)
\ee
which, together with the vanishing of the first variation at a harmonic map
\be
\frac{\partial E(t)}{\partial t}\Big|_{t=0}=0,
\ee
yields 
\be\label{poiqw12}
\begin{split}
&E(\phi(t))-E(\phi_0)=\\
&=\int_0^t ds\int_0^s du 
\int_{M}d\,\mathsf{vol}\Big(\|\nabla_{\partial/\partial u} d\phi(u)\|^2-
g^{\alpha\beta} G_{ij}\, R_{ijkl}\,\partial_u\phi^i\partial_\alpha\phi^j \partial_u \phi^k\partial_\beta\phi^l\Big)\geq0,
\end{split}
\ee
where the inequality holds because the target space sectional curvatures are non-positive. Setting $t=1$
we get $E(\phi_1)\geq E(\phi_0)$;
inverting the r\^ole of the two maps we get $E(\phi_1)=E(\phi_0)$, so \emph{all} homotopic harmonic maps have the same energy $E(\phi_0)$, and hence they all realize the \emph{absolute minimum} of the energy in their class. Then each of the two terms inside the inner integral in the \textsc{rhs} of \eqref{poiqw12} is point-wise zero for all $u$'s. We conclude that 
\be
E(\phi(t))-E(\phi_0)=0\quad\text{for all }t,
\ee
and thus all maps $\phi(t)$ ($t\in[0,1]$) have the same minimal value $E(\phi_0)$ of the energy.
From the variational charaterization of harmonic maps we get:
\begin{pro}[\textbf{Theorem 9.7.2} \label{pppouer}
of ref.\!\cite{JJost}] The maps $\phi(-,t)\colon M\to \Gamma\backslash G/K$
form a one-parameter family of homotopic harmonic (hence tamed) maps.
\end{pro}

A harmonic map $\phi$ is a stationary classical soliton of
 the  $\sigma$-model with target $\Gamma\backslash G/K$
 defined in the ``space-time'' $\R_\text{time}\times M$
whose mass is $E(\phi)$.
\textbf{Proposition \ref{pppouer}} says that whenever the twisted tamed map $\phi$ is not unique in its class,
 we have a continuous family of $\sigma$-model solitons
  which are degenerated in mass and carry the same topological charges.
As physicists we do not expect such a huge degeneracy, unless it is a consequence of a symmetry of the $\sigma$-model. 
We make this physical intuition into a math definition: 
\begin{defn}\label{uuniess} We say that the tamed map $\phi\colon M\to\Gamma\backslash G/K$
is \emph{essentially unique} in its homotopy class
 iff all tamed maps homotopic to it are obtained from $\phi$
by the action of a symmetry of the associated $\sigma$-model. 
\end{defn}

We expect the tamed map to be essentially unique in this precise sense.
This physical intuition turns out to be fully correct: see \textbf{Lemma \ref{LLLem}} below.
As a preliminary step we need to discuss the way symmetry acts.

 \subparagraph{Action of symmetry.} The subgroup of the $\sigma$-model symmetry  which leaves invariant the monodromy representation $\rho$ -- that is, which commutes with the
 topological charges --
  is the centralizer $C\subset G$ of $\Gamma$. $C$ is a real Lie group\footnote{\ Indeed an algebraic group over $\R$.}
identified with the centralizer of $G^\prime\subset G$
(where, as before, $G^\prime\equiv \overline{\Gamma}^{\,\R}$). When $G^\prime$ is reductive (as required for $\phi$ to exists) but not necessarily semi-simple
\be
G^\prime = Z(G^\prime)\times G^\prime_{ss},\qquad G^\prime_{ss}\ \ \text{semi-simple}
\ee
 $C$ is the center $Z(G^\prime)$ of $G^\prime$ times the Lie group $G^{\prime\prime}$
 such that $G^\prime\times G^{\prime\prime}$ is a maximal subgroup of $G$: 
  \be
C=Z(G^\prime)\times G^{\prime\prime},\qquad  G^\prime\times G^{\prime\prime}\hookrightarrow
G\ \text{maximal}.
 \ee 
 Let $K^\prime\subset G^\prime$, $K^{\prime\prime}\subset G^{\prime\prime}$ be maximal compact subgroups, and $K\subset G$ a maximal compact subgroup containing $K^\prime\times K^{\prime\prime}$.
 We have a chain of totally geodesic embeddings
 \be\label{poiunz}
 \Gamma\backslash G^\prime/K^\prime \xrightarrow{\ \;\iota_2\ } \Gamma\backslash G^\prime/K^\prime\;
 \text{\begin{LARGE}$\times$\end{LARGE}}\; G^{\prime\prime}/K^{\prime\prime}
 \xrightarrow{\ \;\iota_1\ } \Gamma\backslash G/K.
 \ee
Since $\Gamma$ is neat, by the Cartan-Hadamard theorem all three locally symmetric manifolds in
eqn.\eqref{poiunz}
are $K(\Gamma,1)$ spaces; then by the Whitehead theorem \cite{altop} there exists 
a chain of deformation retractions 
\be
\Gamma\backslash G/K\xrightarrow{\ \;r_1\ }
\Gamma\backslash G^\prime/K^\prime\;
 \text{\begin{LARGE}$\times$\end{LARGE}}\; G^{\prime\prime}/K^{\prime\prime}
 \xrightarrow{\ \;r_2\ }
\Gamma\backslash G^\prime/K^\prime.
\ee
Let
\be\label{arg1}
\phi\colon M\to \Gamma\backslash G/K
\ee
be a twisted tamed map with monodromy representation $\rho$.
We consider the two maps
\be
f_1\equiv r_1\circ \phi\colon M\to \Gamma\backslash G^\prime/K^\prime\;
 \text{\begin{LARGE}$\times$\end{LARGE}}\; G^{\prime\prime}/K^{\prime\prime},
 \qquad f_2\equiv r_2\circ r_1\circ \phi\colon M\to  \Gamma\backslash G^\prime/K^\prime,
\ee 
which (by construction) are twisted by the same $\rho$. By \textbf{Theorem} \eqref{hharmm} there exist tamed maps
$\phi_1$ and $\phi_2$ homotopic (respectively) to $f_1$, $f_2$. Since $\iota_1$, $\iota_2$ are totally geodesic,
we obtain three tamed maps $M\to \Gamma\backslash G/K$ twisted by the same $\rho$
\be
\phi,\qquad \iota_1\phi_1,\qquad \iota_1\iota_2\phi_2,
\ee
where the second (third) is a tamed map whose image is fully contained in the submanifold
 $\Gamma\backslash G^\prime/K^\prime\;
\times\; G^{\prime\prime}/K^{\prime\prime}$ (resp.\! $\Gamma\backslash G^\prime/K^\prime$)
of $\Gamma\backslash G/K$. By its very definition the third tamed map factorizes as
\be\label{iiuyqq}
M\xrightarrow{\ \iota_2\phi_2\ } \Gamma\backslash G^\prime/K^\prime\;
 \text{\begin{LARGE}$\times$\end{LARGE}}\; G^{\prime\prime}/K^{\prime\prime}\xrightarrow{\ \iota_1\ } \Gamma\backslash G/K,
\ee
while the projection of the first arrow $\iota_2\phi_2$ in the factor space $G^{\prime\prime}/K^{\prime\prime}$
is a constant map. The other projection is just the tamed map
\be
M\xrightarrow {\ \phi_2\ } G^\prime/K^\prime \equiv Z(G^\prime)\big/[K^\prime \cap Z(G^\prime)]
 \text{\begin{LARGE}$\times$\end{LARGE}}\, G^\prime_{ss}\big/[K^\prime\cap G^\prime_{ss}]
\ee 
which is identified with a pair of solitons for the two $\sigma$-models with respective target spaces the two
factors in the \textsc{rhs}. The first factor space is  a locally flat Abelian group $A$. The ``swampy'' conditions on $\Gamma$ say that this Abelian factor is absent.
For $\Gamma$ just reductive, the $\sigma$-model soliton \eqref{iiuyqq} decomposes in three
items: \textit{(i)} a constant
map into $G^{\prime\prime}/K^{\prime\prime}$, \textit{(ii)} a soliton of the
$\sigma$-model with target $A$, and \textit{(iii)} a soliton of the $\sigma$-model 
with target the semi-simple part of the double coset $\Gamma\backslash G^\prime/K^\prime$.

The symmetry $Z(G^\prime)\times G^{\prime\prime}$ acts
on the $\sigma$-model soliton \eqref{iiuyqq} by translations of the image of
the constant map in $G^{\prime\prime}/K^{\prime\prime}$ and shift symmetries 
of the Abelian soliton.\footnote{\
The fundamental group of the family of solitons produced by acting on \eqref{iiuyqq}
with the symmetry group is $\Gamma\cap Z(G^\prime)$, so that this family
is homotopic to the space of all solitons with target space $\Gamma\backslash G/K$, cfr.\!
footnote \ref{morse}.} In the ``swampy'' case the Abelian sector is absent,
and the symmetry produces out of the solution \eqref{iiuyqq}
a $G^{\prime\prime}/K^{\prime\prime}$-family of harmonic solitons. \textbf{Proposition \ref{pppouer}}
requires 
this family to decompose into geodesic sub-families; they have the form
\be
M\times [0,1]\xrightarrow{\ \;\phi_2\,\times\, \lambda\ \;} \Gamma\backslash G^\prime/K^\prime
\text{\begin{LARGE}$\times$\end{LARGE}}\, G^{\prime\prime}/K^{\prime\prime}\xrightarrow{\ \;\iota_1\ }
\Gamma\backslash G/K
\ee
with $\lambda\colon[0,1]\to G^{\prime\prime}/K^{\prime\prime}$ a geodesic arc connecting two points in $G^{\prime\prime}/K^{\prime\prime}$.

 \subparagraph{Uniqueness properties II.}
 We have
 \begin{lem}\label{LLLem}$M$ compact. Suppose $\Gamma$ satisfies the ``swampy'' conditions, i.e.\!
 $G^\prime\equiv\overline{\Gamma}^{\,\R}$ is a non-compact semi-simple algebraic subgroup of $G$. Let 
 \be
 \phi_0,\phi_1\colon M\to \Gamma \backslash G/K
 \ee
 be two twisted harmonic (hence tamed) maps in the same homotopy class ($\equiv$
 twisted by the same monodromy representation $\rho$). Then there is an element
 $g\in G^{\prime\prime}$ of the $\sigma$-model symmetry group such that
 \be
 \phi_1= g\cdot \phi_0.
 \ee
 \end{lem}
 
Since we have already proven that in the homotopy class of $\phi_0,\phi_1$
 there is at least one $G^{\prime\prime}/K^{\prime\prime}$-family
of tamed maps generated by a factorized tamed map 
\be
M\xrightarrow{\ \;\phi_\ast\ } \Gamma\backslash G^\prime/K^\prime
\xrightarrow{\ \;\iota_1\iota_2\ } \Gamma\backslash G/K,
\ee
while this family form a \emph{full} $G^{\prime\prime}$-orbit, we conclude
\begin{corl} Under the assumptions of the {\rm\textbf{Lemma},}
the set of tamed maps 
\be
\phi\colon M\to \Gamma\backslash G/K
\ee
twisted by $\rho$ is given by 
a \emph{unique} $G^{\prime\prime}/K^{\prime\prime}$-family of  maps
factorized as in eqn.\eqref{iiuyqq}.
In particular $\phi$ is \emph{essentially unique} in its homotopy class in the sense of
{\bf Definition \ref{uuniess}}. 
\end{corl}

\begin{proof}[Proof of the {\bf Lemma}]
Preliminary we recall from the previous sections some facts about the non-compact
symmetric space $G/K$. Let $\mathfrak{g}=\mathfrak{k}\oplus \mathfrak{p}$
be the orthogonal decomposition of the Lie algebra of $G$ (with $\mathfrak{k}\equiv \mathfrak{Lie}(K)$).
The tangent bundle $T(G/K)$ is the homogeneous bundle associated to
the $K$-module $\mathfrak{p}$. We choose a global section $s\colon G/K \to G$
(it exists since $G/K$ is contractible). The Levi-Civita connection of $G/K$
(in the trivialization of $T(G/K)$ defined by the chosen $s$) is
\be\label{exp1}
\nabla=d+(s^{-1}ds)_{\mathfrak{k}}
\ee 
where as before the subscript means orthogonal projection to the subspace
$\mathfrak{k}\subset\mathfrak{g}$. Then
\be\label{exp2}
-R_{ijkl}\mspace{1mu} v^i w^j v^k w^l= \Big\| \big[\iota_v(s^{-1}ds)_{\mathfrak{p}}, \iota_w(s^{-1}ds)_{\mathfrak{p}}\big]\Big\|^2\geq0.
\ee

Now let $\phi_0,\phi_1\colon M\to \Gamma\backslash G/K$ be 
two tamed maps twisted by the same $\rho$. We write $\phi_t\equiv \phi(-,t)$ 
for the geodesic family of interpolating tamed maps in {\bf Proposition
\ref{pppouer}}, parametrized by $t\in[0,1]$, and consider its lift to the universal cover $\tilde M$ of $M$
\be
\tilde\phi_t\colon \tilde M\times [0,1]\to G/K.
\ee
$\tilde\phi_t$ is $\rho$-twisted, that is, if $\pi_1(M)\ni \xi\colon \tilde M\to \tilde M$ is an element of the deck group
of the cover $\tilde M\to M$, we have
\be
\xi^*\mspace{2mu}\tilde\phi_t= \rho(\xi)\cdot \tilde\phi_t.
\ee
For each fixed $t\in [0,1]$ we define the $\rho$-twisted map 
\begin{gather}
\Phi_t\equiv s\circ \tilde\phi_t\colon \tilde M\to G \\
\xi^*\mspace{2mu}\Phi_t= \rho(\xi)\cdot \Phi_t.\label{rrhtwis}
\end{gather}

We study how the $\rho$-twisted map
$\Phi_t$ varies under an infinitesimal deformation of the parameter $t$.
The infinitesimal deformation is given by the vector field
\be\label{uuuytqw}
\delta_t\equiv (\Phi^{-1}_t\partial_t \Phi_t)_{\mathfrak{p}}
\ee
which is a section of the pulled back tangent bundle
\be
\tilde\phi_t^*\mspace{2mu} T(G/K)\to \tilde M
\ee
equipped with the natural pulled back Levi-Civita connection on $G/K$. 

We know from the proof of {\bf Proposition
\ref{pppouer}} that for all $t\in[0,1]$ each of the two non-negative
terms in the integrand in the \textsc{rhs} of \eqref{poiqw12}
vanishes point-wise for all $t$. Using the explicit expressions
for the connection and curvature on $G/K$, eqns.\eqref{exp1} and \eqref{exp2},
 these two
conditions become
\be
d\delta_t+\big[(\Phi_t^{-1}d\Phi_t)_{\mathfrak{k}},\delta_t\big]=
\Big\|\big[(\Phi_t^{-1}d\Phi_t)_{\mathfrak{p}},\delta_t\big]\Big\|^2=0,
\ee
which may be combined in the single equation
\be
d\delta_t+\big[\Phi_t^{-1}d\Phi_t,\delta_t\big]=0\quad\Rightarrow
\quad d\big(\Phi_t \delta_t \Phi_t^{-1}\big)=0,
\ee
whose solution is
\be\label{009aa1}
\Phi_t \delta_t \Phi_t^{-1}= v_t\in \mathfrak{g} \ \ \text{constant along $\tilde M$}.
\ee

Consistency of eqn.\eqref{009aa1}  with the
action of the deck group \eqref{rrhtwis}
requires that the element $v_t\in\mathfrak{g}$
belongs to the kernel of the adjoint action of $\Gamma\equiv \rho(\pi_1(M))$. By definition,
this kernel is nothing else than the Lie algebra $\mathfrak{g}^{\prime\prime}$ of
$G^{\prime\prime}\subset G$; then
\be
v_t\in\mathfrak{g}^{\prime\prime}\quad\text{for all $t\in[0,1]$.}
\ee 
Comparing eqns.\eqref{uuuytqw} and 
\eqref{009aa1} we get
\be\label{jjjjas1v}
\begin{aligned}
\delta_t\equiv \Phi_t^{-1}v_t\mspace{1mu}\Phi_t\in\mathfrak{p},\quad&\Rightarrow\quad
\Big(\Phi_t^{-1}\partial_t \Phi_t -\Phi_t^{-1} v_t\mspace{1mu} \Phi_t\Big)_{\!\mathfrak{p}}=0\\
&\Rightarrow\quad \Phi_t^{-1}\partial_t \Phi_t -\Phi_t^{-1} v_t\mspace{1mu} \Phi_t\in \mathfrak{k}.
\end{aligned}
\ee
Therefore by a $K$-gauge transformation (i.e.\! by a $t$-dependent change of trivialization $s$ of the pulled back
tangent bundle) of the form
\be
\Phi_t\to \Phi_t\,U_t,\qquad U_t\in K,
\ee 
we may set to zero the expression inside the big parenthesis in eqn.\eqref{jjjjas1v}; then 
\be
\frac{\partial}{\partial t}\Phi_t= v_t\,\Phi_t,\quad\text{with}\quad v_t\in\mathfrak{g}^{\prime\prime}
\ \text{constant along $\tilde M$ (but depending on $t$)},
\ee
so that the infinitesimal deformation in the geodesic family of tamed maps
is produced by the action of an infinitesimal $G^{\prime\prime}$-symmetry.
Integrating in $t$ we get
\be
\Phi_1 = P\mspace{-1mu}\exp\!\left(\int_0^1 v_t\, dt\right)\cdot\Phi_0,\quad\text{with}\quad
P\mspace{-1mu}\exp\!\left(\int_0^1 v_t\, dt\right)\in G^{\prime\prime},
\ee 
which is equivalent to the \textbf{Lemma}.\end{proof}

\subparagraph{Structure and rigidity of tamed maps.}
The two main results we are after are now easy consequences of the \textbf{Corollary}.

 \begin{stru} $M$ \emph{compact.} Let $G^\prime\equiv \overline{\rho(\pi_1(M))}^{\mspace{1.5mu}\R}\subset G$ be the $\R$-Zariski closure of the monodromy group
 which we assume to be either trivial or semi-simple. Let $G^{\prime\prime}\subset G$
be the centralizer of $G^\prime$ in $G$, $K^\prime \subset G^\prime$, $K^{\prime\prime} \subset G^{\prime\prime}$ maximal compact subgroups, and $K\subset G$ a maximal compact subgroup
 containing $K^\prime\times K^{\prime\prime}$. 
Then the harmonic (in fact tamed) map $\phi\colon M\to \Lambda\backslash G/K$ (resp.\! the covering $\rho$-twisted tamed map $\tilde\phi$)
 factors as 
\be\label{oop}
\begin{aligned}
&\phi\colon &&M\xrightarrow{\ \phi^\prime\,\times\,\phi^{\prime\prime}\ } \Gamma\backslash G^\prime/K^\prime
\,\text{\begin{LARGE}$\times$\end{LARGE}}\, G^{\prime\prime}/K^{\prime\prime}\xrightarrow{\ \;\iota_1\ } \Gamma \backslash G/K,\\ 
\mathit{resp.}\quad &\tilde\phi\colon &&\widetilde{M}\xrightarrow{\ \tilde\phi^\prime\,\times\,\phi^{\prime\prime}\ } G^\prime/K^\prime\,\text{\begin{LARGE}$\times$\end{LARGE}}\, G^{\prime\prime}/K^{\prime\prime}\xrightarrow{\ \;\tilde\iota_1\ } G/K,
\end{aligned}
\ee
where the $\phi^{\prime\prime}\colon M\to G^{\prime\prime}/K^{\prime\prime}$ is a \emph{constant} map.
\end{stru}

Were it not for the assumption that $M$ is compact, the above statement would be identical in form
to the structure theorem 
for the period map $p$ in Hodge theory \cite{periods,reva,revb,MT4} which, in general, factorizes as in the commutative diagram
\be\label{permapa}
\begin{gathered}
\xymatrix{
\mathscr{M}\ar[r]\ar@/_6.3pc/[drr]^(0.6){\mu}\ar@/^2pc/[rr]^{p}\ar[dr]^(0.6){\phi^\prime}& \Gamma\backslash G^\prime/[H\cap G^\prime] \times G^{\prime\prime}/[H\cap G^{\prime\prime}]\ar@{->>}[d]\mspace{3mu} \ar@{^{(}->}[r] & \Gamma\backslash G/H\ar@{->>}[d]\\
&\Gamma\backslash G^\prime/[K\cap G^\prime]\times G^{\prime\prime}/[H\cap G^{\prime\prime}] \mspace{2mu}\ar@{^{(}->}[r] & \Gamma\backslash G/K}
\end{gathered}
\ee
where again the maps into the factor space $G^{\prime\prime}/[H\cap G^{\prime\prime}]$
are constant.
Here for a Hodge structure of \emph{odd} weight $k$ and Hodge numbers $\{h^{p,q}\}$
\cite{Gbook,periods,reva,revb}
\be\label{perdomains}
H=\prod_{p+q=k\atop
p<k/2}U(h^{p,q})\subset U(h)\equiv K,\quad \text{and}\quad G=Sp(2h,\R),\quad \text{with}\quad 2h=\sum_{p+q=k}h^{p,q}.
\ee 

The Hodge-theoretic period map $p$ satisfies, in addition to structural factorization, 
the infinitesimal period relations \eqref{pqw12bbb}  \cite{Gbook,periods,reva,revb} -- these relations are akin to the 
restrictions on the gauge coupling $\mu$ coming from supersymmetry: e.g.\! the $k=3$ Hodge structure of Calabi-Yau 3-folds are equivalent to the relations of special K\"ahler geometry in the sense of $\cn=2$ SUGRA \cite{cec,stro,swampIII}. Most of these relations follow from the condition that the gauge coupling map is tame (e.g.\! if $M$ is K\"ahler the map is pluriharmonic
and the domestic geometry reduces to $tt^*$), 
hence they are automatically satisfied in the present set-up.

\begin{rigi} Under the assumptions above, the tamed map
$\phi^\prime\colon M\to \Gamma\backslash G^\prime/K^\prime$ is \emph{rigid},
that is, unique in its homotopy class. 
\end{rigi}

\subsubsection{The case of OV manifolds}\label{caseOB:sb}

In the previous subsection we got all the desired properties of domestic geometry
 in case $M$ is compact,
which unfortunately is \emph{not} a natural property in our physical applications. 
Our next goal is to show that in the arguments of \S.\ref{caseOCC:scc}
we
 may drop the assumption that $M$ is compact and replace it by the two physically natural hypotheses:
 \begin{itemize}
\item[\textit{(i)}] the source space is an OV manifold $\mathscr{M}$;
\item[\textit{(ii)}] there exists a $\rho$-twisted smooth map $\tilde\phi_0$ which, when seen as a map $\phi_0\colon \mathscr{M}\to
\Gamma\backslash G/K$, has finite energy $E(\phi_0)<\infty$.
\end{itemize}
The fact that all good properties remain true shows that of OV manifolds are really ``magic''.
The rough idea is that a finite-energy harmonic map $\mu$ whose source is an OV space
$\mathscr{M}$ behaves similarly to a map $\phi$ with a compact source space 
because the map $\mu$ must be ``trivial at infinity'' in $\mathscr{M}$. Our task is to make this idea precise. We consider the various aspects (existence, tameness, uniqueness, structural factorization, and rigidity)
one by one.

\subparagraph{Existence.} In this paragraph $M$ is any complete Riemannian manifold,
compact or otherwise.
Under the hypothesis \textit{(ii)}, it makes perfect sense to talk of continuous deformations of
$\phi_0$ which decrease its energy, so the variational strategy for the existence/regularity problem
 is still meaningful: we may think of deforming
continuously the map until we reach the absolute minimum value of the energy in
the homotopy class defined by the monodromy representation $\rho$. The tension-flow is an efficient way of implementing the deformation
in the direction of steepest descent of energy, cfr.\! eqn.\eqref{torflowa}.
Pursuing this strategy one gets:\footnote{\ \textbf{Warning.} In \S.\,\ref{caseOB:sb} we adopt the terminology of Corlette \cite{corlette2}:
a $\rho$-equivariant map $f\colon \tilde M\to G/K$ is called a \emph{$\rho$-twisted map $M\to G/K$}
(instead of a $\rho$-twisted map $\tilde M\to G/K$ as it is more natural in the physical parlance).
A $\rho$-twisted map $M\to G/K$ can also be defined as a section of the bundle
$\tilde M\times_\rho G/K$ \cite{corlette2}. }

\begin{thm}[Corlette \cite{corlette2}]  Suppose $\rho\colon \pi_1(M)\to G$ is a homomorphism with Zariski dense image and there exists a $\rho$-twisted map $\mu$ from $M$ to $G/K$ \emph{with finite energy.} Then there is a $\rho$-twisted harmonic map with finite energy from $M$ to $G/K$. 
\end{thm} 

\begin{rem} More generally, we may consider the case where the image $\Gamma\subset G$ of $\rho$
is Zariski dense in a non-compact, semi-simple subgroup $G^\prime\subset G$ (as in \S.\ref{caseOCC:scc} with ``swampy'' $\Gamma$). Then the above \textbf{Theorem}
shows the existence of a finite-energy, $\rho$-twisted harmonic map $\mu$
which factorizes through a $\rho$-twisted harmonic map $\mu^\prime$ 
as in the commutative diagram
\be\label{09nnnaq}
\begin{gathered}
\xymatrix{M \ar@/^2pc/[rr]^\mu \ar[r]_(0.33){\mu^\prime} & \Gamma\backslash G^\prime/K^\prime\;
\ar@{>>->}[r]_(0.54){\iota} & \Gamma\backslash G/K}
\end{gathered}\qquad \iota:\ \begin{smallmatrix}\text{\bf totally geodesic embedding}\end{smallmatrix}.  
\ee
We stress that when $G^\prime\subsetneq G$ the above \textbf{Theorem}, by itself,
 says nothing
about the possible existence of \emph{non}-factorized harmonic maps, i.e.\! it refers only to the existence of harmonic maps enjoying the structural factorization \eqref{09nnnaq}. Below we shall see that all
finite-action $\rho$-twisted harmonic maps are factorized as in \eqref{09nnnaq}.
\end{rem}

Roughly speaking, finite energy corresponds to finite volume of the image $\mu(\mathscr{M})$:
thus if the scalars' manifold $\mathscr{M}$ satisfies (our slightly stronger version of)
the standard swampland conjectures,
\emph{any} gauge-coupling map $\mu\colon \mathscr{M}\to \Gamma\backslash Sp(2h,\Z)/U(g)$
 is homotopic to a harmonic one, namely the fixed point $\mu_\text{harm}$ of the tension flow with initial condition $\mu$.
 The heuristic physical arguments of \S.\ref{s:tension} suggest that, in the extreme IR limit of a consistent quantum gravity, the physical coupling $\mu_\text{phys}$ actually coincides with $\mu_\text{harm}$ at least approximately and away from finite-distance singularities. 
 We stress that this statement is literally  true in all known examples
 of reliable quantum-consistent effective theories of gravity. 
 
 \medskip

\subparagraph{Tameness.} The Bochner argument around eqn.\eqref{weitza} still works
in the non-compact case \emph{provided}
we can show that the surface term in the integration over $\mathscr{M}$ of the \textsc{lhs} of
 \eqref{weitza} vanishes; in this case
 we conclude that the finite-action $\rho$-twisted harmonic map $\mu$ is actually tamed. 
 Thus to show the \textbf{Tamed property} for OV manifolds stated at the beginning of
 \S.\,\ref{jjjjasz},
 we have only to justify the dropping of the boundary term in the integration by parts of the
 \textsc{lhs} in
 eqn.\eqref{weitza} under our two assumptions that $\mathscr{M}$ is OV and $\mu$ has finite energy. 
 We defer this technicality to
appendix \ref{tamhol}.

\subparagraph{Essential uniqueness, structural factorization, and rigidity.}
As explained in the footnote \ref{iiiiiuy}, the crucial equation \eqref{forrrm} holds for any geodesic family of finite-energy maps whether the manifold $\mathscr{M}$ is compact or not. The only way
compactness enters in the game is that it guarantees the finite-energy condition for all smooth maps,
whereas in the non-compact case one should add the finite-energy condition as an independent
hypothesis and hence the results apply only to a small sub-class of harmonic maps. 
Therefore, when the finite-energy condition is satisfied, all formal consequences of \eqref{forrrm}
follow. In particular all finite-energy tamed maps, if non-rigid, belong to
one-parameter families of maps with the same energy. In turn this leads
to essential uniqueness of the finite-energy tamed maps (in their homotopy class)
 in the sense of \textbf{Definition \ref{uuniess}}.
In particular, the argument at the end of \S.\ref{caseOCC:scc}
yield
\begin{stru}
All finite-energy harmonic (hence tamed) $\rho$-twisted maps 
\be
\mu\colon\mathscr{M}\to \Gamma\backslash G/K
\ee
factorize as in eqn.\eqref{09nnnaq}. They form a $G^{\prime\prime}/K^{\prime\prime}$-family
on which the $\sigma$-model symmetry group $G^{\prime\prime}$ acts transitively. \emph{(Concretely $G^{\prime\prime}$
acts by varying the geodesic embedding $\iota$ in \eqref{09nnnaq}).}
\end{stru} 

The same argument shows that the factor map $\mu^\prime$
in \eqref{09nnnaq} is rigid.
\medskip

Nothing is said about the harmonic/tamed maps of \emph{infinite} energy
(which are expected to be the large majority).

\subsubsection{Structure of the gauge coupling $\mu$}\label{s:stru2}

We have seen above that in (arithmetic) domestic geometry the crucial    
\textbf{structural factorization} holds for a gauge coupling $\mu$.
Here we write this property in a more detailed (and convenient) form.
To simplify the formulae, we omit writing the boring factor space
$G^{\prime\prime}/K^{\prime\prime}$ and the trivial constant map into it.\footnote{\ For instance, in $\cn=2$ supergravity we write the gauge coupling as a function of the vector-multiplet scalars,
instead of a function of all scalars which is constant in the hypermultiplet scalars.}
With this convention, $\mu$
 factorizes as in the commutative diagram\vskip-3pt
\be\label{oop2a}
\xymatrix{\mathscr{M}\ar[r]_(0.27){\phi^\prime}\ar@/^2.5pc/[rr]^\mu & \Gamma\backslash \overline{\Gamma}^{\mspace{1mu}\R}/[ \overline{\Gamma}^{\mspace{1mu}\R}\cap U(h)] \mspace{4mu}\ar@{^{(}->}[r]_(0.475){\iota} & {\Gamma}\backslash Sp(2h,\R)/U(h)}
\ee
If $\mu$ is harmonic, the real Lie group (or, rather, algebraic group over $\R$) $\overline{\Gamma}^{\mspace{2mu}\R}\subset Sp(2h,\R)$
must be reductive. Hence, modulo finite groups, it has the form
\be
\overline{\Gamma}^{\mspace{1mu}\R}\cong A\times G_1\times \cdots \times G_s
\ee
with $A$ Abelian and $G_\ell$ simple. Correspondingly (up to commensurability \cite{morris})
\be
\Gamma\cong \Gamma_A\times \Gamma_1\times\cdots\times \Gamma_s
\ee with $\Gamma_A\subset A$ and
$\Gamma_\ell\subset G_\ell$. Since we are assuming $\Gamma$ to be neat and generated by unipotents, $\Gamma_A$ must be trivial. 
Then the real Lie group $\overline{\Gamma}^{\mspace{1mu}\R}\subset Sp(2h,\R)$ is either trivial or semi-simple.
In the fist case the gauge couplings $\mu$ are field-independent numerical constants. An example of this situation is given by 
the compactification of Type IIB on a rigid Calabi-Yau \cite{Cecotti:2018ufg}. 
 
\medskip

If $\Gamma$ is not trivial, the harmonic map $\phi^\prime$ decomposes into a $s$-tuple of partial maps
\be
 \phi_\ell\colon \cm\to \Gamma_\ell\backslash G_\ell/[G_\ell\cap U(h)],\quad \ell=1,\cdots,s.
\ee
We stress that all spaces through which the gauge coupling $\mu$ factorizes, i.e.\! $ \Gamma\backslash \overline{\Gamma}^{\mspace{1mu}\R}/[K\cap \overline{\Gamma}^{\mspace{1mu}\R}]$ and the $\Gamma_\ell\backslash G_\ell/[G_\ell\cap K]$ are OV manifolds.

\medskip

As already mentioned, the ``structural factorization'' \eqref{oop2a} is identical in form\footnote{\ The Griffiths period map satisfies in addition
the IPR, so we have the more detailed factorization of $\mu$
in the diagram \eqref{permapa}. } to the structure theorem for Griffiths period maps in modern Hodge theory \cite{reva,revb,MT4,periods},
which is satisfied by the low-energy effective theories of Type II compactified on a geometric family of Calabi-Yau,
whose couplings are determined by the Griffiths period map \cite{cec,stro}, and which is the main condition discriminating the 
quantum consistent $\cn=2$ supergravities from the ones belonging to the swampland \cite{swampIII}. In particular, whenever the
$\cn=2$ SUGRA satisfies the structure theorem automatically satisfies all the relevant swampland conjectures, see \cite{swampIII}.  
\medskip

It is remarkable that the very same structural properties hold in full generality -- even when the moduli space $\mathscr{M}$ has general holonomy $\mathfrak{so}(m)$ and \emph{no} natural complex structure -- by virtue of the ``magical'' properties of the OV manifolds, provided we add to the list of the swampland conjectures the statement that the IR gauge coupling $\mu$ is harmonic of finite-energy. 
\medskip

We take the above state of affairs as evidence that our working hypothesis is somehow on the right track. 

\subsubsection{Swampland conditions for $\cn=2$ SUGRA}

This paragraph is a comment about reference \cite{swampIII}.
As we know, on SUSY grounds, the vector-multiplet couplings of
a 4d $\cn=2$ SUGRA are dictated by special K\"ahler geometry
which is equivalent to a variation of Hodge structure of weight 3
with $h^{3,0}=1$. In \cite{swampIII} it was observed that a deep 
problem in math is to determine which VHS arise from geometry,
i.e.\! describes an actual family of Calabi-Yau manifolds. 
A condition mathematicians have proven to be necessary is that
the period map $p$ of the VHS satisfies the structure theorem
\cite{reva,revb,MT4,periods}.
It was proposed in \cite{swampIII} that this same condition is also
a swampland criterion for 4d $\cn=2$ effective theories.
Then  the question was the logic relation between this 
new criterion and the Ooguri-Vafa geometric swampland conjectures
\cite{OoV}. The fact that the validity of the structure theorem of \cite{MT4}
implies the swampland conjecture is easy to see \cite{swampIII}.
It was initially believed that the structure theorem was a {stronger}
requirement than the Ooguri-Vafa ones: the naive feeling was that the structure
theorem is a stringent condition with lots of Number Theoretical and Algebro-Geometric
aspects whereas the OV statements looked like simple qualitative properties of
the relevant geometries. However, now we see that the OV properties are strong enough
(when supplemented by the conditions following from $\cn=2$ supersymmetry and
some mild regularity assumption) to actually imply the structure theorem of 
\cite{MT4}, so that the two set of conditions are essentially equivalent.
%

\subsection{First applications}

In the context of 4d $\cn=2$ supergravity, the Hodge-theoretic structure theorem -- which holds only in a tiny subset of the space of all formal $\cn=2$
\textsc{sugra} which includes all the ones arising from string theory -- has a lot of interesting implications \cite{swampIII}, which include the completeness of
instanton corrections \cite{vvvaf}.   

If our working hypothesis is correct, a structure theorem of exactly the same form holds for all consistent effective theories of quantum gravity.
However the powerful results in the $\cn=2$ situation arise from the interplay between two pieces of information: the structure theorem and special K\"ahler geometry.
In the non-SUSY case, where $\mathscr{M}$ has generic holonomy, the second ingredient is lacking,
and we are able to extract from the structure theorem much weaker physical consequences -- which however have the merit of being (conjecturally) true in full generality.
\medskip

\subparagraph{The $\cn=2$ case.} 
Let us briefly recall the situation in the $\cn=2$ context.\footnote{\ Our summary below is rather rough and the statements are not meant to be technically precise; see \cite{swampIII} for a more precise discussion.}  In this case the gauge coupling $\mu$, seen as a fibration over its image
$\mathscr{B}$,
\be
\mu\colon \mathscr{M}\to \mathscr{B}\equiv \mu(\mathscr{M})\subset \boldsymbol{\Gamma}\backslash Sp(2h,\R)/U(h)
\ee
is essentially trivial in the sense that $\mathscr{M}=\mathscr{M}_\text{hyper}\times \mathscr{M}_\text{vector}$, 
and the Griffiths' infinitesimal period relations \cite{Gbook,periods,reva,revb}, together with the Torelli theorem \cite{torelli1,torelli2,torelli3}, say
that the period map $p$ is a Griffiths-horizontal, holomorphic embedding of the universal cover of $\mathscr{M}_\text{vector}$
into the Griffiths period domain $\boldsymbol{D}_h$  (cfr.\! eqn.\eqref{perdomains})
\be
p\colon \widetilde{\mathscr{M}}_\text{vector}\to \boldsymbol{D}_h\overset{\rm def}{=}Sp(2h,\R)\big/[U(1)\times U(h-1)].
\ee
Composing with the canonical projection $\boldsymbol{D}_h\twoheadrightarrow\boldsymbol{H}_h$,
we see that the non-holomorphic smooth map $\widetilde{\mathscr{M}}_\text{vector}\to \widetilde{\mathscr{B}}\equiv \widetilde{\mu}(\widetilde{\mathscr{M}})$
 is a local isomorphism by horizontality of $\tilde p$. The K\"ahler form on $\widetilde{\mathscr{M}}_\text{vector}$ is pull back
 $\tilde p^*F$ of the curvature 2-form $F$ of the Hodge line bundle $\boldsymbol{\cl}\to\boldsymbol{D}_h$ (i.e.\! the homogeneous bundle over $\boldsymbol{D}_h$ defined by the fundamental character of the $U(1)$ factor in $H=U(1)\times U(h-1)$). Then in the $\cn=2$ case
  there is a simple relation between the gauge couplings $\tau(\phi)_{ab}$ and the scalar metric $G(\phi)_{ij}$ expressed by the moduli space Einstein equation \eqref{EE}.
   In particular, all isometries of $\widetilde{\mathscr{M}}_\text{vector}\hookrightarrow \boldsymbol{D}_h$ is the restriction of an $Sp(2h,\R)$ symmetry of the ambient space
  $\boldsymbol{D}_h$.

\subparagraph{The general case.} 
In the case of a general quantum-consistent effective theory -- with NO supersymmetry --  we do not expect a simple relation between
the couplings
$\tau(\phi)_{ab}$ and $G(\phi)_{ij}$.
However, to the extend that $\mu$ is harmonic, quantum consistency still implies subtle relations between the two couplings. In particular the scalars' metric $G(\phi)_{ij}$ is constrained by the condition that
the gauge coupling $\mu$ is harmonic for $G(\phi)_{ij}$. 
This severely restricts the allowed scalar metric.
E.g., when
 $\mathscr{M}$ is a complex OV manifold 
and $\mu$ is pluri-harmonic, it  requires $G(\phi)_{ij}$ to be K\"ahler. The scalars metric satisfies also other strong constraints: \textit{(i)}
the infinite group $\cg$ acts by isometries on $G(\phi)_{ij}$, \textit{(ii)} the volume is finite,
and \textit{(iii)} Ricci curvature satisfies the required bounds. Thus, even if the effective theory has no \textsc{susy},
 for a given gauge coupling $\mu$ there is not that much freedom in the choice of the scalars' metric $G(\phi)_{ij}$ if we wish to avoid ending in the swampland.

\medskip

Unfortunately, for a non-SUSY theory the relation between $\mu$ and the consistent scalar's metric $G(\phi)_{ij}$ is rather implicit.
For this reason, in absence of \textsc{susy} it is hard to rephrase the structure theorem for $\mu$ in terms of
geometric proprieties of $G(\phi)_{ij}$.

The structure theorem refers to properties of $\mathscr{B}\equiv \mu(\mathscr{M})$, seen as a submanifold of
the Siegel variety $Sp(2h,\Z)\backslash \boldsymbol{H}_h$, rather than directly to the intrinsic geometry of  $\mathscr{M}$. In the $\cn=2$ case $\mathscr{B}\cong\mathscr{M}_\text{vector}$ so this is not a limitation,
but in general the two spaces are quite different. 
Anyhow the $\cn=2$ statements of  \cite{swampIII} hold for general domestic geometries when referred to $\mathscr{B}$. 
In particular we have the \textit{dycothomy:}
\begin{itemize}
\item[(a)] either $\widetilde{\mathscr{B}}\subset\boldsymbol{H}_h$ is a totally geodesic submanifold, hence symmetric, and the fibers $\tilde\mu^{-1}(b)\subset \widetilde{\mathscr{M}}$ are minimal submanifolds;
\item[(b)] 
or no continuous symmetry of the ambient space $\boldsymbol{H}_h$ leaves $\widetilde{\mathscr{B}}$ fixed (as a set).
\end{itemize}
Possibility (a) corresponds to very special effective theories which look like consistent truncations of 
some $\cn\geq 3$ supegravity. The second possibility is the generic case. 
In the $\cn=2$ case this implies completeness of instanton corrections (which is expected on physical grounds \cite{vvvaf}), and this implication is likely to extend to more general situations.  

\appendix

\section{SUGRA spaces: tamed maps vs.\! special holonomy}\label{ttamed}

We want to show that if $X$ is a symmetric space
relevant for 4d SUGRA not of the form $SO(m,1)/SO(m)$ or $SU(m,1)/U(m)$
all tamed maps $f\colon X\to Y$ are totally geodesic.

We consider the symmetric Riemannian manifolds of type III, i.e. of the form $G/K$ where $G$ is some
real Lie algebra and $K$ its maximal compact subgroup. By general theory its holonomy Lie algebra is $\mathfrak{k}\equiv\mathfrak{Lie}(K)$. The space $SO(m,1)/SO(m)$ has dimension $m$ and strictly generic holonomy algebra $\mathfrak{so}(m)$, so 
has no non-trivial parallel forms, and hence in this case \emph{tamed} $\equiv$ \emph{harmonic}. The complex hyperbolic space $SU(m,1)/U(m)$ has complex dimension $m$
and holonomy Lie algebra $\mathfrak{su}(m)$, so it is a strict K\"ahler manifold and hence
for the complex hyperbolic spaces \emph{tamed} $\equiv$ \emph{pluri-harmonic}.
We shall call $SO(m,1)/SO(m)$ and $SU(m,1)/U(m)$ the \emph{strict} cases.
For all other type III symmetric spaces the holonomy algebra $\mathfrak{hol}(G/K)$ is neither generic nor strict K\"ahler, so for these spaces \emph{tamed} is strictly stronger than being harmonic or pluri-harmonic. We consider the cosets $G/K$ relevant for 4d SUGRA.

We write  $T\cong T^\vee$ for the (irreducible) holonomy representation of the symmetric space $G/K$.
The corresponding Lie algebras decomposes as $\mathfrak{g}= \mathfrak{k}\oplus T$ and the
holonomy representation on $T$ is induced by the adjoint action of $\mathfrak{g}$ on itself.
$G/K$ has a non-trivial algebra $\cp^\bullet$ of parallel forms $\Omega^{(s)}\in \wedge^{k_s}T$. We consider their 
annihilator algebra
\be
\mathfrak{a}\overset{\rm def}{=}\Big\{ a_{ij} \in \otimes^2 T\cong \mathrm{End}(T) \colon a_{[i_1 j}\,\Omega^{(s)}_{j i_2\cdots i_{k_s}]}=0\ \ 
\forall\; \Omega^{(s)}\in \cp^\bullet\Big\}
\ee
A map $f$ is tamed iff $D_i \partial_j f$ is contained in $\mathfrak{a}\cap \odot^2 T$; when this space is zero and $f$ is tamed 
we must have $D_i \partial_j f=0$, that is,
\be
\mathfrak{a}\cap \odot^2 T=0\quad\Longrightarrow\quad \text{all tamed maps are totally geodesic.
}\label{jjja1z}
\ee

$\mathfrak{a}\subset \otimes^2 T$, is a real Lie subalgebra of $\mathfrak{sl}(T)$, contains $\mathfrak{k}$ and is a $\mathfrak{k}$-invariant subspace; hence it has the form
$\mathfrak{a}= \mathfrak{k}\oplus \mathfrak{b}$ with $\mathfrak{b}\subset (\otimes^2T)_\text{traceless}$.
The algorithm goes through the following steps.
\textbf{(1)} we show that $\mathfrak{a}\cap \wedge^2T=\mathfrak{k}$ while all irreducible $\mathfrak{k}$-representations
in $(\odot^2T)_\text{traceless}$ are self-dual, so we infer that $\mathfrak{a}$ is a \emph{reductive} Lie algebra with
maximal compact subalgebra $\mathfrak{k}$. Writing $A$, $K$ for the corresponding group, $A/K$ is a, possibly trivial  or
reducible, symmetric space. \textbf{(2)} One checks in the Cartan table of symmetric space which groups
$A$, $K$ are allowed and reads from them the candidate $\mathfrak{b}$. \textbf{(3)} Finally one checks
that the candidate $\mathfrak{b}\not\subset \odot^2 T$, getting a paradox. \textbf{(4)} We conclude that  
$\mathfrak{a}\cap \odot^2 T=0$ and apply \eqref{jjja1z}.
\medskip

We run the algorithm one space at the time.
\medskip

$\bullet$ $\cn=8$ \textsc{sugra}. The scalars' space is
 $E_{7(7)}/SU(8)$. $K=SU(8)$ and $T$ is the $\boldsymbol{70}$ i.e.\! $T=\wedge^4 F$ ($F$ stands for  the fundamental of $SU(8)$). $\wedge^6 T$ contains a singlet, i.e. on the symmetric space 
 \be
 E_{7(7)}/SU(8)
 \ee
  we have a non-trivial parallel 6-form
 and $\mathfrak{so}(T)\not\subset \mathfrak{a}$. On the other hand, as $\mathfrak{su}(8)$-modules
\be
 \mathfrak{so}(T)= \wedge^2(\wedge^4F)= \mathfrak{su}(8)\oplus\boldsymbol{2352}
\ee
 so $\mathfrak{su}(8)$ is the maximal compact subalgebra of $\mathfrak{a}$. On the other hand
\be
\odot^2 T=\odot^2(\wedge^4 F)= \boldsymbol{720}\oplus \boldsymbol{1764} 
\ee
 and both representations are self-dual. Hence $\mathfrak{a}$ is semi-simple with maximal compact subalgebra
 $\mathfrak{su}(8)$ and $A/K$ is a non-compact symmetric space. Since $E_{7(7)}/SU(8)$ is the only non-trivial non-compact symmetric space of holonomy $SU(8)$, we must have either $\mathfrak{a}\cap\odot^2 T$ equal zero or $T$.
 But $T\not\subset \odot^2 T$ and the second possibility is ruled out. 
 \medskip

$\bullet$ $\cn=6$ \textsc{sugra}. The scalars' space $SO^*(12)/U(6)$ is K\"ahler; $T=(\wedge^2 F\oplus \wedge^2 \overline{F})_\R$ ($F$ is the fundamental of $U(6)$). Then
\be
\wedge^2 T= (\wedge^2 F)\otimes (\wedge^2 \bar F)\oplus \big(\wedge^2(\wedge^2 F)\oplus \wedge^2 (\wedge^2 \bar F)\big)_\R= \boldsymbol{1}\oplus\mathfrak{su}(6)\oplus \boldsymbol{189}\oplus (\boldsymbol{105}\oplus\overline{\boldsymbol{105}})_\R 
\ee
where $\boldsymbol{1}$ is the K\"ahler form. $\wedge^3(\wedge^2 F)\otimes \wedge^3(\wedge^2\bar F)$
contains 2 singlets, so that we have a parallel (3,3) form different from the cube of the K\"ahler form and
$\mathfrak{u}(15)\not\subset \mathfrak{a}$. On the other and 
\be
\mathfrak{u}(15)= (\wedge^2 F)\otimes (\wedge^2\bar F)=\mathfrak{u}(6)\oplus\boldsymbol{189}
\ee
so the maximal compact subalgebra of $\mathfrak{a}$ is $\mathfrak{u}(6)$.
Now we have 3 non-trivial symmetric spaces to consider namely $SU(6,1)/U(6)$, $Sp(12,\R)/U(6)$ and $SO^*(12)/U(6)$
 with would-be $\mathfrak{b}$, $(F\oplus \bar F)_\R$, $(\odot^2 F\oplus \odot^2\bar F)_\R$ and $(\wedge^2 F\oplus \wedge^2 \bar F)_\R$, respectively. Since $\odot^2 T$ contains only the $U(1)$ characters $4,0,-4$ we get a contradiction in all cases.
 We conclude that $\mathfrak{b}=0$, so $\mathfrak{a}\cap \odot^2 T=0$. 
 \medskip 
 
$\bullet$ $\cn=5$ \textsc{sugra} has $G/K=SU(5,1)/U(5)$ which is a strict case.
\medskip
 
$\bullet$ $\cn=4$ \textsc{sugra}. The (universal cover of) the scalars' space is reducible 
\be
SU(1,1)/U(1)\times SO(6,k)/[SO(6)\times SO(k)],
\ee
 the first factor is strict, as it is the second one when $k=1$. We focus on the second factor and assume $k\geq 2$ then $T=V_6\otimes V_k$, where $V_k$ is the vector of $SO(k)$. We have a parallel 6-form and dually a parallel $6(k-1)$ form, hence $\mathfrak{so}(6k)\not\subset\mathfrak{a}$.
 Since
\be
\wedge^2 T= \mathfrak{so}(6)\otimes (\boldsymbol{1}\oplus (\odot^2 V_k)_\text{traceless})\oplus (\boldsymbol{1}\oplus \boldsymbol{20})\otimes \mathfrak{so}(k)
\ee
we have that $\mathfrak{so}(6)\oplus\mathfrak{so}(k)$ is the maximal compact subalgebra of $\mathfrak{a}$.
We have two possible non-trivial symmetric spaces 
\be
SL(6,\R)/SO(6)\times SL(k,\R)/SO(k)\quad\text{and}\quad
SO(6,k)/[SO(6)\times SO(k)]
\ee with would-be $\mathfrak{b}\subset (\odot^2 V_6)_\text{traceless}\oplus
(\odot^2 V_k)_\text{traceless}$ and $\mathfrak{b}\subset V_6\otimes V_k$, respectively. The first one obviously does not preserve the parallel forms,
 and the second one is not contained in $\odot^2 T$. This rules out also $SL(6,\R)/SO(6)$ and $SL(k,\R)/SO(k)$
and one remains with $\mathfrak{a}=\mathfrak{so}(6)\oplus\mathfrak{so}(k)$.
\medskip

$\bullet$ $\cn=3$ \textsc{sugra}. $SU(3,k)/[SU(3)\times U(k)]$ again is K\"ahler and for $k=1$ strict. $T=(F_3\otimes F_k\oplus \bar F_3\otimes \bar F_k)_\R$. We have parallel $(3,3)$ and $(k,k)$ forms, the maximal compact subalgebra is
$\mathfrak{su}(3)\oplus\mathfrak{u}(k)$; going through the various symmetric spaces, one concludes that $\mathfrak{b}=0$.

\section{No boundary term in the Bochner argument}\label{tamhol}

As explained at the end of \S.\ref{s:stru1}, we have to show that if $\mathscr{M}$ is a OV manifold
and $f\colon\mathscr{M}\to \Lambda\backslash G/K$ is a \emph{finite-energy} harmonic map to a locally symmetric space of non-compact type, 
  the surface term 
\be\label{sjowzero}
\int_\mathscr{M} 
d \Big(g_{ij}\,\ast (d f^i \wedge \Omega)\wedge D\ast (d f^j\wedge \ast\mspace{3mu} \Omega)\Big)
\ee
vanishes (cfr.\! eqn.\eqref{weitza}).
We proceed by adapting the argument in \cite{corlette2}.
As in \S.\ref{s:cusps}, for all $R>0$ we write
 $h_R\colon \mathscr{M}\to \R$ for a Lipschitz continuous function
 such that for some fixed constant $k > 0$ \cite{ricciflat}: 
\be
0\leq h_R\leq 1,\qquad h_R=\begin{cases} 1 &\text{for }r\leq R\\
0 &\text{for }r\geq 2R,
\end{cases}\qquad \big|dh_R\big|<\frac{k}{R},
\ee
and assume \textbf{Condition $\boldsymbol{(\ast)}$} 
i.e.\! eqn.\eqref{xxxqwrt56}
\be
|\Delta h_R|<C.
\ee
We then proceed as in reference \cite{corlette2}:
\be
\left|\int_\mathscr{M} h_R\, \Delta |df|^2\, d\mspace{2mu}\mathsf{vol}\right|=\left| \int_\mathscr{M} (\Delta h_R)\,|df|^2\, d\mspace{2mu}\mathsf{vol}\right|\leq \int_\mathscr{M} |\Delta h_R|\,|df|^2\,d\mspace{2mu}\mathsf{vol}\leq C\,E(f)
\ee
and taking $R\to\infty$
\be
\left|\int_\mathscr{M} \Delta |df|^2\,d\mspace{2mu}\mathsf{vol}\right|\leq C\, E(f).
\ee
Now let $f\colon \mathscr{M}\to \Lambda\backslash G/K$ be harmonic of finite energy, $E(f)<\infty$. Since the target space has non-positive sectional curvatures, and the Ricci tensor of $\mathscr{M}$ is bounded below (cfr.\! eqn.\eqref{eeeeecq}),
\be
R_{ij}\geq -K\,g_{ij},
\ee
the Bochner formula of Eells and Sampson for harmonic maps \cite{eells}\footnote{\ $h_{ab}$ and $R_{abcd}^h$ are, respectively, the metric and the Riemann tensor of the target space.}
\be
\frac{1}{2}\Delta|df|^2 =|\nabla d f|^2+R_{ij}\, h_{ab}\,\partial^i f^a\, \partial^j f^b-R_{abcd}^h\, g^{ik} g^{jl}\, \partial_i f^a\, \partial_j f^b\,
\partial_k f^c\,\partial_l f^d
\ee
gives a bound of the form
\be
|D_i\partial_j f|^2\leq  \tfrac{1}{2}\,\Delta|df|^2+K |df|^2\quad\Rightarrow\quad \int_\mathscr{M} |D_i\partial_j f|^2 \,d\mspace{2mu}\mathsf{vol}\leq \big(\tfrac{1}{2}\,C+K\big)E(f).
\ee
This bound implies that both $(df \wedge \Omega^{(s)})$ and $D\ast(df\wedge \ast\mspace{4mu}\Omega^{(s)})$ have finite $L^2$ norms.
Now the boundary term that we have to show to vanish, \eqref{sjowzero},
 is the limit as $R\to\infty$ of
\be
\int_\mathscr{M} h_R\; d\langle \ast(df\wedge \Omega), D\ast(df\wedge \ast\mspace{4mu}\Omega)\rangle =
-\int_\mathscr{M} dh_R\wedge  \langle \ast(df\wedge \Omega), D\ast(df\wedge \ast\mspace{4mu}\Omega)\rangle
\ee
while $|dh_R|\leq k/R$ so that
\be
\left|\int_\mathscr{M} dh_R\wedge  \langle \ast(df\wedge \Omega), D\ast(df\wedge \ast\mspace{4mu}\Omega)\rangle\right| \leq \frac{k}{R}\; \Big\|\,df\wedge \Omega\,\Big\|_{L^2}\!\cdot \Big\|D\ast(df\wedge \ast\mspace{4mu}\Omega)\Big\|_{L^2}
\ee
which goes to zero as $R\to\infty$.

\end{document}